\documentclass[lettersize,journal]{IEEEtran}
\usepackage{graphicx}
\usepackage{latexsym}
\usepackage{tabularx}
\usepackage{cite}
\usepackage[labelformat=simple]{subcaption}

\usepackage{caption} 
\usepackage{color}
\usepackage{wrapfig}
\usepackage{amsfonts,amssymb}
\usepackage{amsthm,amsmath}
\usepackage{url}
\usepackage[encapsulated]{CJK} 
\usepackage{indentfirst}
\usepackage{soul}
\usepackage{comment}
\usepackage{setspace}
\usepackage{balance}
\usepackage{array}
\usepackage{orcidlink}
\usepackage{hyperref}
\hypersetup{colorlinks,
      linkcolor=blue,
      citecolor=blue,
      urlcolor=black,
      filecolor=blue}

\usepackage{pifont}

\usepackage{makecell}

\newcommand{\noindentparagraph}[1]{\textbf{\textit{#1}}}

\usepackage{xcolor}
\newcommand{\chengwei}[1]{{\color{blue}#1}}

\definecolor{dv}{named}{black}

\makeatletter
\newcommand\footnoteref[1]{\protected@xdef\@thefnmark{\ref{#1}}\@footnotemark}
\makeatother

\usepackage{enumitem}

\usepackage{algorithmicx}
\usepackage{algorithm}
\usepackage[noend]{algpseudocode}
\usepackage{makecell}
\usepackage{booktabs}
\usepackage{multirow}
\usepackage{adjustbox}

\def\BibTeX{{\rm B\kern-.05em{\sc i\kern-.025em b}\kern-.08em
    T\kern-.1667em\lower.7ex\hbox{E}\kern-.125emX}}



\makeatletter

\newcommand{\algmargin}{\the\ALG@thistlm}
\newlength{\whilewidth}
\algnewcommand{\parState}[1]{\State%
  \parbox[t]{\dimexpr\linewidth-\algmargin}{\strut #1\strut}}

\algrenewtext{If}[1]
    {\parbox[t]{\dimexpr\linewidth-\algmargin}
    {\hangindent\whilewidth\algorithmicif\ #1\ \algorithmicthen}}

\algrenewtext{ElsIf}[1]
    {\parbox[t]{\dimexpr\linewidth-\algmargin}
    {\hangindent\whilewidth\algorithmicelsif\ #1\ \algorithmicthen}}

\algnewcommand\algorithmicinput{\textbf{Input:}}
\algnewcommand\INPUT{\item[\algorithmicinput]}

\algnewcommand\algorithmicoutput{\textbf{Output:}}
\algnewcommand\OUTPUT{\item[\algorithmicoutput]}

    \newcounter{phase}[algorithm]
    \newlength{\phaserulewidth}
    \newcommand{\setphaserulewidth}{\setlength{\phaserulewidth}}
    
    \setphaserulewidth{.7pt}
\makeatother

\usepackage{thmtools}

\theoremstyle{definition}
\newtheorem{defi}{Definition}
\theoremstyle{plain}
\newtheorem{remark}{Remark}
\newtheorem{lemma}{Lemma}
\newtheorem{corollary}{Corollary}
\newtheorem{theorem}{Theorem}
\newtheorem{assumption}{Assumption}

\DeclareMathOperator*{\argmax}{arg\,max}
\DeclareMathOperator*{\argmin}{arg\,min}

\graphicspath{ {images/} }

\usepackage{etoolbox}

\renewcommand{\baselinestretch}{0.985}
\normalsize
\allowdisplaybreaks

\begin{document}

\title{Totoro$^+$: An Adaptive and Scalable Edge Federated Learning System}

\author{
        Cheng-Wei~Ching\,\orcidlink{0000-0001-6621-4907},
        Xin~Chen\,\orcidlink{0009-0008-8188-8211},
        Taehwan Kim\,\orcidlink{https://orcid.org/0009-0008-4378-1691},
        Jian-Jhih~Kuo\,\orcidlink{0000-0002-1051-5089},  
        Dilma Da~Silva\,\orcidlink{0000-0001-6538-2888}, and 
        Liting~Hu\,\orcidlink{0009-0007-7222-5507}
\vspace{-0.1in}
        
\IEEEcompsocitemizethanks{
    \IEEEcompsocthanksitem This work was supported in part by the National Science Foundation under Grants
    NSF-OAC-2313738, NSF-CAREER-2313737, and NSF-CNS-2322919, and in part by the National Science and Technology Council under Grants 113-2221-E-194-040-MY3, and 114-2628-E-194-002-MY3.
    A preliminary version of this paper appeared at the Proceedings of the Nineteenth European Conference on Computer Systems (EuroSys’24)~\cite{totoro}. \textit{(Corresponding author: Liting Hu.)}	
	\IEEEcompsocthanksitem Cheng-Wei Ching and Liting Hu are with the Department of Computer Science and Engineering, University of California, Santa Cruz, Santa Cruz, CA 95064 USA (e-mail: cching1@ucsc.edu; liting@ucsc.edu). 
    \IEEEcompsocthanksitem Xin Chen is with the Department of Computer Science and Engineering, Georgia Institute of Technology, Atlanta, GA 30332 USA (e-mail: xchen384@gatech.edu).
    \IEEEcompsocthanksitem Taehwan Kim is with the Department of Computer Science, Virginia Tech, Blacksburg, VA, 24061 (e-mail: tk020@vt.edu).
    \IEEEcompsocthanksitem Jian-Jhih Kuo is with the Department of Computer Science and Information Engineering, National Chung Cheng University, Taiwan, and also with the Advanced Institute of Manufacturing with High-tech Innovations, National Chung Cheng University, Taiwan (e-mail: lajacky@cs.ccu.edu.tw).
    \IEEEcompsocthanksitem Dilma Da Silva is with Texas A\&M University, College Station, TX 77843 USA (e-mail: dilma@cse.tamu.edu).

}
}
 

\IEEEtitleabstractindextext{
\begin{abstract}
Federated Learning (FL) is an emerging distributed machine learning (ML) technique that enables in-situ model training and inference on decentralized edge devices. We propose Totoro$^+$, a novel scalable FL system that enables massive FL applications to run simultaneously on edge networks. The key insight is to explore a distributed hash table (DHT)-based peer-to-peer (P2P) model to re-architect the centralized FL system design into a fully decentralized one. In contrast to previous studies where many FL applications shared one centralized parameter server, Totoro$^+$ assigns a dedicated parameter server to each application. Any edge node can act as any application’s coordinator, aggregator, client selector, worker (participant device), or any combination of the above, thereby radically improving scalability and adaptivity. Totoro$^+$ introduces three innovations to realize its design: a locality-aware P2P multi-ring structure, a publish/subscribe-based forest abstraction, and a game-theoretic path planning model with a guarantee of an $\epsilon$-approximate Nash equilibrium. Real-world experiments on 500 Amazon EC2 servers show that Totoro$^+$ scales gracefully with the number of FL applications and $N$ edge nodes speeds up the total training time by $1.2\times-14.0\times$, achieves $\mathcal{O}(\log N)$ hops for model dissemination and gradient aggregation with millions of nodes, and efficiently adapts to the practical edge networks and churns.
\end{abstract}

\begin{IEEEkeywords}
Distributed and parallel systems for machine learning, federated learning, game theory, and edge computing.
\end{IEEEkeywords}

}

\maketitle

\IEEEdisplaynontitleabstractindextext

\IEEEpeerreviewmaketitle

\section{Introduction}

\IEEEPARstart{W}{ith} the rise of 5G and the growth of connected devices, federated learning (FL) enables local data processing and machine learning at the network edge. This approach reduces data transmission and enhances privacy by eliminating the need for centralized servers. FL has been applied in various domains, including human activity prediction~\cite{asynchronous_online_fl, communication_efficient_fl}, sentiment analysis~\cite{mocha}, language processing~\cite{fl_mobile_keyboard, applied_federated_learning}, and enterprise systems~\cite{ibm_fl_framework}.

\begin{figure}[t]
  \centering
  \includegraphics[clip,width=1.\columnwidth]{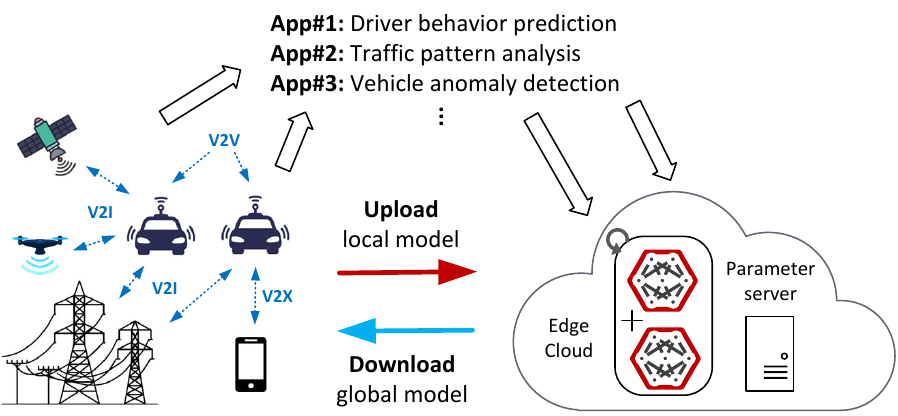}
  \caption{An edge federated learning use case based on autonomous vehicles.}
    \vspace{-0.05in}
  \label{fig:example}
\end{figure}


\noindentparagraph{The problem.} We consider a typical edge computing architecture where millions of edge devices (e.g., smart wearables, autonomous vehicle sensors) connect to the cloud through an edge layer. This layer comprises hundreds of thousands of server-grade machines, gateways, and routers—termed ``edge nodes'', managed by various providers across geographically distributed sites. Raw data is collected and stored on these edge nodes, enabling machine learning models to be trained locally without transferring raw data to a central server.


Federated learning (FL) tools and frameworks—such as TensorFlow Federated~\cite{tensorflow_federated}, LEAF~\cite{femnist_dataset}, FLOAT~\cite{float}, REFL~\cite{refl}, and PySyft~\cite{pysyft}—have gained traction. However, building an efficient FL system on practical edge networks remains challenging due to dynamic edge environments and the high scalability demands of emerging FL applications.


Figure~\ref{fig:example} illustrates a use-case scenario that highlights these challenges. In future intelligent transportation systems such as the efforts currently funded by the US Department of Transportation~\cite{transportation_research_plan}, autonomous vehicles are interconnected and equipped with wide-area network access
They continuously collect sensor and behavioral data such as speed, engine performance, driving patterns, and environmental conditions. On the back end, numerous FL applications will run concurrently on these vehicles or edge nodes, performing in-situ training using the collected data.
Examples of such FL applications include driver behavior prediction, which trains Long-Short Term Memory (LSTM) networks~\cite{lstm} to anticipate lane change actions~\cite{du2023driver}; traffic pattern analysis, which learns traffic patterns to redirect optimal routes~\cite{wang2023fedstream}; and anomaly detection, which employs federated learning to identify and respond to vehicle malfunctions or hazardous road conditions in real time~\cite{wang2022ai}.


\noindentparagraph{Our solution: Totoro$^+$.} We present Totoro$^+$, an adaptive and scalable federated learning (FL) system designed to support large-scale concurrently running FL applications on edge networks and adapt to edge dynamics. Rather than modifying existing FL systems, we adopt a fully decentralized architecture. Unlike prior approaches where multiple FL applications rely on shared centralized parameter servers, Totoro$^+$ assigns each application its own parameter server, preventing overload on any single edge node. Moreover, any edge node can dynamically serve as an application's coordinator, aggregator, client selector, worker, or any combination of these roles, significantly enhancing scalability and adaptivity.


The key idea is to leverage a distributed hash table (DHT)-based peer-to-peer (P2P) architecture to redesign FL systems. P2P models, widely used in file sharing (e.g., BitTorrent~\cite{bittorrent}, Storj~\cite{storj}, Freenet~\cite{freenet}), peer-assisted CDNs~\cite{netsession}, and blockchains~\cite{blockchain_paper}, treat all nodes equally, without a central server, as they collaborate to perform tasks or deliver services. For instance, in BitTorrent~\cite{bittorrent}, users both download and upload file pieces to/from peers. Inspired by this model, we enable all edge nodes to jointly handle FL training and testing for a large number of applications in parallel.


Totoro$^+$ introduces three innovations to realize its fully decentralized design: \emph{a locality-aware P2P multi-ring structure}, \emph{a publish/subscribe-based forest abstraction}, and \emph{a game-theoretic path planning model}. 

First, edge nodes self-organize into a scalable DHT-based P2P overlay. The \textit{locality-aware P2P multi-ring structure} enforces administrative boundaries between edge sites, enabling efficient local FL execution and preserving site-specific control flows.
Second, Totoro$^+$ builds \textit{a publish/subscribe-based forest abstraction} on top of the locality-aware P2P multi-ring structure for model broadcasting and gradient aggregation. This decouples the centralized FL architecture into a ``fully'' decentralized architecture, where each FL application operates on an independent dataflow tree, significantly enhancing system scalability. Besides individual dataflow trees for FL applications, we introduce an advertise-discover (AD) tree that enables edge nodes to publish and locate FL applications within the overlay. 
Third, to address edge dynamics (e.g., link failures, stragglers, and bandwidth constraints over edge networks) Totoro$^+$ employs \textit{a game-theoretic path planning model} that dynamically replans data transfer paths, ensuring robust and adaptive model broadcast and gradient aggregation.


\noindentparagraph{Summary of results.} We implement Totoro$^+$ using the open-source Pastry DHT~\cite{pastrycode} and PyTorch~\cite{pytorch} frameworks,%
~\footnote{\url{https://github.com/UCSC-ELVES-Lab/Totoro-EuroSys-2024}}
and evaluate it on various FL tasks with real-world datasets across 500 Amazon EC2 servers. Compared to state-of-the-art FL systems~\cite{openfl,fedscale}, Totoro$^+$ achieves superior scalability and load balancing, efficiently distributing masters across large-scale edge networks without overloading individual nodes. It accelerates total training time by $1.2\times$ to $14.0\times$, supports $\mathcal{O}(\log N)$ hops for model dissemination and gradient aggregation at million-node scales, and robustly handles network churn and unreliable edge conditions.


\begin{table}[t!]
\scriptsize
\centering
\begin{tabular}{p{2.5cm}|p{2.5cm}|p{2.5cm}}
    \Xhline{1.5\arrayrulewidth}
    \noalign{\vskip 2pt}
    \rule{0pt}{.8\normalbaselineskip}
    \textbf{Features} 
    & \textbf{Totoro~\cite{totoro}} 
    & \textbf{Totoro$^+$ (this work)} \\
    \noalign{\vskip 2pt}\hline
    \noalign{\vskip 2pt}

    \rule{0pt}{.8\normalbaselineskip}
    Application advertisement and discovery
    & Join-by-AppId
    & Advertise-discover tree  \\

    \rule{0pt}{.8\normalbaselineskip}
    Path planning model
    & Bandit model without traffic congestion
    & Game-theoretic model with traffic congestion \\

    \rule{0pt}{.8\normalbaselineskip}
    Algorithmic complexity
    & $\mathcal{O}(\log N\cdot I_{KL})$
    & $\mathcal{O}(\log N\cdot Matmul)$ \\

    \rule{0pt}{.8\normalbaselineskip}
    Theoretical guarantee
    & Heuristic
    & $\epsilon$-approximate Nash equilibrium \\

    \rule{0pt}{.8\normalbaselineskip}
    Failure recovery
    & Worker 
    & Worker and master \\

    \noalign{\vskip 2pt}
    \Xhline{1.5\arrayrulewidth}
\end{tabular}
\caption{Comparison between Totoro and Totoro$^+$, where $I_{KL}$ is the iteration complexity of solving the Kullback–Leibler divergence-based convex feasibility, and $Matmul$ is matrix multiplication.}
\vspace{-0.05in}
\label{tab:diff_summary}
\end{table}

\noindentparagraph{Contributions.} This paper makes the following contributions :
\begin{itemize}[leftmargin=*, topsep=0.5ex]
\item \emph{Problem:} We analyze the software architecture of current FL systems and identify key challenges in applying FL to practical edge networks. 
\item \emph{Key idea:} Totoro$^+$ explores the DHT-based P2P model to propose a novel, fully decentralized ``many masters/many workers'' architecture that significantly enhances scalability and adaptability. 
\item \emph{Totoro$^+$ design \& implementation:} Totoro$^+$ incorporates three key innovations: a locality-aware P2P multi-ring structure, a publish/subscribe-based forest abstraction, and a game-theoretic path planning model with theoretical guarantees of an $\epsilon$-approximate Nash equilibrium.
\item \emph{Results:} Our evaluation demonstrates Totoro$^+$'s substantial improvements in scalability and adaptability over state-of-the-art FL systems.
\end{itemize}

A preliminary conference paper version~\cite{totoro} has been published in Proceedings of the Nineteenth European Conference on Computer Systems (EuroSys’24). 
{\color{dv}
The summary of changes is presented in Table~\ref{tab:diff_summary}.
}
In this journal version, we provide the following three significant improvements: 
First, we introduce the advertise-discover (AD) tree to the publish/subscribe-based forest abstraction. The AD tree enables edge nodes to advertise and discover FL applications running over the DHT-based P2P overlay for FL application searching and crowdsourcing~\cite{fl_framework_crowdsourcing,fl_collusion_crowdsourcing} {\color{dv}
(see Appendix~\ref{app:advertise_discover} for a detailed comparison with the conference version).
}

Secondly, we replace the bandit-based path planning model with the game-theoretic path planning model, which formalizes the path planning problem as a congestion game with bandit feedback and introduces a game-theoretic distributed hop-by-hop algorithm with the guarantee of an $\epsilon$-approximate Nash equilibrium  
{\color{dv}
(see Appendix~\ref{app:math_model_comparison} for a detailed comparison with the conference version).
}

Thirdly, we provide respective failure recovery mechanisms for master and worker node failures. When a worker node fails, its child node sends a \texttt{JOIN} message using AppId as the key to find a new parent node to recover the dataflow tree. The master node in each communication round replicates the training state across $k$ nodes in its neighborhood set. When the master node fails, its immediate node sends a \texttt{JOIN} message using AppId as the key to find an alternative master node. The new master node retrieves a state replica to proceed with the training process. 

Finally, we demonstrate that the game-theoretic path planning model is a better solution for adapting to edge dynamics with experiments on packet latency, and Nash regret. Also, we provide more extensive experiment results to demonstrate further Totoro$^+$'s scalability to massive FL applications, fault tolerance to simultaneous node failures, and adaptivity to edge dynamics.

\section{Background}
\label{sec:challenge}
We start with a quick primer on the software architecture used in state-of-the-art federated learning (FL) frameworks.



State-of-the-art FL frameworks (e.g., Meta's PAPAYA~\cite{papaya}, LinkedIn's FLINT~\cite{flint},
Google's federated learning framework~\cite{google_fl_sys},
IBM Federated Learning framework~\cite{ibm_fl_framework}, and
Apple's federated task processing system~\cite{apple_fl_system}) commonly adopt a centralized or hierarchical ``single master/many workers'' architecture. In this architecture, the master is typically deployed on a parameter server, acting as a coordinating service provider without data. The workers, on the other hand, connect to numerous edge devices and serve as both the data owners and beneficiaries of federated learning. 

Figure~\ref{fig:fl_archi} 
illustrates the data pipeline between these components
in a typical FL framework.

The high-level design involves two parts: the parameter server that runs the ``master'', and the end-user devices that run the ``workers''. The parameter server comprises three main components: Coordinator, Selector, and Aggregator. While the number of Selectors and Aggregators can scale elastically based on the workload demand, there is only one Coordinator. We summarize the functions of these components as follows.

\noindentparagraph{Coordinator.}
The Coordinator performs three main functions:
\begin{itemize}[leftmargin=*, topsep=0.5ex]
    \item \textit{Task Assignment}. Whenever a new FL application is submitted to the system, the Coordinator assigns the application to a single Aggregator based on the workload among Aggregators and the estimated workload of the application such as application concurrency and model size.

    \item \textit{Client Assignment}. For each available client, the Coordinator constructs a list of eligible applications. The Coordinator assigns each client to an application.

    \item \textit{Task Migration}. The Coordinator moves applications between Aggregators when it detects failed or overloaded Aggregators.  
\end{itemize}

\noindentparagraph{Aggregator.} The Aggregators are \textit{persistent and stateful} to avoid a substantial cold start overhead for a new application. Each Aggregator is responsible for one single FL application and carries out three main responsibilities.
\begin{itemize}[leftmargin=*, topsep=0.5ex]
    \item \textit{Gradient Aggregation}. Once a client completes training, it uploads the trained serialized gradient update to the server. This update is then pushed into an in-memory queue on the Aggregator, which aggregates client gradient updates to produce new versions of the server model of an application. 

    \item \textit{Client Guidance}. The Aggregator guides clients into running the client protocol, such as downloading, uploading, and training configurations, by responding clients requests from the Selector.

    \item \textit{Client Tracking}. The Aggregator is responsible for tracking (i) if clients are satisfied with their assigned applications, and (ii) if clients are still eligible for their assigned applications. The tracking information will be used by the Coordinator to run client assignments.    
\end{itemize}

\noindentparagraph{Selector.} The Selector is the only component that directly communicates with clients and plays two roles:

\begin{itemize}[leftmargin=*, topsep=0.5ex]
    \item \textit{Task Advertisement}. The clients check in with the Selector and report their eligibility for available applications. The Selector summarizes client availability for the Coordinator.

    \item \textit{Request Forwarding}. When a client is assigned an eligible application. The Selector forwards the client to the Aggregator responsible for that application. The Selector also routes clients' requests to the corresponding Aggregator, such as model broadcasting, and reporting client status and gradient updates. 
\end{itemize}

\begin{figure}[t!]
  \centering
  \includegraphics[clip,width=0.7\columnwidth]{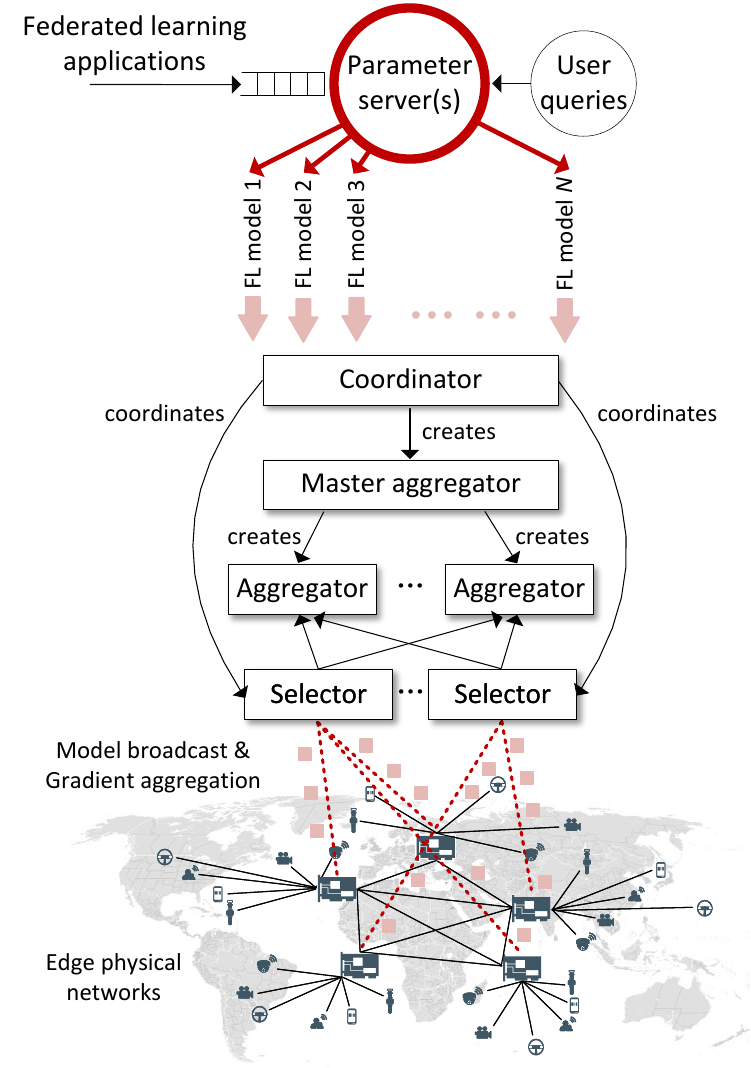}
  \caption{The data pipeline in a typical FL framework.}
    \vspace{-0.05in}
  \label{fig:fl_archi}
\end{figure}



\section{Challenges}
\label{sec:challenges}
We highlight the challenges that we face when applying FL to real-world edge networks in this section.


\subsection*{Challenge \#1: Scaling gracefully with the number of diverse FL applications and edge nodes.}
As emerging edge applications and edge devices grow in quantity and complexity, the number of FL applications submitted to the edge will likely become \emph{huge}. As shown in Figure~\ref{fig:example}, different FL applications may require training of various FL models for different driver profiles (e.g., speeds, distances, passengers), vehicle conditions, and environmental factors (e.g., intersections, pedestrians, moving objects) simultaneously based on the same raw data. This results in the generation of a vast number of FL tasks. Therefore, key considerations in addressing this challenge include:

\noindentparagraph{Distributed task management.} 
As shown in Figure~\ref{fig:fl_archi}, existing production FL systems typically rely on a single instance of Coordinator to direct Aggregators and Selectors across all FL training and testing tasks. While this hierarchical design scales well in datacenter environments, it struggles in edge systems with millions of nodes and many concurrent FL applications. The challenge is amplified by the presence of multiple edge providers~\cite{edge_multiple_providers}, each managing its own set of edge nodes, resulting in no global view of workloads and resources. Without this global view, existing FL systems cannot balance FL tasks effectively across the edge, hindering their ability to scale with new nodes and applications and increasing overall training time. 
Besides longer training time, existing FL systems may fail to advertise the existing/new FL tasks to new/existing clients to further strengthen model accuracy without a global view of running applications.

\noindentparagraph{Application-specific customization.} 
As FL applications diversify with emerging use cases, flexible designs for participant selection~\cite{tifl,oort,client_selection_fl}, compression~\cite{qsgd,signsgd}, and communication protocols—synchronous~\cite{fedavg}, semi-synchronous~\cite{semi_sync_fl}, or asynchronous~\cite{async_fed_opt}—are increasingly needed. However, existing FL systems typically share a single parameter server across applications and enforce fixed FL policies, limiting their ability to support diverse requirements. Achieving application-specific customization often requires using multiple FL frameworks, which compromises efficiency, maintainability, and simplicity.


\subsection*{Challenge \#2: Adapt to practical edge network conditions such as varying bandwidths, unreliable links, high churn, and workload surges.}

The second challenge arises from the first one. To scale with massive FL applications, cloud administrators usually partition all nodes into many sets and assign a parameter server per each set of nodes in a static manner (e.g., one parameter server per rack)~\cite{float,refl, fedscale}. While this assignment approach may work well in datacenters, it lacks the agility to adapt quickly to edge platforms that have millions of resource-constrained nodes. 

Edge environments are characterized by the following unique difficulties:
(1) edge network link delays are unpredictable and vary stochastically due to unreliable links and random access protocols (e.g., in wireless networks~\cite{adaptive_fl_in_edge, semi_sync_fl}), client mobility (e.g., in mobile ad-hoc networks~\cite{fl_for_edge,fl_framework_crowdsourcing}), and randomness of demand (e.g., workload surges~\cite{edge_vision, edge-demlearn}) in arbitrary (and dynamic) geographical edge locations, and 
(2) edge nodes fail or lag unexpectedly (e.g., due to signal attenuation, interference, and wireless channel contention~\cite{adaptive_fl_in_edge, semi_sync_fl}). 
However, unlike datacenter servers, edge nodes have limited computing resources (few-core processors, little memory, and little permanent storage~\cite{edge_vision}) and no backpressure~\cite{ching2025agiledart}. As such, there is little room to adapt to the edge dynamics or handle stragglers by over-provisioning resources or replicating links like previous studies. Moreover, edge nodes may belong to different owners who do not share operation or application information and have to compete for communication resources over edge networks~\cite{kalmia, congestion_games}. Therefore, efforts should be made to replan the data transfer paths \textit{dynamically} and \textit{autonomously} to adapt to the edge dynamics and constrained network bandwidth.

\section{Totoro$^+$ design}
Totoro$^+$ proposes a novel scalable FL system for practical edge networks. To achieve this, Totoro$^+$ incorporates several key techniques and components, which we describe in detail.

\begin{figure}[t]
  \centering
  \includegraphics[clip,width=.9\columnwidth]{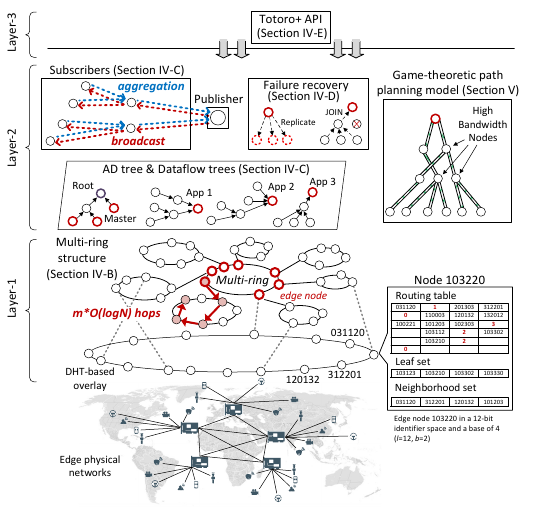}
  \caption{The Totoro$^+$ system overview.}
  \label{fig:overview}
  \vspace{-.05in}
\end{figure}

\subsection{Overview}
Totoro$^+$’s system goals are:

\begin{itemize}[leftmargin=*, topsep=0.5ex]

\item \emph{Scalability.} Totoro$^+$ can scale to process a vast number of FL applications' tasks simultaneously on millions of edge nodes without introducing any centralized bottleneck. 

\item \emph{Adaptivity.} Totoro$^+$ can quickly adapt to the practical edge networks characterized by varying bandwidths, unreliable links, high churn (nodes join and leave), and workload surges in arbitrary geographical edge locations.

\item \emph{Good FL performance.} When handling a vast number of FL applications, Totoro$^+$ can speed up the training process for each of them.

\end{itemize}

As shown in Figure~\ref{fig:overview}, Totoro$^+$ has three layers: a locality-aware P2P multi-ring structure, a publish/subscribe-based forest abstraction, and a high-level API.

\noindentparagraph{Layer 1: locality-aware P2P multi-ring structure.} All distributed edge nodes are self-organized into a DHT-based P2P overlay. Each node has a unique 128-bit NodeId in a very large circular NodeId space. In an edge network with $N$ nodes, the DHT-based P2P overlay guarantees that, no matter where the source node is, any FL data (e.g., model or gradient) can be routed to any destination node within $\mathcal{O}(\log N)$ hops. 
Compared to existing DHTs studies~\cite{pastry, chord, tapestry}, our innovations are\chengwei{:} (1) Totoro$^+$ divides the original single P2P ring structure into many smaller, more manageable locality-aware P2P multi-ring structures that enable locality-aware FL processing; and (2) Totoro$^+$ designs a new boundary-aware two-level routing table that ensures administrative isolation for privacy concerns.

\noindentparagraph{Layer 2: publish/subscribe-based forest abstraction.} 
Built upon Layer 1’s locality-ware P2P multi-ring structure, Totoro$^+$ introduces a new publish/subscribe-based forest abstraction that manages a vast number of FL applications in a scalable manner. 
Each FL application is assigned a dynamically structured dataflow tree that operates with maximum independence and is responsible for disseminating the model from the master to the workers and aggregating the gradients from the workers to the master. 
In addition to dataflow trees, the masters of all dataflow trees construct an advertise-discover (AD) tree to advertise and discover FL applications running over the overlay.
These trees together form a ``forest''. Our innovations are (1) a fully decentralized architecture---unlike other federated learning systems, our system does not have a static assignment of parameter servers. Instead, any edge node can be automatically promoted as a parameter server (master) when workload surges, which significantly improves load balancing and scalability. (2) Advertising and discovering FL applications---new edge nodes that just joined the overlay can easily locate and subscribe the FL applications running over the overlay.
(3) DHT-based routing---the time complexity of model/AppId broadcast, and gradient/AppId aggregation is limited to $\mathcal{O}(\log N)$ hops.

\noindentparagraph{Layer 3: high-level API.} 
We provide a high-level API to abstract away the complexities of P2P overlay construction, dynamic-structured dataflow tree construction, model/AppId dissemination, and gradient/AppId aggregation. Totoro$^+$ supports application-specific customization, allowing application owners to set their own FL policies. 

\subsection{Layer 1: locality-aware P2P multi-ring structure}
\label{subsec:layer1}
Many times, geographic diversity or location matters for training FL applications. For example, a road traffic detection application may require nodes with varying weather condition information in different geographic locations~\cite{du2023driver, wang2023fedstream, wang2022ai}. Training a model on a medical disease prevalent in a certain region may require information from a specific location~\cite{fedhealth,fedsens}. 

Therefore, we organize distributed edge nodes into a locality-aware P2P multi-ring structure to enable locality-aware FL processing.

First, we organize distributed edge nodes into a DHT-based P2P ring overlay, which is similar to the BitTorrent nodes that use the Kademila DHT~\cite{kademlia} for ``trackerless'' torrents. Each edge node is assigned a unique 128-$bit$ NodeId in a very large circular NodeId space (e.g., $0 \sim 2^{128}$). NodeIds are used to identify edge nodes and route FL data (e.g., model or gradient) over large-scale edge networks. 

\begin{figure*}[t] 
 \centering 
 \includegraphics[scale=0.9]
 {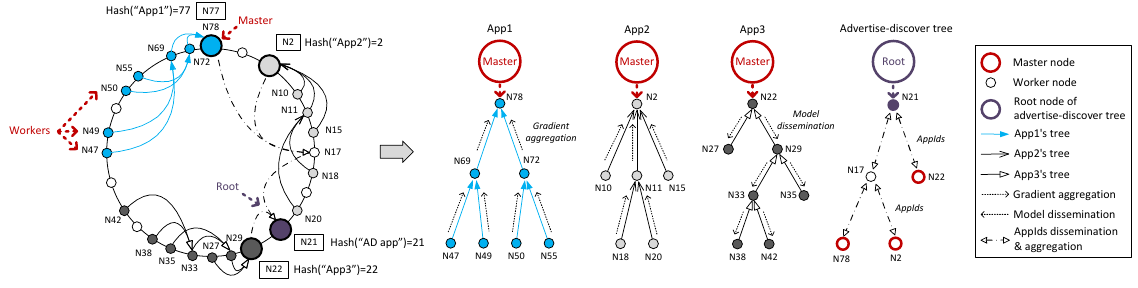}
 \caption{The workflow of Totoro$^+$'s publish/subscribe-based ``forest'' abstraction with three dataflow trees for three FL applications and an advertise-discover tree for application advertise and discovery.}
 \label{fig:forest}
\end{figure*}

To do that, each node needs to maintain three data structures: a routing table, a leaf set, and a neighborhood set.

\begin{itemize}[leftmargin=*, topsep=0.5ex]
\item \emph{Routing table} is used for routing FL data. It consists of node characteristics organized in rows by the
length of the common prefix. The routing works based on prefix-based matching. Every node knows $m$ other nodes in the ring and the distance of the nodes it knows increases exponentially. The routing jumps closer and closer to the destination, like a greedy algorithm, within $\lceil \log_{2^b} N-1\rceil$ hops, where $2^b-1$ is the routing table's entry size.

\item \emph{Leaf set} is used for rebuilding the routing tables upon failures.

\item \emph{Neighborhood set} contains a fixed number of nodes that are ``physically'' closest to that node for maintaining the locality properties. 
\end{itemize}

Second, we divide the original large P2P ring into $m$ smaller, more manageable, locality-aware rings, which we call ``multi-rings”, using Ratnasamy and Shenker’s distributed binning algorithm~\cite{distributed_binning} ($m$ is a configurable parameter). Each ring is an ``edge zone’’ and is characterized by a maximum desired network round-trip time (RTT), called \emph{diameter}. 

Third, we design a new routing table to enable administrative isolation. The challenge is, \emph{how to achieve path convergence to enable administrative isolation, i.e., data paths from different nodes in an edge site should converge at a node in that edge site?}
Existing DHTs~\cite{pastry, chord, kademlia, tapestry} do not guarantee path convergence as those systems try to optimize the search path to reduce response latency. To route a packet to an arbitrary destination key, the packet will be routed to the destination node in another site as long as it has a longer NodeId prefix matching the key. To address this challenge, we make the following changes to existing routing tables: (1) each NodeId now has $(m+n)$-bit, where the $m$-bit prefix presents the zone Id and the $n$-bit suffix represents the NodeId within a zone. Let $P$ denote the prefix of NodeId, i.e., $P_1...P_n$. Let $S$ denote the suffix of NodeId, i.e., $S_1...S_n$. Then the NodeId equals $D=P*2^n+S$; correspondingly (2) each node’s routing table will have two levels: the level 1 routing table and the level 2 routing table. The $i_{th}$ entry in the level 1 routing table with $m$ entries at peer $x$ equals to $(P_x+2^{i-1})$mod $2^m\ast2^n$. The $i_{th}$ entry in the level 2 routing table with $n$ entries at peer $y$ equals to $(S_y+2^{i-1})$mod $2^n$.

To achieve administration isolation, the administrator of an edge site can leverage the level 1 routing table to control
the data flow among different edge zones. For example, when an FL application should be running only within an edge site, any packet generated by that FL application should only travel within this edge site. Administrators can check the destination of packets. If the packet's destination shares a different prefix with the administrator's zone Id, the administrator can block the packet before routing it outside the edge zone.

\subsection{Layer 2: publish/subscribe-based forest abstraction}
Built upon Layer 1, Totoro$^+$ introduces a new publish/subscribe-based ``forest'' abstraction for managing a vast number of FL applications in a scalable manner. Specifically, our goal is to achieve a balanced distribution of masters for hundreds of thousands of FL applications, ensuring that they are not concentrated on a few overloaded nodes. The key innovation is leveraging DHTs to decompose the FL system architecture from \emph{1:n} to \emph{m:n}, where each FL application can be assigned a dynamic-structured dataflow tree that operates with maximum independence, thereby radically improving load balance and scalability. A DHT is a hash table that partitions the key space and distributes the parts across a set of nodes, providing a lookup service similar to a hash table. The trick with DHTs is that the node that gets to store a particular key is found by hashing that key, so in effect, the hash-table buckets are now independent nodes in a network. 

We utilize the fully decentralized nature of DHT to process FL applications' tasks at an extreme scale: unlike other FL systems, Totoro$^+$ does not have a static assignment of a parameter server. Instead, the parameter server is broken down into many components, such as the coordinator, client selector, and aggregator. Any edge node can act as any FL application's coordinator, aggregator, client selector, worker (participating edge device), or any combination of the above, thereby radically improving load balance and scalability.

Figure~\ref{fig:forest} lays out the construction of the publish/subscribe-based “forest”. It has the following steps. 

Built on top of Layer 1, all edge nodes are structured into a DHT-based P2P overlay. Here we use one ring as an example. The DHT-based P2P overlay guarantees that: \emph{given a message and a key, no matter where the source node is, the message can be reliably routed to the node whose NodeId is numerically closest to that key, within $\lceil \log_{2^b} N-1\rceil$ hops, where $2^b-1$ is the routing table’s entry size.}

The first step is constructing application-based logical ``trees’’ of nodes and ensuring that these trees are well balanced over the large-scale edge topologies (Figure~\ref{fig:forest} left). 
\begin{enumerate} [leftmargin=*, topsep=0.5ex, label=\alph*.]
    \item When any new FL application is launched, we calculate the application's AppId, which equals the cryptographic hash of the application's textual name, the creator’s public key, and a random salt, \emph{AppId=hash(``FL application'')}. The hash is computed using the collision resistant SHA-1 hash function, ensuring a uniform distribution of AppIds. 

    \item Then, the edge node processing the application’s data routes a \texttt{JOIN} message using AppId as the key. Since all nodes belonging to the same application use the same key, their \texttt{JOIN} messages eventually arrive at a rendezvous node, with NodeId numerically close to AppId. The rendezvous node is set as the root of this application's tree. 

    \item The unions of all \texttt{JOIN} messages' paths are registered to construct the tree, in which the internal node, as the forwarder, maintains a children table for the group containing an entry (IP address and AppId) for each child. All of the trees together form a ``forest'' abstraction.

    \item For each application's tree, we designate the root node as the master, the internal nodes as the coordinator, aggregator, and client selector components, and the leaf nodes as the workers (participating edge devices).
\end{enumerate}

\noindent\underline{\emph{Rationale:}} (1) Since different applications have different AppIds, the paths and the rendezvous nodes of their spanning trees will also differ, resulting in an even distribution of trees across all edge nodes. (2) Because all nodes are equal, each node can serve as a leaf, internal, or root node for different applications, thus removing the scalability bottleneck without overburdening any single node. 

The second step is implementing a topic-based publish/ subscribe messaging protocol within the dataflow tree. Each FL application has a ``topic''. The master is the ``publisher'', and the workers are the ``subscribers''. 

\begin{enumerate} [leftmargin=*, topsep=0.5ex, label=\alph*.]
    \item \emph{Model broadcast.} The master disseminates the FL model to the workers along the tree. Then each worker independently trains a local model and computes the model updates (e.g., gradients and weights) on the local data.
    
    \item \emph{Gradient aggregation.} Once the workers have completed the computation, the master aggregates model updates from the workers, in which each level of the tree progressively aggregates the updates from tree leaves to the root. To meet the diverse needs of different applications, owners can specify different aggregation functions in their trees. For instance, FedAvg~\cite{fedavg} works well in most situations, while FedProx~\cite{fedprox} demonstrates superior performance in highly heterogeneous settings for more stable and accurate convergence.
\end{enumerate}

\noindent\underline{\emph{Rationale:}} (1) Due to the loosely coupled interaction between the publisher and subscribers, Totoro$^+$ can simultaneously support a large number of dataflow trees with a wide range of tree sizes and a high rate of membership turnover. (2) The use of DHTs enables the efficient construction of aggregation trees and multicast services, as their converging properties guarantee model broadcast or gradient aggregation to be fulfilled within only $\mathcal{O}(\log N)$ hops, which places an upper bound on worst-case latency for data transfers. 

The third step is constructing an advertise-discover (AD) tree that enables edge nodes to advertise and discover FL applications running over the overlay (Figure~\ref{fig:forest} right). 
\begin{enumerate} [leftmargin=*, topsep=0.5ex, label=\alph*.]
    \item Similar to the first step, each master of various applications' trees routes a \texttt{JOIN} message using \emph{AppId=hash(``AD application'')} as the key. Since all masters use the same key, their \texttt{JOIN} messages eventually arrive at a rendezvous node whose NodeId is numerically closest to AppId. The rendezvous node becomes the root of the AD tree.
    
    \item The unions of all \texttt{JOIN} messages' paths are registered to construct the AD tree. In addition to children tables for FL applications, the internal node in the tree works as a forwarder that maintains a children table consisting of the IP addresses of each child and AppIds of FL applications each child is hosting. 

    \item The root of the AD tree is the ``publisher'', and other nodes are the ``subscribers''. The publisher disseminates AppIds of FL applications to each node along the tree, whereas the subscribers in different levels of the tree progressively aggregate AppIds of FL applications to the root.

    \item Whenever a new edge node that joins the DHT-based P2P overlay wants to discover FL applications running over the overlay or an existing edge node intends to switch to a different FL application running on the overlay, it can route a \texttt{JOIN} message using the same AppId to join the AD tree as a subscriber and receive the AppIds of interest. After that, it can route a \texttt{LEAVE} message using the same AppId to leave the AD tree.
    
\end{enumerate}

\noindent\underline{\emph{Rationale:}} 
(1) Since edge nodes may join and leave the overlay frequently and suddenly, joining nodes can subscribe to the AD tree first and receive the information of current FL applications instead of sending broadcast messages to all other nodes, making Totoro$^+$ adaptive to practical edge networks. 
(2) Because edge nodes can leave the AD tree immediately after they receive the necessary information, the tree may have $M$ nodes plus some intermediate nodes $N'$ most of the time and $N'<<N$, where $M$ denotes the number of current FL applications and $N$ denotes the number of nodes in the overlay. Therefore, with the use of DHTs, the AppId broadcast and aggregation can also be fulfilled within $\mathcal{O}(\log (M+N'))$ hops. 

{\color{dv}
A comparison of Totoro and Totoro$^+$ in terms of FL application discovery and advertisement is presented in Appendix~\ref{app:advertise_discover}.
}


\subsection{Failure recovery}
In the case of node failures, we use a parallel recovery approach to repairing the dynamic-structured dataflow trees. A dataflow tree has a master and several workers, and the parallel recovery approach leverages two respective ways to deal with the failure of a worker and the failure of a worker.

\noindentparagraph{Worker node fails.} Periodically, each internal node in the tree sends a keep-alive message to its child nodes. A child node suspects its parent fails when it cannot receive keep-alive messages from its parent node. In such a scenario, the child node routes a \texttt{JOIN} message using AppId as the key. The overlay network will route the message to a new parent, create an alternative route, and repair the dataflow tree.

\noindentparagraph{Master node fails.} In each training round, we replicate the state associated with each application's master across the $k$ nodes in its neighborhood set ($k=2$, by default), which consists of a fixed number of nodes that are ``physically'' closest to the master node. If the master node fails, its immediate child node will detect it, route a \texttt{JOIN} message using AppId as the key, and identify a new master node to take over the role of the failed master. According to Layer 2: publish/subscribe-based forest abstraction, the new master node's NodeId will be numerically closest to the AppId, and it will first contact the nodes that have state replicas to recover the state and inform workers to continue the training tasks like aggregation, global model update, and broadcast.

The tree repair process scales well: failure detection is done by sending messages to child nodes only. Failure recovery is also local. Only a small number of nodes $(\mathcal{O}(\log_{2^b}N))$ is involved, in which $2^b$ is the fanout of the tree. In addition, the replicas of states are typically transmitted over local networks with plentiful bandwidth; hence, the communication overhead is negligible.

\begin{table}[t!]
\centering
\scriptsize	
\setlength{\belowcaptionskip}{-10pt} 
\begin{tabular}{|p{6.5cm}|}
\hline
\textbf{Join(IP, port, site)} \\
Edge node joins the DHT-based P2P overlay network. \\
\hline
\textbf{CreateTree(app\_id)} \\
Application owner creates a dynamic-structured dataflow tree and configures the parameters (e.g., fanout).\\
\hline
\textbf{Subscribe(app\_id)} \\
Edge node sends a \texttt{JOIN} message to subscribe to a dynamically-structured dataflow tree
with the topic equal to app\_id. Application owner can specify her client selection function in the API. 
\\
\hline
\textbf{Unsubscribe(app\_id)} \\
Edge node sends a \texttt{LEAVE} message to unsubscribe to a dynamically-structured dataflow tree
with the topic equal to app\_id.\\
\hline
\textbf{Broadcast(app\_id, object)} \\
Application’s master disseminates the model or AppIds to workers along its dynamically-structured dataflow tree. Application owner can specify her compression function in the API.\\
\hline
\textbf{onBroadcast(app\_id, object)} \\
Callback. Invoked when the worker receives any model, updates, or AppIds from the application’s master. \\
\hline
\textbf{Aggregate(app\_id, object)} \\
Application’s master aggregates updates or AppIds from the workers to the root. Application owner can specify her aggregation function in the API. \\
\hline
\textbf{onAggregate(app\_id, object)} \\
Callback. Invoked when any internal node receives updates or AppIds from a child node.\\
\hline
\textbf{onTimer(app\_id)} \\
Callback. Invoked periodically. Application’s master uses it to get the information about the progress of training (e.g., round\_num, accuracy, straggler\_id) and inference. \\
\hline
\end{tabular}
\caption{Totoro$^+$ API.}
  \label{tab:api}
\vspace{-0.05in}
\end{table}

\subsection{Layer 3: high-level API}
We abstract key components of Totoro$^+$ to provide an easy-to-use API (see Table~\ref{tab:api}). Totoro$^+$ supports
application-specific customization, allowing application
owners to set their own FL policies. The Pastry DHT and Scribe multicast infrastructure are written in Java, with end-user functionality encapsulated in a Python API. This is done so that users do not have to deal with two libraries in separate languages. Thus Totoro$^+$ can be easily integrated with popular FL frameworks such as PyTorch~\cite{pytorch}, PySyft~\cite{pysyft}, and TensorFlow Federated~\cite{tensorflow_federated}.

\noindentparagraph{Application-level customization.}
Totoro$^+$'s APIs support application-level customization.
For example, to prevent potential leakage of model weights to other nodes, application owners can specify various privacy techniques in \texttt{Aggregate(app\_id, object)}, such as differential privacy~\cite{differentially_private_fl}, secure aggregation~\cite{papaya}, and homomorphic encryption~\cite{homomorphic_fl}. Nodes that subscribe to an FL application adhere to the privacy technique specified by that application during the FL process. Take differential privacy as an example, if an application owner launches an FL application and specifies the use of differential privacy with Gaussian noise to secure weights, the edge nodes, as the master, coordinator, aggregator, and client selector, will operate following the privacy technique. Similarly, the leaf nodes, serving as workers, will apply Gaussian noise to local training. 
Moreover, to realize client selection techniques, application owners can rule out edge nodes that do not meet specific requirements (e.g., data distributions, battery, network conditions, etc.) when receiving \texttt{JOIN} messages from edge nodes.


\noindentparagraph{Multi-rings.} Application owners can specify whether their applications span multiple zones. For example, a road traffic detection application may require nodes with different weather condition information in different geographical locations. If an application needs to ingest data sources across multiple zones, it will traverse multiple zones (at most $m$) to build the dataflow tree, resulting in $m\times\mathcal{O}(\log N)$ routing hops.


\section{Game-theoretic Path Planning Model}

Totoro$^+$ introduces a new game-theoretic path planning model that can replan the data transfer paths in dynamically structured dataflow trees to adapt to unreliable edge networks. 

In most real-world edge networks, link delays are unpredictable and vary stochastically~\cite{edge_network_bandwidth}. When a link becomes slow, it can disrupt communication with its child and parent nodes, thereby affecting the entire path from the leaves to the root passing through this link in dataflow trees.

The challenge is that, in many cases, we don't know the quality of the network links in advance, such as the probability of successfully transmitting packets in wireless sensor networks~\cite{edge_network_bandwidth}. This information is often obtained by actually sending packets and observing the outcomes. Furthermore, the network links have limited capacities~\cite{edge_network_bandwidth}. When a link carries too much traffic, its overall data rate dramatically decreases. 
Therefore, we face two dilemmas when planning the paths for model dissemination and gradient aggregation. 

One is \textit{whether to explore new or unknown links or exploit well-known ones}. If we only rely on the known links, we may miss out on finding a better one with a higher success rate or data rate. However, exploring too many new links may result in more packet losses and higher communication delays. The other is \textit{whether to prioritize link quality or link capacity.} Maximizing link quality avoids packet loss and errors but may overload these paths, leading to congestion and increased latencies. However, prioritizing link capacity may increase the risk of transmission failures instead.

This is where Multi-Armed Bandit (MAB) algorithms~\cite{stochastic_online_shortest_path,regret_analysis} and game theory~\cite{congestion_games,algorithmic_game_theory,nash_regret_avg} come into play. 
Imagine someone at a bustling amusement park, surrounded by various facilities like roller coasters, merry-go-rounds, and Ferris wheels. Each facility differs in its entertainment type and its queue length. Initially, she explores different facilities randomly, collecting data on two key aspects: the joy each ride brings and the length of the waiting time. Her goal is to maximize her day by finding the facilities that offer the best balance of fun and short waiting times. However, her decision-making process influences the decisions of others. The presence of other tourists, who also choose facilities based on their own experiences and preferences, affects the queues at each location. The collective decisions of all tourists shape these queues, each acting on their assessments of enjoyment and waiting times. As time passes, some facilities may become busier, while others may see reduced waiting times. She must constantly adapt to these changing conditions, balancing her preferences for fun with the evolving waiting times.


{\color{dv}
Therefore, the path-planning problem can be formulated as a general-sum congestion game with bandit feedback~\cite{gen_sum_game_icml_2005, gen_sum_game_jmlr_2023}: we follow a policy to explore different paths to learn about their rewards (e.g., link quality and link capacity) or exploit the paths that offer the highest rewards and update the policy based on the rewards. And the rewards are influenced by collective decisions on paths. 
We gradually find out the optimal policy to select paths over edge networks. 

The formal definitions and theoretical concepts underlying this formulation are
provided in Appendix~\ref{app:game_prelim}.
Specifically, we review general-sum matrix games, policy representations, Nash
equilibrium and Nash regret, as well as potential and congestion games, together
with different feedback models. 
In the following, we build upon these preliminaries to formulate a routing problem, define the corresponding optimization objective, and develop a game-theoretic
algorithm with theoretical performance guarantees. In addition, we present a comparison of the mathematical models underlying Totoro and Totoro$^+$ in Appendix~\ref{app:math_model_comparison}.
}

\subsection{Problem formulation}
A dataflow tree in Totoro$^+$ consists of at most $N$ nodes, and there are at most $P$ paths to reach the root of the tree. Since each path $p\in [P]$ has different success rate, and capacity, we let $\mathcal{R}^p(\cdot | k,\theta_p)\in [0,1]$ denote the reward distribution for path $p$ with mean $r^p(k,\theta_p)$ given that $k$ nodes select path $p$ with mean success rate being $\theta_p$.

The higher the mean success rates, the higher the rewards. In contrast, the more nodes select the same path, the less the rewards they receive as the success rate and data rate drop.

Suppose the joint paths chosen by all nodes are denoted by $\textbf{p}=[p_1, p_2, \cdots, p_N]$, where $p_n$ denotes the path selected by node $n$, then we let $n^p(\textbf{p})$ denote the number of nodes selecting path $p$ given the joint action $\textbf{p}$. A node $n$ selects path $p$ and gets the reward $r_n(p)$ drawn from the reward distribution $\mathcal{R}^p(\cdot | n^p(\textbf{p}),\theta_p)$ with mean $r^p(n^p(\textbf{p}),\theta_p)$.

\noindentparagraph{Optimization objective.} 
We aim to efficiently route $T$ packets (such as gradients in many rounds) from a worker node to a master node, minimizing the overall time taken. In other words, our goal is to maximize total rewards. Also, since collective actions taken by all nodes influence the rewards, we resort to see how those actions affect others' decisions. 

First, we define policy $\pi_n$ followed by node $n$ to select paths, where $\pi_n$ is from the probability simplex over node $n$'s action space $P_n$, denoted by $\Delta(P_n)$, where $P_n\subseteq [P]$ is the set of paths available for node $n$'s selection. Similarly, we can define a general policy $\pi=[\pi_1, \pi_2, \cdots, \pi_N]$ as a tuple in the joint space $\Delta(P_1)\times\cdots \times\Delta(P_N)$. Hence, we can have $\textbf{p}=(p_1, p_2, \cdots, p_N)\sim \pi$ and $p_n\stackrel{i.i.d}{\sim} \pi_n$.

Then, since the policy influences the rewards the nodes receive, we use the rewards to represent \textit{the value of a policy}. We define the value of policy $\pi$ for node $n$ as $V^\pi_n = \mathbb{E}_{\textbf{p}\sim\pi}[r_n(\textbf{p})]$, where $r_n(\textbf{p}) = r^{p_n}(n^{p_n}(\textbf{p}),\theta_{p_n})$. Let $\pi_{-n}$ be the marginal joint policy of nodes $1,\cdots, n-1, n+1,\cdots, N$. We refer to~\cite{congestion_games,nash_regret_avg} to  define the \textit{Nash Regret} after $T$ packets as
\begin{align}
    \text{Nash-Regret}(T)=\sum_{t=1}^T \max_{n\in [N]}(V^{\pi^+_n,\pi^t_{-n}}_n - V^{\pi^t}_n),
\end{align}
where $\pi^t$ is the policy for packet $t$, 
$V^{\pi^t}_n$ is the value of policy $\pi^t$ for node $n$, 
$\pi^+_n=\argmax_{\mu \in \Delta(P_n)} V_n^{\mu, \pi^t_{-n}}$ represents the best response of node $n$ under policy $\pi$, and $V^{\pi^+_n,\pi^t_{-n}}_n$ is the value of the policy $(\pi^+_n,\pi^t_{-n})$. The Nash Regret indicates the value difference between when a node unilaterally switches its policy and when it follows the same policy. The optimization objective is to learn a general policy $\pi^*$ over $T$ packets such that 
\begin{align}
    \text{Nash-Regret}(T) \leq \epsilon,
\end{align}
where $\epsilon >0$ is a given tolerance threshold.

\subsection{Our algorithm}
\label{subsec: our solution}

\begin{algorithm}[t]
    \small
    \caption{Game-theoretic distributed hop-by-hop routing algorithm}
    \label{pseudo:our_algorithm}
    \begin{algorithmic}[1]
        \INPUT{mixture weights $\alpha,\beta \in [0,1]$; initial policy $\pi^1_n$ for all $n\in N$.}
        \For{episode $k=1, 2, \cdots$}
            \For{packet $t=1,\cdots, \tau$}
                \State \parbox[t]{200pt}{Each node $n$ in parallel selects hop $p^{k,t}_n\sim \pi^k_n$ to send a packet, observes reward $r^{k,t}_n$. \strut}
                \label{pseudocode:collect_info}
            \EndFor
            \For{node $n=1,\cdots, N$} \textbf{in parallel}                
                \State \parbox[t]{200pt}{$\rho_n^k \leftarrow \argmin_{\lambda\in \Delta(P_n)} \text{det}(M(\lambda))$.\strut}
                \label{pseudocode:d_optimal_design}
                                
                \State \parbox[t]{200pt}{$\widehat{\nabla}^k_n\Phi(p) \leftarrow \frac{1}{\tau}\sum_{t=1}^\tau \psi(p)^\top M(\pi^k_n)^{-1}\psi(p^{k,t}_n)r^{k,t}_n$, for all $p\in P_n$.\strut}
                \label{pseudocode:estimate_gradient}
                \State \parbox[t]{200pt}{$\Tilde{\pi}^{k+1}_n\leftarrow \argmax_{\lambda\in \Delta(P_n)}\langle \lambda, \widehat{\nabla}^k_n\Phi \rangle$.\strut}  \label{pseudocode:get_optimal_policy_by_gradient}
                \State \parbox[t]{200pt}{$\pi^{k+1}_n \leftarrow \alpha\big[\underbrace{\pi^k_n+\beta(\Tilde{\pi}^{k+1}_n-\pi^k_n)}_{\text{Frank-Wolfe update}}\big]$ 
                $+ \underbrace{(1-\alpha)\rho^k_n}_{\text{Exploration}}$.\strut}\label{pseudocode:update_policy_with_exploration}
            \EndFor
        \EndFor
    \end{algorithmic}
\end{algorithm}

We propose a game-theoretic distributed hop-by-hop routing algorithm based on the bandit feedback model~\cite{stochastic_online_shortest_path, bandit_algo_book,online_influence,congestion_games,algorithmic_game_theory}. The pseudo-code of the game-theoretic distributed hop selection algorithm is presented in Algorithm~\ref{pseudo:our_algorithm}.
{\color{dv}
The
notations used in the algorithm are summarized in Appendix~\ref{app:notation_table}.
}

When node $n$ receives a packet $t$ in episode $k$ from a previous node, it follows the current policy $\pi^k_n$ to select the next hop $p^{k,t}_n$, 
sends the packet, and observes the corresponding reward $r^{k,t}_n$ (line~\ref{pseudocode:collect_info}). Whenever node $n$ has sent $\tau$ packets and collected $\tau$ rewards, it updates the policy $\pi^k_n$ based on the $\tau$ rewards collected.

First, we identify an exploratory policy with its policy correlation matrix:
\begin{align}
    M(\lambda) = \sum_{p\in \lambda} \lambda(p)\psi(p)\psi(p)^\top,
    \label{eq:policy_correlation_matrix}
\end{align}
which has the minimum determinant among all possible policies in the probability simplex $\Delta(P_n)$ over node $n$'s action space $P_n$ (line~\ref{pseudocode:d_optimal_design}), where $\lambda(p)$ denotes the probability of selecting the next hop $p$, and $\psi(p)$ represents the vector-valued function mapping hop $p$ to a one-hot vector. 

Next, we use linear regression with $\tau$ rewards to estimate the gradient of the mapping from policy $\pi^{k}_n$ to potential rewards for node $n$ (line~\ref{pseudocode:estimate_gradient}), where $M(\pi^k_n)^{-1}$ represents the inverse of $M(\pi^k_n)$. We yield $|P_n|$ gradient estimators, which form a $|P_n|$-dimension vector $\widehat{\nabla}^k_n\Phi = [\widehat{\nabla}^k_n\Phi(p_1), \cdots, \widehat{\nabla}^k_n\Phi(p_{|P_n|})]$. With the gradient estimator for all possible hops, we calculate the dot product of the estimator and all possible policies $\lambda \in \Delta(P_n)$ and obtain the optimal policy $\Tilde{\pi}^{k+1}_n$ that yields the maximum dot product (line~\ref{pseudocode:get_optimal_policy_by_gradient}).

Lastly, we leverage the Frank-Wolfe method~\cite{frank_wolfe_algo} to linear combine current policy $\pi^k_n$ with optimal policy $\Tilde{\pi}^{k+1}_n$:
\begin{align}
    \pi^k_n+\beta(\Tilde{\pi}^{k+1}_n-\pi^k_n),
    \label{eq:frank_wolfe_update}
\end{align}
where $\beta$ is in the range $[0,1]$. The Frank-Wolfe update in Eq.~(\ref{eq:frank_wolfe_update}) guarantees that the results of the linear combination fall in the probability simplex $\Delta(P_n)$. In addition, we add exploratory policy $\rho^k_n$ to the Frank-Wolfe update (line~\ref{pseudocode:update_policy_with_exploration}):
\begin{align}
    \alpha\big[\pi^k_n+\beta(\Tilde{\pi}^{k+1}_n-\pi^k_n)\big] + (1-\alpha)\rho^k_n,
    \label{eq:frank_wolfe_update_with_exploration}
\end{align}
where $\alpha$ is also in the range $[0,1]$ and determines the extent to explore different policies. Node $n$ follows the final policy $\pi^{k+1}_n$ by Eq.~(\ref{eq:frank_wolfe_update_with_exploration}) to route packet $\tau+1$ to $\tau+\tau$ in episode $k+1$, and hence each episode $k$ has individual policy $\pi^{k}_n$ to route $\tau$ packets. 
{\color{dv}
We provide a numerical example to illustrate the workflow of Algorithm~\ref{pseudo:our_algorithm} in Appendix~\ref{sec:numerical_example}, discuss how local congestion is observed and inferred at the application level in Appendix~\ref{app:local_congestion}, and analyze the implementation and adaptivity of Algorithm~\ref{pseudo:our_algorithm} in Appendix~\ref{sec:impl_algo}.
}

\subsection{Theoretical analysis.}
\label{subsec:theor_anal}
We conduct a theoretical analysis on Algorithm~\ref{pseudo:our_algorithm}
in two aspects: the \textit{Nash regret bound} and \textit{time complexity}. 

\begin{theorem}
    Let $T=k\tau$ and assume that each policy in $\Delta(P_n)$ has no zero element for all $n$. By running Algorithm~\ref{pseudo:our_algorithm} with gradient estimator $\widehat{\nabla}^k_n\Phi(p)$ defined in line~\ref{pseudocode:estimate_gradient} and exploratory policy $\rho^k_n$ defined in line~\ref{pseudocode:d_optimal_design}, if $k\geq \frac{2}{N}$ , then with probability $1-\delta$, we have    
    \begin{align}
        \text{Nash-Regret}(T) \leq \Tilde{\mathcal{O}}(N^2 T^{5/6}\log N).\nonumber
    \end{align}
    \label{theo:nash_regret_bound}
\end{theorem}
The proof details of Theorem~\ref{theo:nash_regret_bound} are given in Appendix~\ref{sec:proof_nash_regret}.
The results in Theorem~\ref{theo:nash_regret_bound} demonstrate a \textit{sublinear} Nash regret with polynomial dependence on $N$, the number of nodes. Following Definition 1 in~\cite{congestion_games} and Section 3 in~\cite{nash_regret_avg}, we obtain the following corollary.
\begin{corollary}
    Algorithm~\ref{pseudo:our_algorithm} reaches an $\epsilon$-approximate Nash equilibrium. 
    \label{coro:epsilon_approximate}
\end{corollary}
Corollary~\ref{coro:epsilon_approximate} tells us that nodes that run Algorithm~\ref{pseudo:our_algorithm} to update their policies can find a general policy such that a node receives at most $\epsilon$ more rewards by unilaterally changing its policy, implying reaching an $\epsilon$-approximate Nash equilibrium. In addition, we provide the Nash regret guarantee under bounded asynchrony in the following corollary.

{\color{dv}
\begin{corollary}
\label{coro:async}
Under bounded asynchrony, the Nash regret guarantee in Theorem~\ref{theo:nash_regret_bound} continues to hold, up to an additional sublinear regret term induced by asynchrony.
\end{corollary}
The proof details of Corollay~\ref{coro:async} are given in Appendix~\ref{app:proof_corollary2}.} Finally, we analyze the time complexity of Algorithm~\ref{pseudo:our_algorithm} in the following theorem.

\begin{theorem}
    The time complexity of Algorithm~\ref{pseudo:our_algorithm} in each episode is $\mathcal{O}(\tau \log^3N + |\Delta (P_n)| \log^3N)$, where $\tau$ represents the number of packets sent in each episode, and $|\Delta (P_n)|$ denotes the number of policies node $n$ can adopt.
    \label{theo:time_complexity}
\end{theorem}
The proof details of Theorem~\ref{theo:time_complexity} are given in Appendix~\ref{sec:proof_time_complexity}.

\section{Implementation}
\label{sec: Implementation}

{\color{dv}
We implement Totoro$^+$ on top of the Pastry (v.2.1)~\cite{pastrycode} and PyTorch (v.2.10.0)~\cite{pytorch_software} software stacks. Such implementation
choice is motivated by the following considerations: (1) Pastry~\cite{pastry} is a widely used overlay and routing network for the implementation of a DHT similar to Chord~\cite{chord}. Instead of implementing another distributed system core, we can leverage Pastry’s excellent routing substrate (e.g., $\mathcal{O}(\log N)$ node lookup), self-repairing routing table, message transportation layer, and scalable application-level multicast infrastructure (Scribe~\cite{scribe}). These features greatly simplify the development process. (2) PyTorch~\cite{pytorch} provides a flexible and intuitive API for building neural networks. It also offers a rich ecosystem of libraries (e.g., TorchVision~\cite{torchvision} and HuggingFace Transformers~\cite{huggingface_transformers}) that support many modern neural network models.
}


We made the following major modifications: 
(1) We changed the original single P2P ring structure to a new locality-aware P2P multi-ring structure.  
(2) We implemented a fully decentralized ``many masters/many workers''
architecture by utilizing DHTs and adding several data
structures: a list of operations for tracking routing paths, selecting masters and workers, and constructing dynamically structured training trees. 
(3) We implemented a publish/subscribe messaging pattern for scalable model propagation and gradient aggregation. We introduced a serialization mechanism to convert trained models into binary arrays for low-cost communication over edge networks. 
(4) We introduced an advertise-discover mechanism to enable nodes to advertise and discover FL applications running over the overlay. We implemented a parallel recovery approach to dealing with the failure of masters and the failure of workers.
(5) We implemented a game-theoretic path planning model to ``autonomously'' replan or repair the dynamically structured dataflow tree by collecting feedback, detecting node stragglers or failures, balancing network traffic, and creating alternative routes. 
{\color{dv} We discuss the adaptivity of Totoro$^+$ to edge heterogeneity in Appendix~\ref{app:adap_edge_hetero}.
}

\begin{figure*}[t!]
  \captionsetup[subfigure]{width=4.1cm}
  \subfloat[Real-world edge topologies generated from the EUA dataset~\cite{edge_network_dataset}.]{\includegraphics[width=0.24\textwidth]{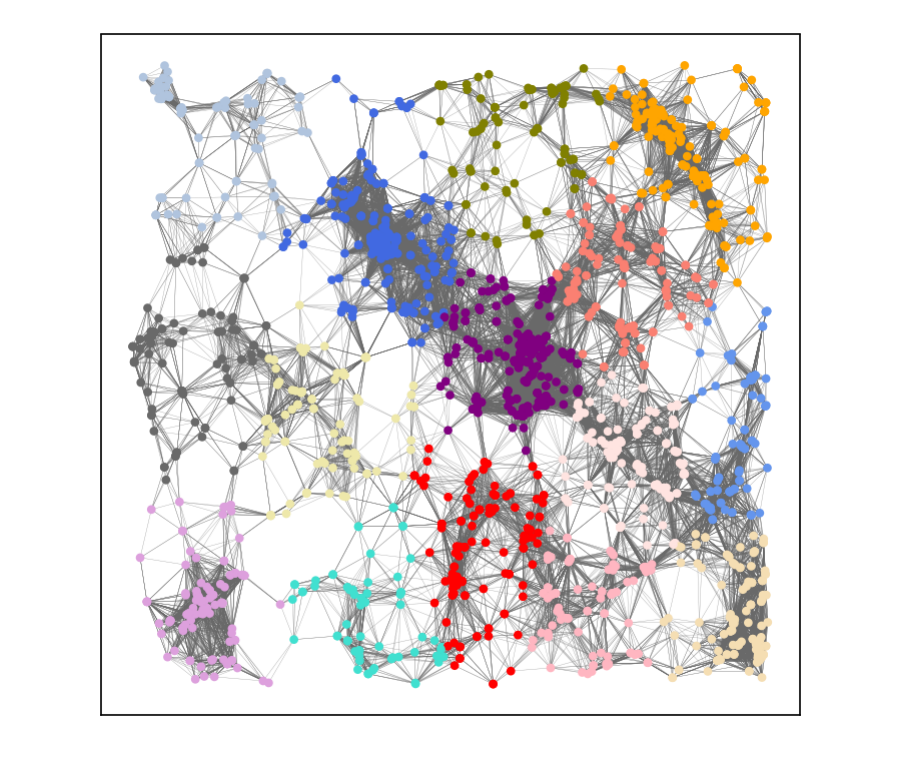}
  \label{fig:edge_topology}}
\hspace{-0.1in}
\hfill
  \subfloat[Normal probability plot of the
number of masters mapped per node.]{\includegraphics[width=0.245\textwidth]{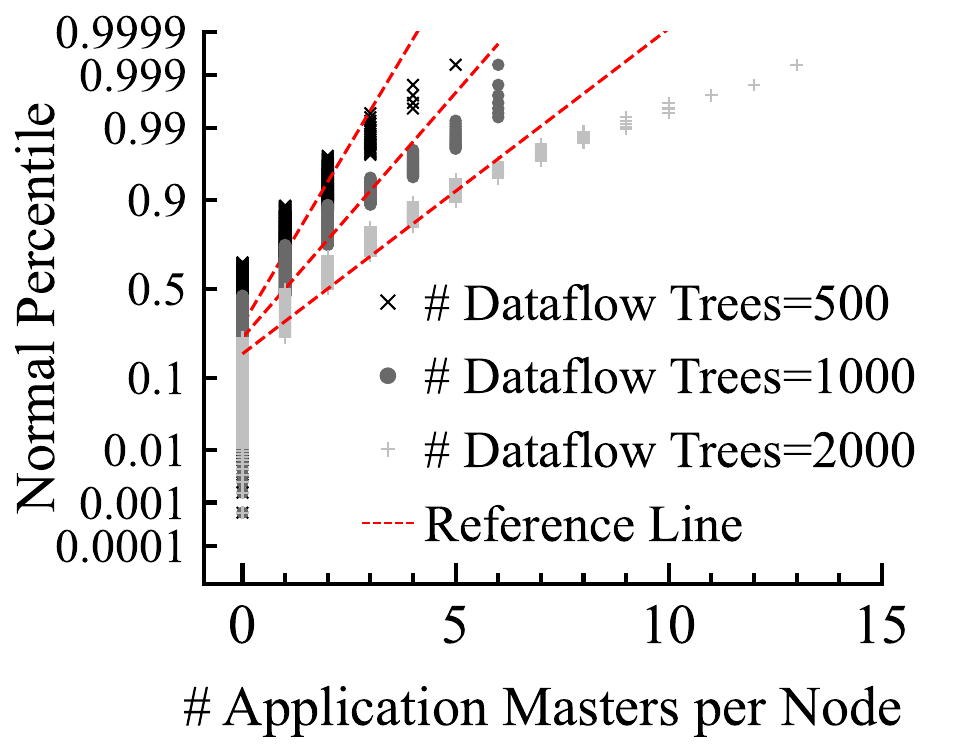}
  \label{fig:cdf}}
\hspace{-0.1in}     
 \hfill
   \subfloat[Totoro$^+$ scales with \#masters to handle the varying workload.]{\includegraphics[width=0.24\textwidth]{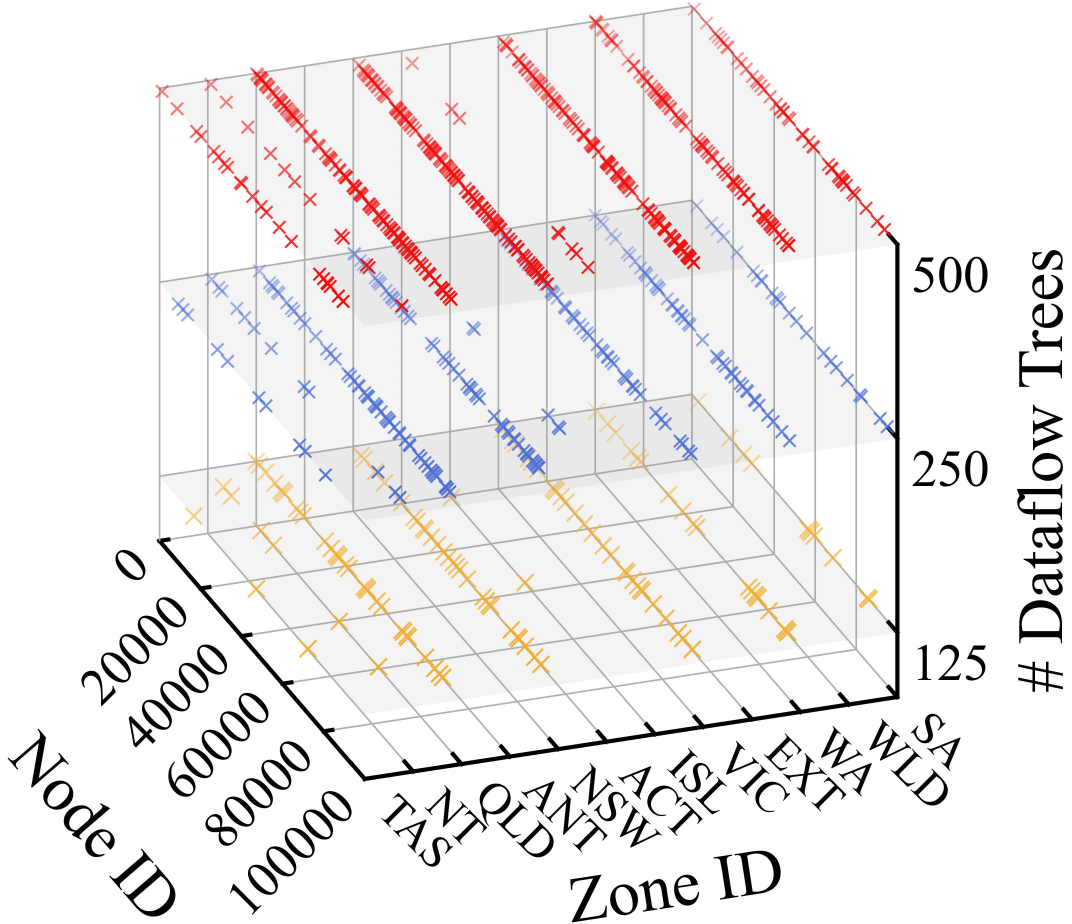}
   \label{fig:3d_plot}}
 \hfill
     \subfloat[The distribution of Totoro’s
dataflow trees over edge topologies.]{\includegraphics[width=0.245\textwidth]{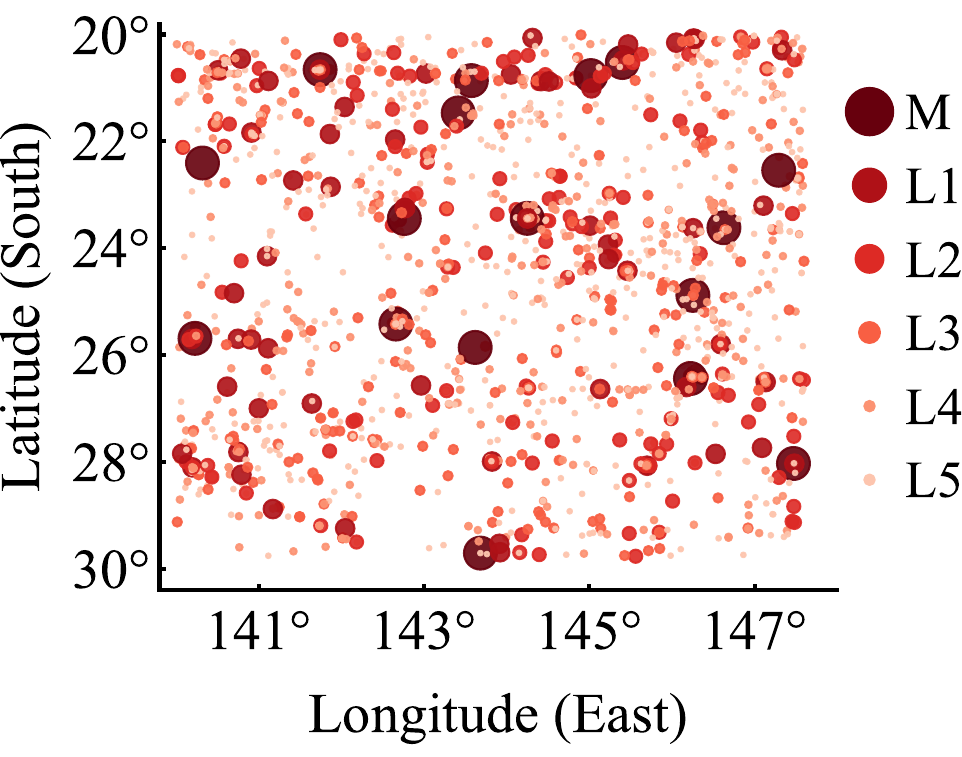}
   \label{fig:dot}}
     \vspace{0.1in}
  \caption{Totoro$^+$ excels at scalability by fairly distributing many dynamic-structured dataflow trees across large-scale edge topologies.}
  \vspace{0.01in}
\end{figure*}

\begin{figure*}[t!]
  \centering  
  \captionsetup[subfigure]{width=0.23\textwidth}
  \subfloat[Model dissemination time.]{\includegraphics[width=0.24\textwidth]{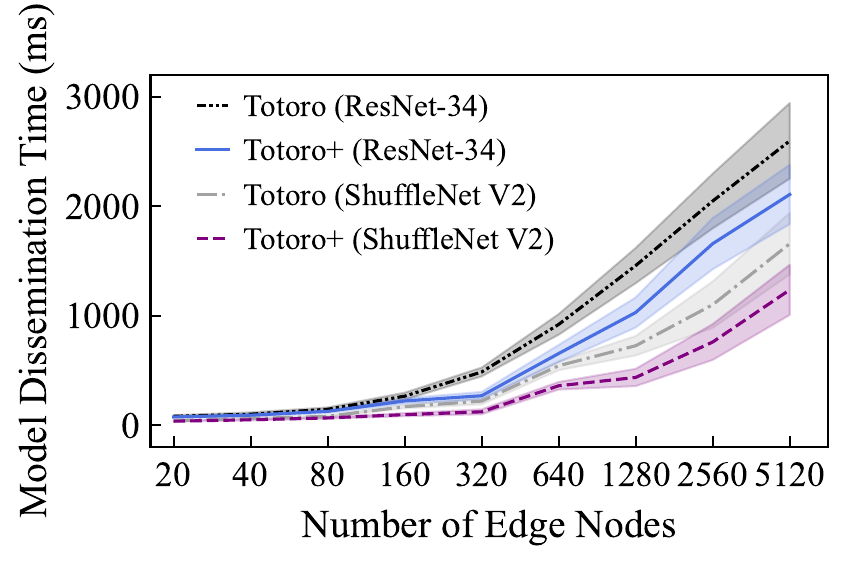}
  \label{fig:model_dissemination}}
 \hfill
  \subfloat[Gradient aggregation time.]{\includegraphics[width=0.24\textwidth]{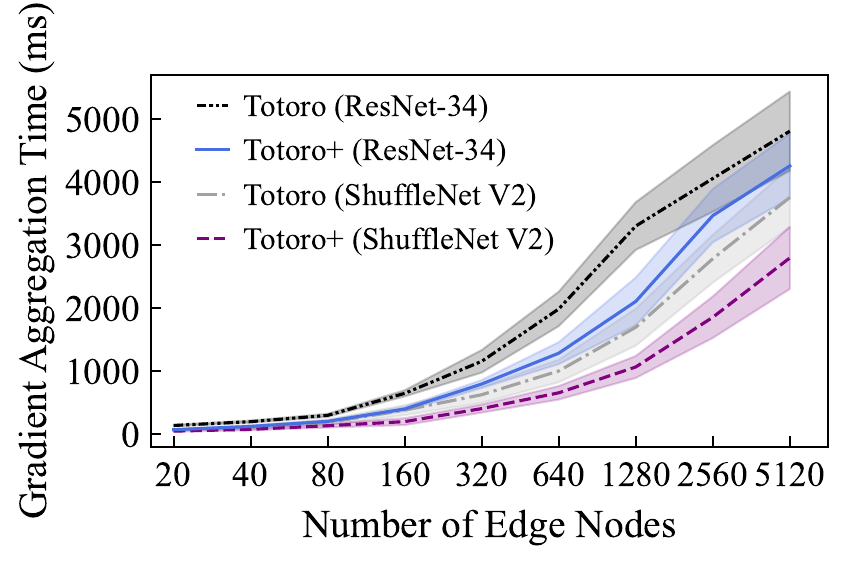}
  \label{fig:gradient_aggregation}}
 \hfill
 \subfloat[Model dissemination time of different fanouts.]{\includegraphics[width=0.24\textwidth]{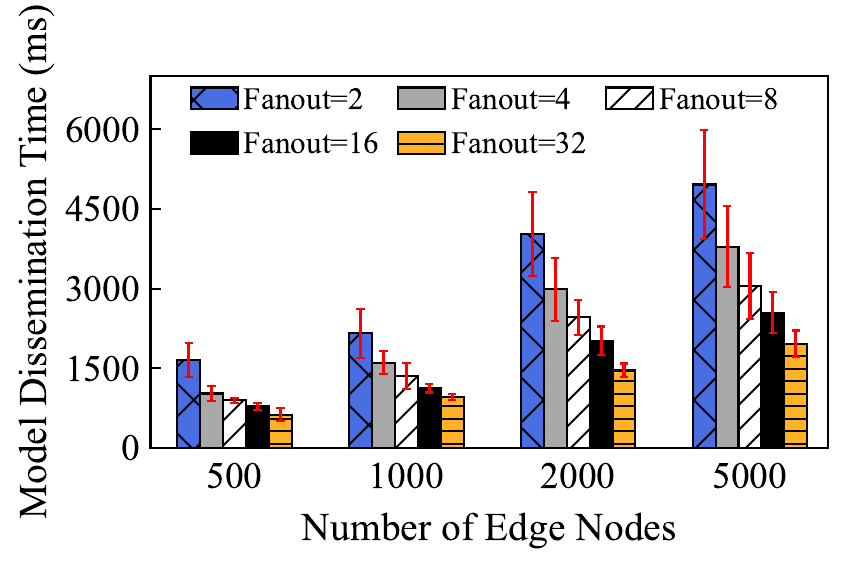}
  \label{fig:fanout_dissemination}}
  \hfill
     \subfloat[Gradient Aggregation time of different fanouts.]{\includegraphics[width=0.24\textwidth]{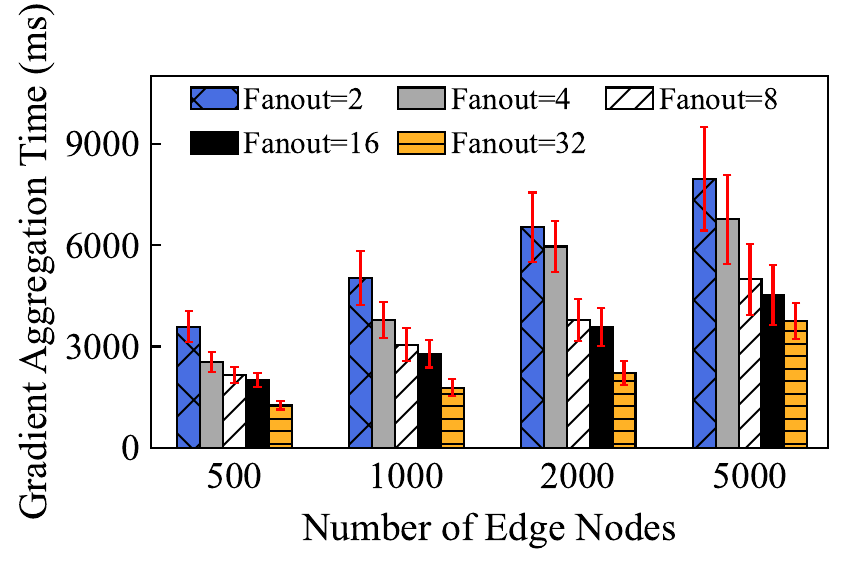}
  \label{fig:fanout_aggregation}}
    \vspace{0.12in}
  \caption{Totoro$^+$ scales with an exponentially increasing number of edge nodes for model dissemination and gradient aggregation.}
  \vspace{-0.05in}
\end{figure*}

\begin{figure}[t!]
  \centering
    \subfloat[TCP traffic cost.]{\includegraphics[width=0.235\textwidth]{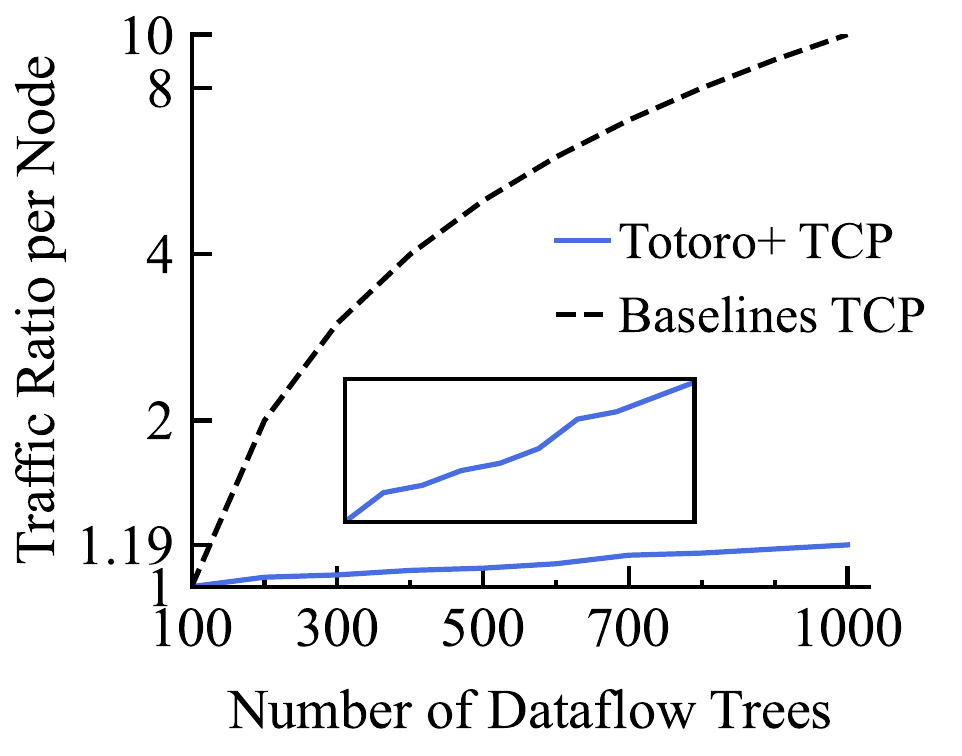}
    \label{fig:tcp}}
    \subfloat[UDP traffic cost.]{\includegraphics[width=0.235\textwidth]{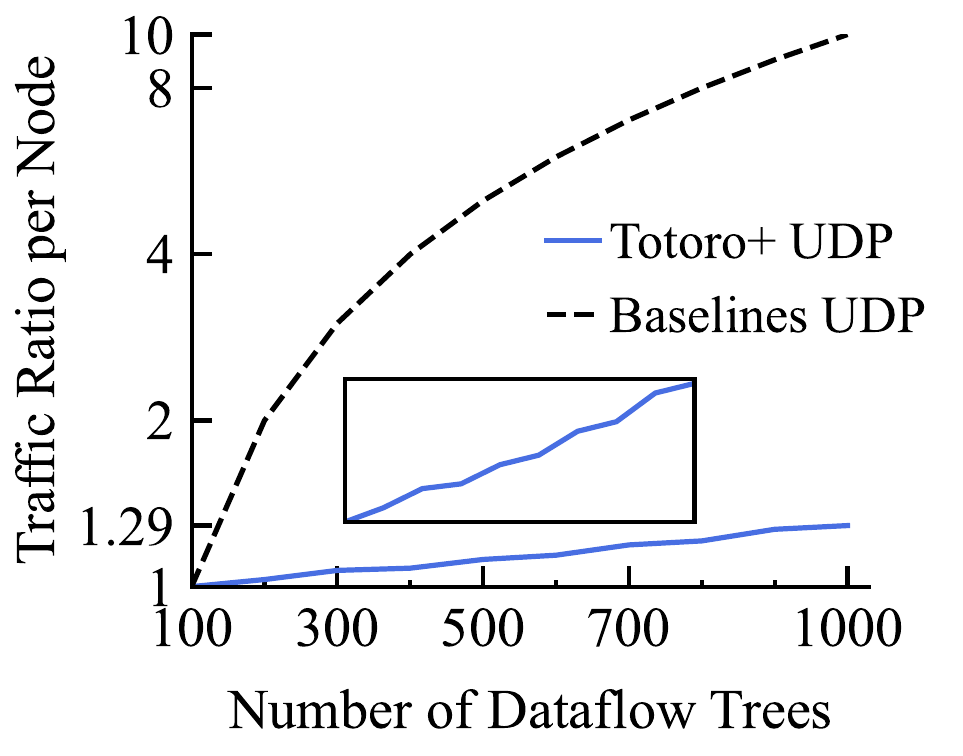}
    \label{fig:udp}}
    \vspace{0.1in}
    \caption{Totoro$^+$ reduces communication cost.}
    \vspace{-0.05in}
  \label{fig:eval_traffic}
\end{figure}

\section{Evaluation}
\label{sec:eval}

We evaluate Totoro$^+$’s performance in a real-world distributed environment that includes 500 Amazon EC2 nodes and uses real-world computer vision (CV) and natural language processing (NLP) datasets at different scales. We have the following key results.

\begin{itemize}[leftmargin=*, topsep=0.5ex]

\item Totoro$^+$ scales the number of masters to handle a varying workload (from 125 to 2000 concurrently running FL applications) in real-world edge topologies (Section~\ref{subsec:eval_scalability}).

\item Totoro$^+$ achieves $\mathcal{O}(\log N)$ hops for model dissemination and gradient aggregation (Section~\ref{subsec:eval_broadcast}).

\item Totoro$^+$ outperforms state-of-the-art FL systems (OpenFL~\cite{openfl} and FedScale~\cite{fedscale}) by speeding up the training time 1.2$\times$-14.0$\times$ to reach the equivalent model accuracy (Section~\ref{subsec:eval_accuracy}).

\item Totoro$^+$ efficiently adapts to edge networks with constrained bandwidth capacity (Section~\ref{subsec:adaptivity_analysis}).

\item Totoro$^+$ recovers from node failures to adapt to node churns (Section~\ref{subsec:failure_recovery}).

\end{itemize}

\subsection{Methodology}

\noindentparagraph{Experimental setup.} Totoro$^+$ is designed to operate in large deployments with millions of edge nodes. However, such a deployment is prohibitively expensive and impractical in an academic environment. As such, we resort to emulating a real-world edge setting on 500 AWS EC2 \texttt{t2.medium} nodes, each of which has 2 vCPUs and 4GB of RAM, and 100 GB of disk: (1) we use one JVM to represent one edge node and emulate up to $100k$ edge nodes on the testbed; (2) we divide the $100k$ edge nodes into geo-distributed zones based on the real-world EUA dataset~\cite{edge_network_dataset}, which consists of 95,271 edge nodes distributed across 12 Australian states and regions. To emphasize the system-level benefits of Totoro$^+$, we conduct experiments using homogeneous nodes in the main paper. 
{\color{dv}
Experimental results under heterogeneous node settings are reported in Appendix~\ref{app:adap_edge_hetero_resource}.
}

\noindentparagraph{Baselines.} We use OpenFL (v.1.3)~\cite{openfl} and FedScale (v.0.5)~\cite{fedscale} as the baseline. FedScale is a scalable and extensible open-source FL system and benchmarking suite developed by Symbiotic Lab~\cite{SymbioticLab}, which provides high-level and flexible API to implement, evaluate, and deploy FL algorithms easily in both standalone (single CPU/GPU) and distributed (multiple machines) settings. 
OpenFL is an open-source framework developed by Intel that runs FL in a single-machine setting.

\noindentparagraph{Parameters.} Totoro$^+$’s Pastry DHT is configured with a leaf set of 24, max open sockets of 5000, and a transport buffer size of 6 MB. We configure different fanouts 8 ($2^3$), 16 ($2^4$), and 32 ($2^5$) for Totoro$^+$'s dataflow trees by changing the DHT routing table base bit values to 3, 4, and 5, respectively. The minibatch size of each node is 20 in image classification and speech recognition tasks. The initial learning rate for the ShuffleNet V2 model is 0.05 and 0.1 for the ResNet-34 model. 

\noindentparagraph{Metrics.}
We focus on Totoro$^+$’s scalability, adaptivity, and FL effectiveness. To evaluate scalability, we measure how numerous applications’ dataflow trees are distributed over large-scale edge topologies. We also measure how Totoro$^+$ scales with the number of nodes in terms of \emph{model broadcast time} and \emph{gradient aggregation time}. To evaluate the FL effectiveness, we measure \emph{time-to-accuracy} performance, which is the duration of model training tasks on the testing set to achieve the target accuracy. 
To evaluate adaptivity, we measure 
\emph{cumulative packet latency},
\emph{Nash regret}, 
\emph{node selection frequencies}, 
and \emph{failure recovery time}. We also measure Totoro$^+$'s \emph{runtime overhead}.

\subsection{Scalability analysis}
\label{subsec:eval_scalability}

Figure~\ref{fig:edge_topology} shows the real-world edge zones generated from the EUA dataset~\cite{edge_network_dataset}. Australian Communications and Media Authority publishes the EUA dataset~\cite{edge_network_dataset}, which contains the geographical locations of 95,271 cellular base stations in 12 states of Australia (\texttt{ACT}: 931, \texttt{ANT}: 15, \texttt{EXT}: 8, \texttt{ISL}: 36, \texttt{NSW}: 24574, \texttt{NT}: 3137, \texttt{QLD}: 21576, \texttt{SA}: 7682, \texttt{TAS}: 3213, \texttt{VIC}: 18163, \texttt{WA}: 15933, \texttt{WLD}: 3). We estimate the maximum round-trip time based on the distances between nodes in the dataset, and then use Ratnasamy and Shenker's distributed binning algorithm~\cite{distributed_binning} to divide them into different zones. 

\noindentparagraph{Totoro$^+$ fairly distributes masters.} Figure~\ref{fig:cdf} shows the normal probability plot of the number of masters mapped on each node in a 1000-node edge zone under stress testing. 
The results show that when we create a large number of dataflow trees like 500, 99.5\% of the nodes are the roots of 3 trees or less. The results illustrate a good load balance among participating devices when performing a large number of FL applications' tasks simultaneously.
Figure~\ref{fig:3d_plot} shows the distribution of Totoro$^+$’s masters over different edge zones that have different workloads. We assume densely populated topologies have heavy workloads and sparsely populated topologies have light workloads. The results show that Totoro$^+$ automatically scales the number of masters to handle the varying workload.

\begin{table*}[t!]
\footnotesize
\setlength{\belowcaptionskip}{-5pt} 
\renewcommand\arraystretch{1.25}
\begin{tabular}{ccccccccccc}
\hline
\multirow{2}{*}{Task} & \multirow{2}{*}{Dataset} & \multirow{2}{*}{\begin{tabular}[c]{@{}c@{}}Accuracy\\ Target\end{tabular}} & \multirow{2}{*}{Model} & \multirow{2}{*}{\begin{tabular}[c]{@{}c@{}}Number of \\ Applications\end{tabular}} & \multicolumn{3}{c}{Speedup for OpenFL \cite{openfl}} & \multicolumn{3}{c}{Speedup for FedScale \cite{fedscale}} \\ \cline{6-11}
 &  &  &  &  & Fan.=8 & Fan.=16 & Fan.=32 & Fan.=8 & Fan.=16 & Fan.=32 \\ \hline
\multirow{3}{*}{\begin{tabular}[c]{@{}c@{}}Speech\\ Recognition\end{tabular}} & \multirow{3}{*}{\begin{tabular}[c]{@{}c@{}}Google\\ Speech~\cite{google_speech_dataset}\end{tabular}} & \multirow{3}{*}{53.0\%} & \multirow{3}{*}{ResNet-34~\cite{residualnet}} & 5 & 3.7$\times$ & 3.1$\times$ & 3.5$\times$ & 3.6$\times$ & 3.0$\times$ & 3.4$\times$ \\ \cline{5-11} 
 &  &  &  & 10 & 6.4$\times$ & 6.2$\times$ & 6.9$\times$ & 6.2$\times$ & 6.0$\times$ & 6.7$\times$ \\ \cline{5-11} 
 &  &  &  & 20 & 13.0$\times$ & 11.8$\times$ & 14.0$\times$ & 12.4$\times$ & 11.3$\times$ & 13.5$\times$ \\ \hline
\multirow{3}{*}{\begin{tabular}[c]{@{}c@{}}Image\\ Classification\end{tabular}} & \multirow{3}{*}{FEMNIST~\cite{femnist_dataset}} & \multirow{3}{*}{75.5\%} & \multirow{3}{*}{ShuffleNet V2~\cite{shufflenet_v2}} & 5 & 1.2$\times$ & 1.4$\times$ & 3.1$\times$ & 1.4$\times$ & 1.6$\times$ & 3.4$\times$ \\ \cline{5-11} 
 &  &  &  & 10 & 2.4$\times$ & 3.2$\times$ & 5.6$\times$ & 2.7$\times$ & 3.5$\times$ & 6.1$\times$ \\ \cline{5-11} 
 &  &  &  & 20 & 5.0$\times$ & 5.5$\times$ & 10.3$\times$ & 5.6$\times$ & 6.1$\times$ & 11.5$\times$ \\ \hline
\multicolumn{1}{l}{} & \multicolumn{1}{l}{} & \multicolumn{1}{l}{} & \multicolumn{1}{l}{} & \multicolumn{1}{l}{} & \multicolumn{1}{l}{} & \multicolumn{1}{l}{} & \multicolumn{1}{l}{} & \multicolumn{1}{l}{} & \multicolumn{1}{l}{} & \multicolumn{1}{l}{}
\end{tabular}
\vspace{-0.15in}
\caption{Summary of time-to-accuracy comparison of Totoro$^+$, OpenFL, and FedScale. We set three different fanouts for Totoro$^+$’s dataflow tree: 8, 16, and 32. Totoro$^+$ outperforms state-of-the-art FL systems by speeding up the training time at different scales. The speedup gap increases as the number of concurrently running applications’ tasks increases. }
  \label{tab:time2accuracy}
\end{table*}

\begin{figure*}[t!]
  \centering
  \captionsetup[subfigure]{width=4.1cm}
  \subfloat[1 application's model]{\includegraphics[width=0.24\textwidth]{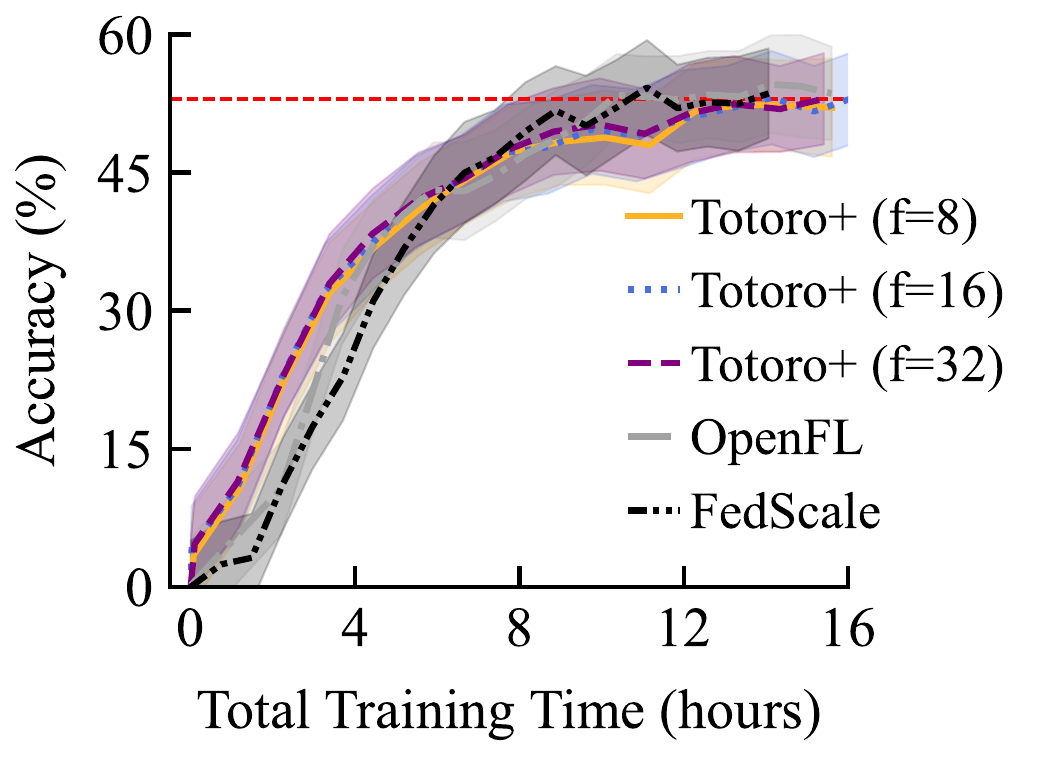}
  \label{fig:speech_1}}
\hfill
  \subfloat[5 models trained simultaneously]{\includegraphics[width=0.24\textwidth]{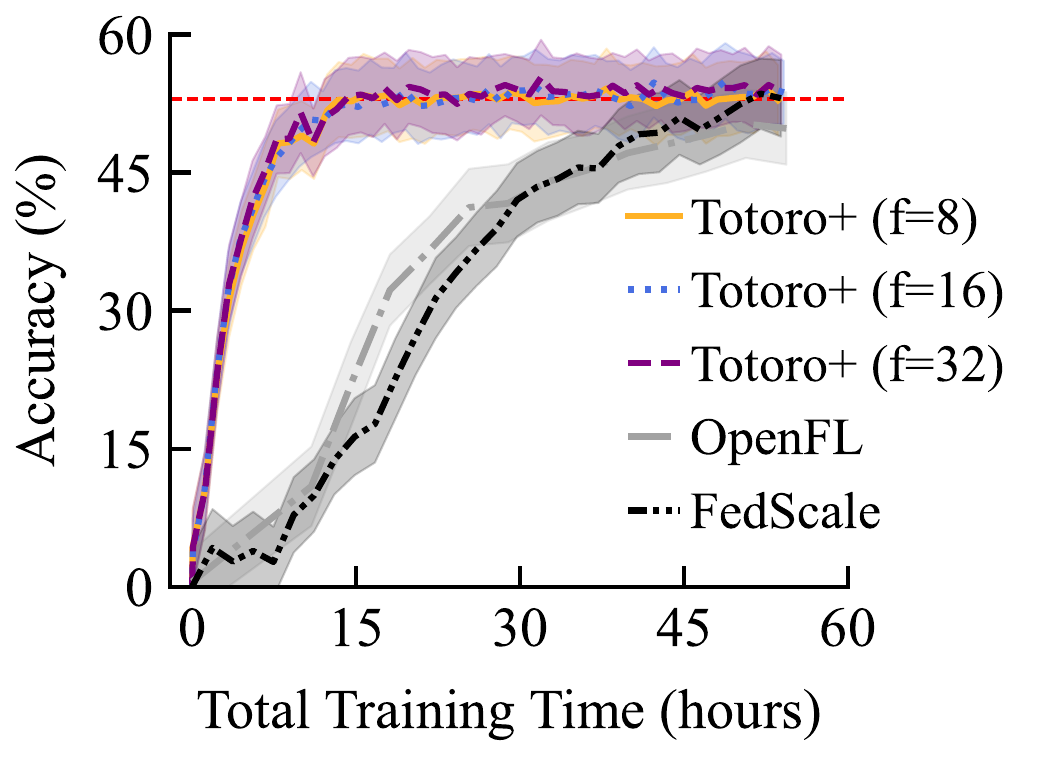}
  \label{fig:speech_5}}
 \hfill
   \subfloat[10 models trained simultaneously]{\includegraphics[width=0.24\textwidth]{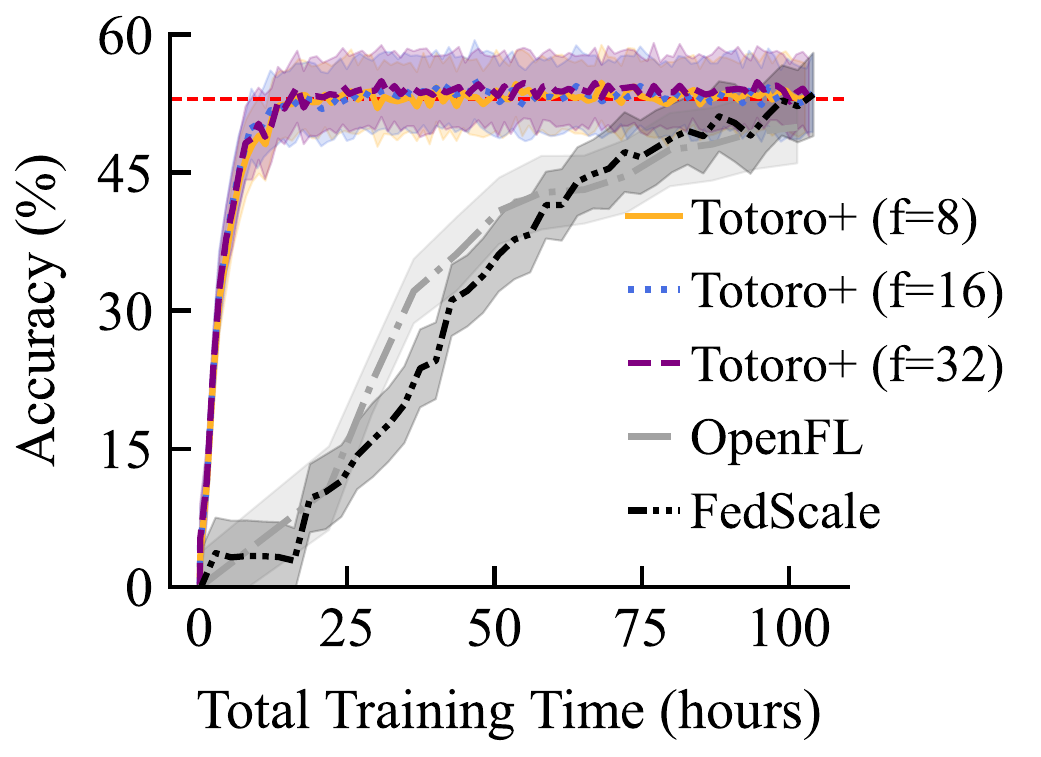}
   \label{fig:speech_10}}
 \hfill
     \subfloat[20 models trained simultaneously]{\includegraphics[width=0.24\textwidth]{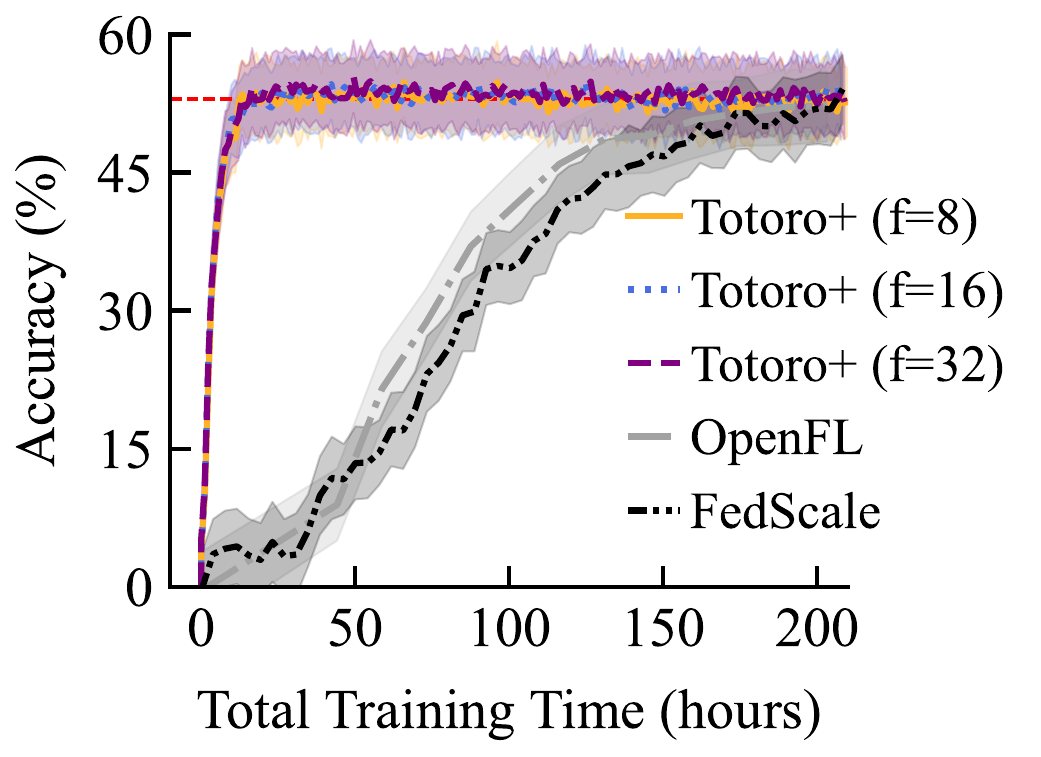}
   \label{fig:speech_20}}
     \vspace{0.1in}
  \caption{Time-to-accuracy comparison of Totoro$^+$, OpenFL, and FedScale. When 5$\sim$20 applications' models are simultaneously trained, Totoro$^+$ speeds up the total training time 3.0$\times$-14.0$\times$ to reach the equivalent model accuracy on the middle-scale Google Speech dataset.}
  \vspace{-0.1in}
  \label{fig:GoogleSpeech}
\end{figure*}

\begin{figure*}[t!]
  \centering
  \vspace{0.15in}
  \captionsetup[subfigure]{width=4.1cm}
  \subfloat[1 application's model]{\includegraphics[width=0.24\textwidth]{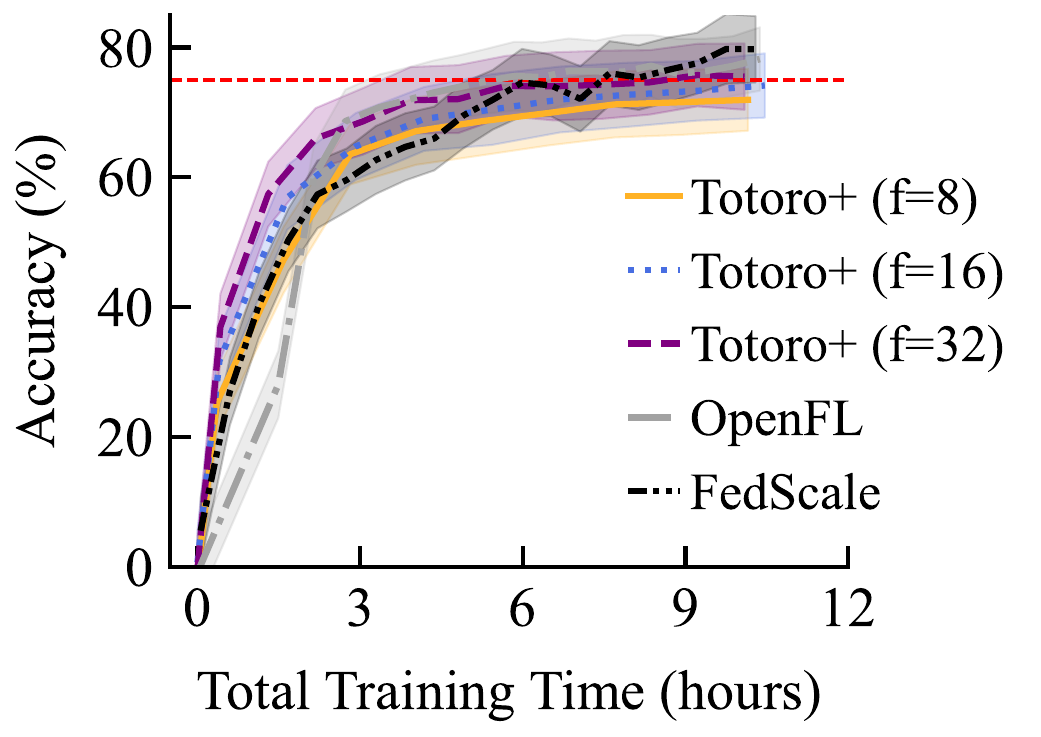}
  \label{fig:femnist_1}}
\hfill
  \subfloat[5 models trained simultaneously]{\includegraphics[width=0.24\textwidth]{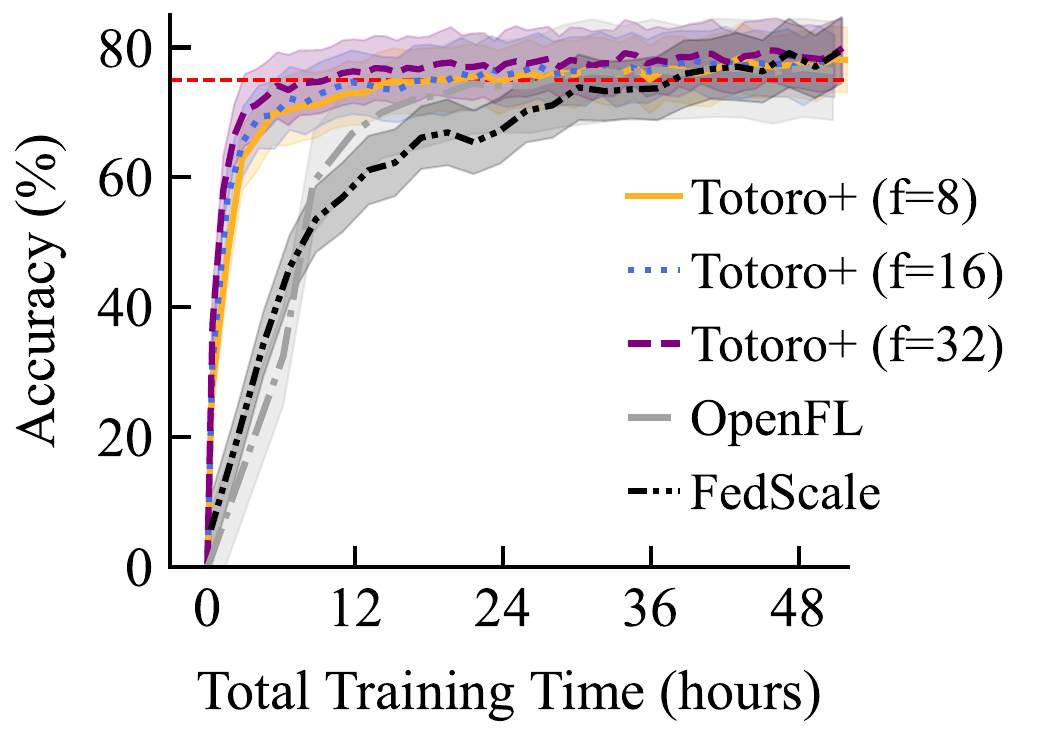}
  \label{fig:femnist_5}}
 \hfill
   \subfloat[10 models trained simultaneously]{\includegraphics[width=0.24\textwidth]{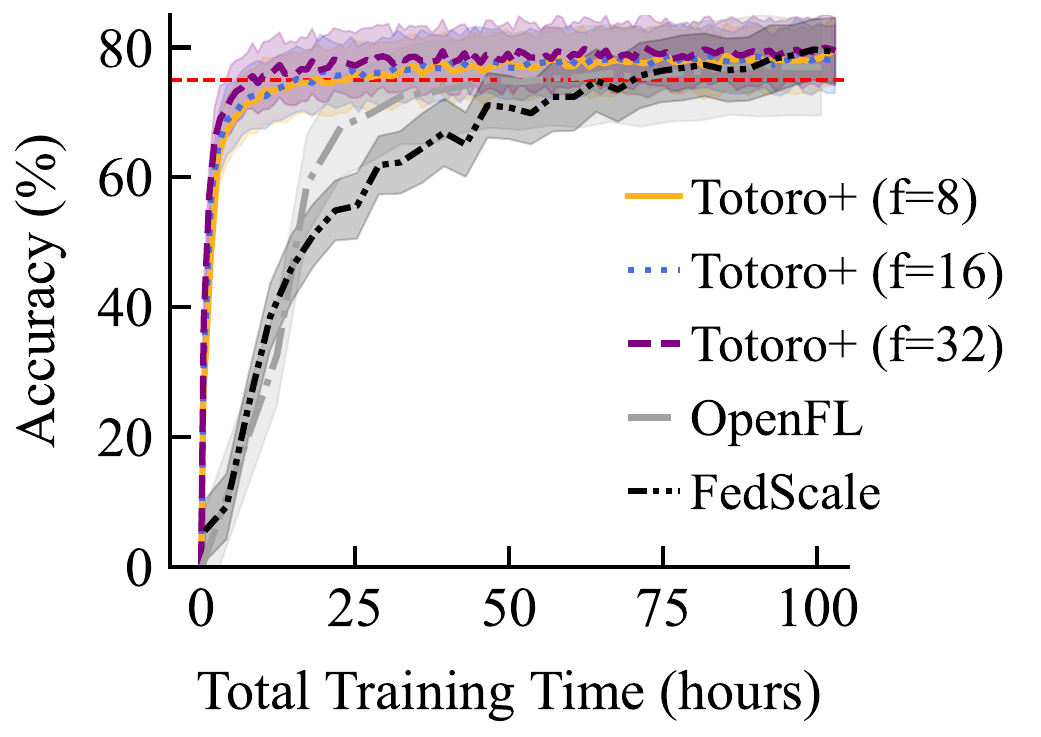}
   \label{fig:femnist_10}}
 \hfill
     \subfloat[20 models trained simultaneously]{\includegraphics[width=0.24\textwidth]{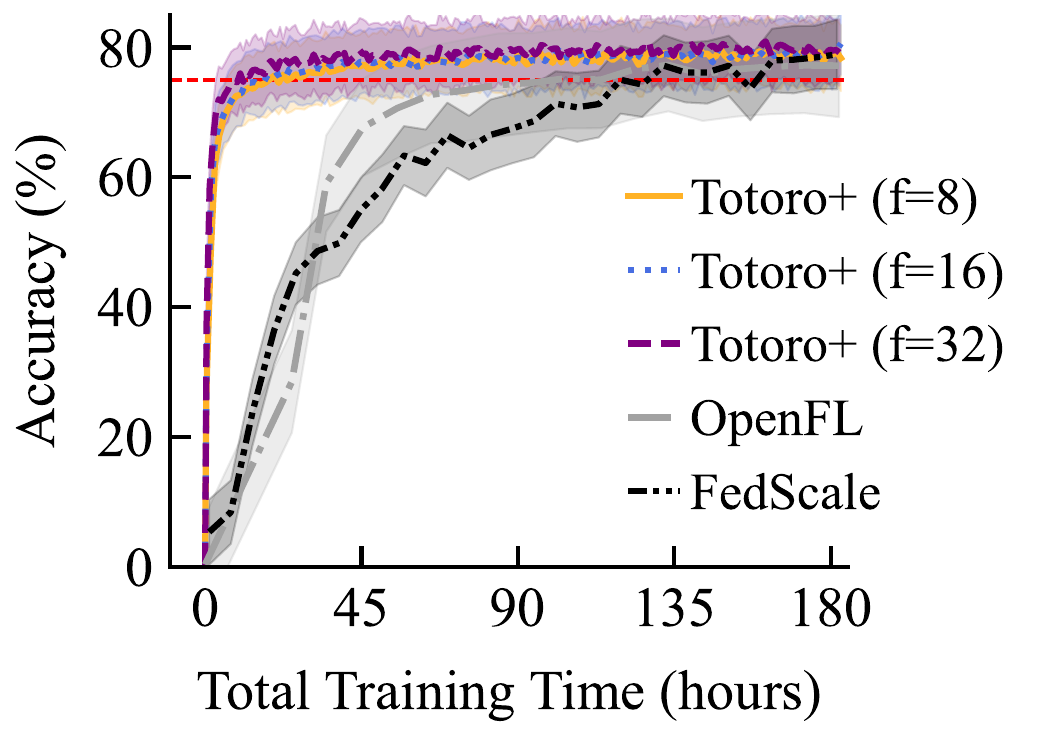}
   \label{fig:femnist_20}}
     \vspace{0.1in}
  \caption{Time-to-accuracy comparison of Totoro$^+$, OpenFL, and FedScale. When 5$\sim$20 applications' models are simultaneously trained, Totoro$^+$ speeds up the total training time 1.2$\times$-11.5$\times$ to reach the equivalent model accuracy on the large-scale FEMNIST dataset.}
  \label{fig:FEMNIST}
\end{figure*}

\noindentparagraph{Totoro$^+$ excels at load balancing.} Figure~\ref{fig:dot} shows the distribution of all branches in 17 Totoro$^+$ dataflow trees on 1946 edge nodes over the 3 most popular topologies. Each tree has a fanout of 8 ($2^3$) and a random number of depths (from level 1 to level 6). Different colors represent different levels of nodes in a tree, with the darkest color representing the root node. The results show that these dataflow trees are well balanced across different topologies, not only the root nodes but also the forwarder nodes and the leaf nodes, demonstrating Totoro$^+$’s attractive scalability and load balancing properties.

\subsection{Model dissemination and gradient aggregation}
\label{subsec:eval_broadcast}

\noindentparagraph{Totoro$^+$ scales with \#nodes.}
Figure \ref{fig:model_dissemination} and Figure~\ref{fig:gradient_aggregation} show Totoro$^+$'s model dissemination and gradient aggregation times for an exponentially increasing number of edge nodes in a single training tree.  
When the number of nodes grows \emph{exponentially} (from 20 to 5120), the dissemination time and aggregation time only increase \emph{linearly}. This is because the dissemination time and aggregation time are limited by the dataflow tree depth ($\mathcal{O}(\log N)$) by using the DHT-based P2P overlay. Therefore, when the number of edge devices grows to the scale of millions or even billions, Totoro$^+$ guarantees that only a few extra hops are needed for them to receive updated FL models or send updated gradients.

\noindentparagraph{Totoro$^+$ offers flexible tree fanouts.}
Figure~\ref{fig:fanout_dissemination} and Figure~\ref{fig:fanout_aggregation} show the model dissemination time and gradient aggregation time of different tree fanouts for the ResNet-34 model. We can observe that a larger fanout leads to less model dissemination and gradient aggregation times because the dataflow trees become less deep. However, if an internal node fails or leaves, the new one needs to rebuild many connections between the child nodes and the parent nodes to rebuild the tree, and the new node may become an I/O or communication bottleneck.

\noindentparagraph{Totoro$^+$ reduces communication cost.}
Figure~\ref{fig:eval_traffic} shows the comparison of the traffic per node of Totoro$^+$ and the baseline FL systems. Communication is a critical bottleneck in FL systems. 
We can observe that the additional network traffic imposed by Totoro$^+$ is small. 
The network traffic is increased by only 1.19$\times$ for TCP and 1.29$\times$ for UDP when the number of dataflow trees is increased by 10$\times$. This is because when a new training tree is created, it merely routes \texttt{JOIN} messages toward the root of the tree using the overlay, adding overlay links to reconstruct the tree without establishing a new connection. 
Then the overhead for creating new dataflow trees can be amortized over the overhead of the overlay.

\subsection{Federated learning effectiveness}
\label{subsec:eval_accuracy}
\noindent In this section, we evaluate Totoro$^+$’s performance on model training when an increasing number of models (1 to 20) are simultaneously trained on large-scale edge topologies, and compare it with OpenFL~\cite{openfl} and FedScale~\cite{fedscale}. 
Both FedScale and OpenFL rely on a server-client architecture and leverage a logically central coordinator to spawn aggregators and orchestrate many distributed clients to collaboratively train a model. For both the Google Speech~\cite{google_speech_dataset} and FEMNIST~\cite{femnist_dataset} datasets, we evenly partitioned the data across the edge nodes such that each node contains samples from all classes. This IID setting allows us to isolate and compare pure system-level performance differences among Totoro$^+$, OpenFL, and FedScale under the same controlled conditions. All compared systems follow the same data partitioning strategy.

\noindentparagraph{Totoro$^+$ reduces the total model training time.} Table~\ref{tab:time2accuracy} summarizes the results of time-to-accuracy comparison of Totoro$^+$, OpenFL and FedScale. When 5$\sim$20 models are simultaneously trained on the same platform, we notice that Totoro$^+$ speeds up the total training time 3.0$\times$-14.0$\times$ to reach the equivalent model accuracy on the middle-scale Google Speech dataset; speedup on the large-scale FEMNIST dataset is 1.2$\times$-11.5$\times$. The speedup gap increases as the number of concurrently running applications' tasks increases. This is because both OpenFL and FedScale rely on a server-client architecture where a logically central coordinator orchestrates many distributed clients to collaboratively train a model. When there are a massive number of models that need to be trained simultaneously, the central coordinator needs to handle them one by one on a first-come, first-served basis, which causes large queuing delays. By contrast, Totoro$^+$’s distributed masters can train many models in parallel, thus precluding them from being stuck in a central coordinator. {\color{dv}To further evaluate the effectiveness of Totoro$^+$, we compare it with decentralized federated learning baselines on additional models and datasets. The results are presented in Appendix~\ref{app:decsys_comparison}.}



\begin{figure}
    \centering
    \includegraphics[width=.6\linewidth]{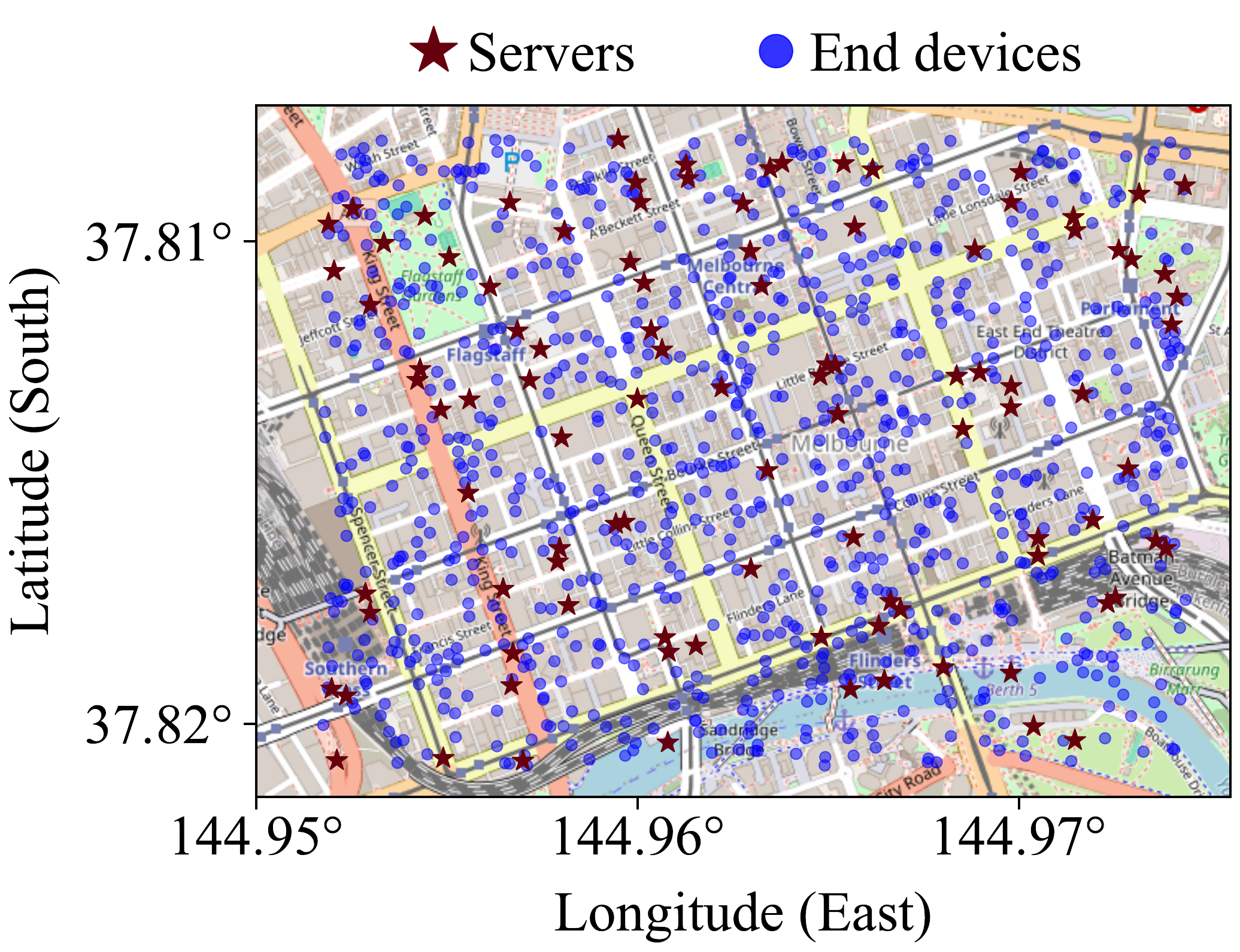} 
    \caption{100 edge servers and 900 end devices in Melbourne, Australia, from the EUA dataset~\cite{edge_network_dataset}}
    \label{fig:node_distribution_melbourne}
    \vspace{-0.05in}
\end{figure}

\begin{figure*}[t!]    
    \includegraphics[width=0.5\linewidth]{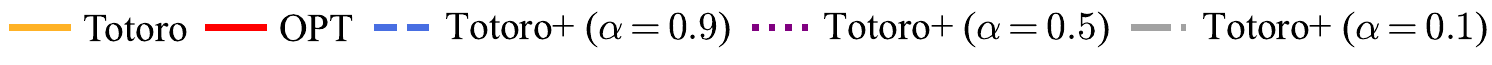}\hspace{-20.in}\\    
    \centering
    \begin{minipage}{.23\textwidth}    
        \centering
        \captionsetup{width=3.8cm} 
        \includegraphics[width=\linewidth]{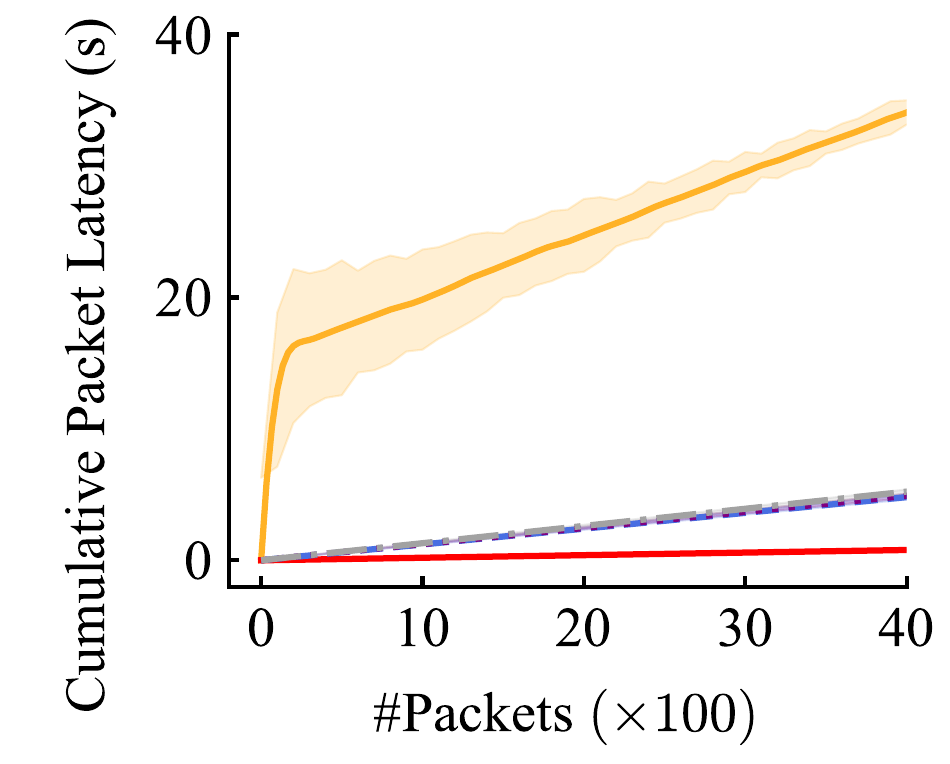}
        \caption{Cumulative packet latency of different algorithms.}
        \label{fig:cum_packet_delays}
    \end{minipage}\hfill
    \begin{minipage}{.23\textwidth}
        \centering
        \captionsetup{width=3.8cm} 
        \includegraphics[width=\linewidth]{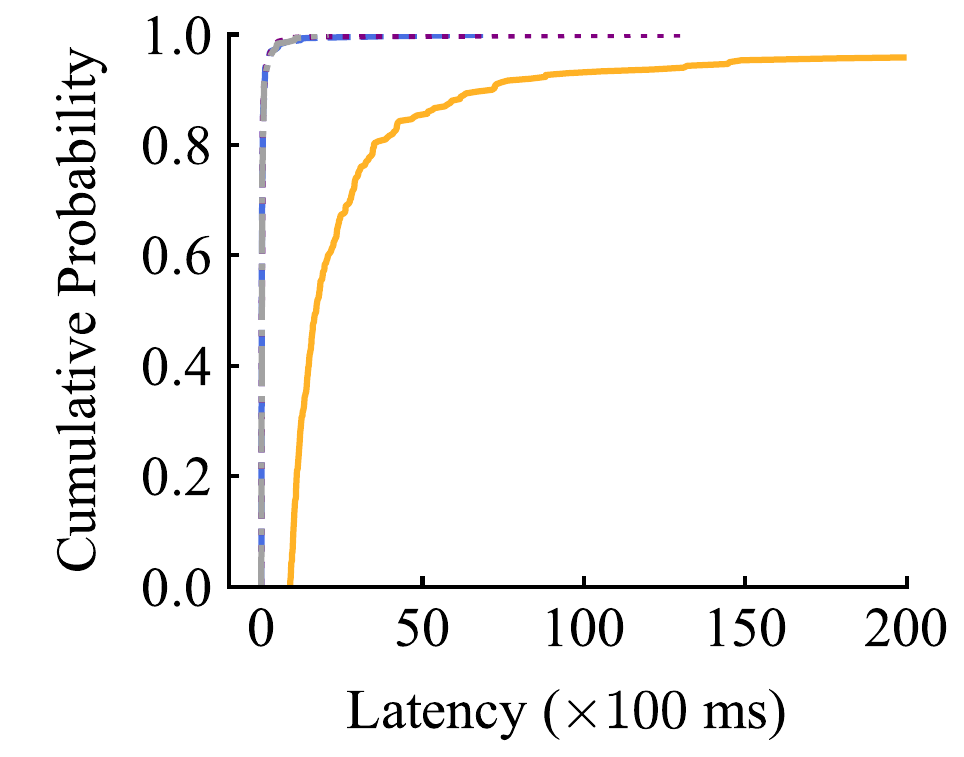} 
        \caption{Cumulative probabilities of different algorithms over latency.}
        \label{fig:latency_cdf}
    \end{minipage}\hfill    
    \begin{minipage}{.23\textwidth}
    \vspace{-0.15in}
        \centering
        \captionsetup{width=3.8cm} 
        \includegraphics[width=\linewidth]{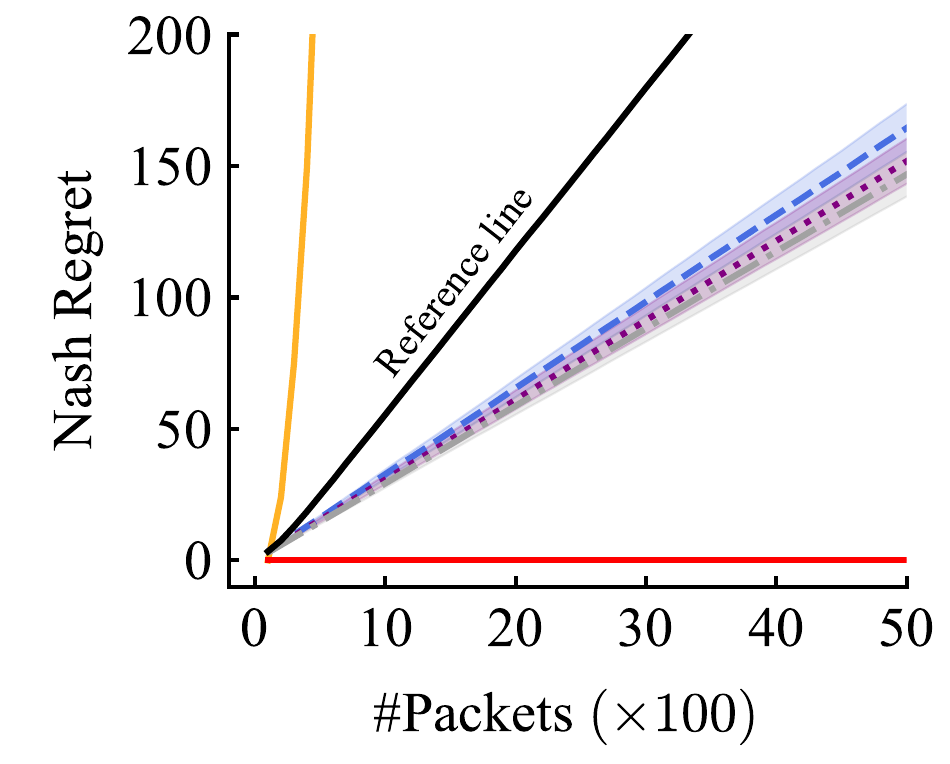}
        \caption{Nash regret of different algorithms.}
        \label{fig:nash_regret}
    \end{minipage}
    \begin{minipage}{.25\textwidth}
        \centering
        \captionsetup{width=3.8cm} 
        \includegraphics[width=\linewidth]{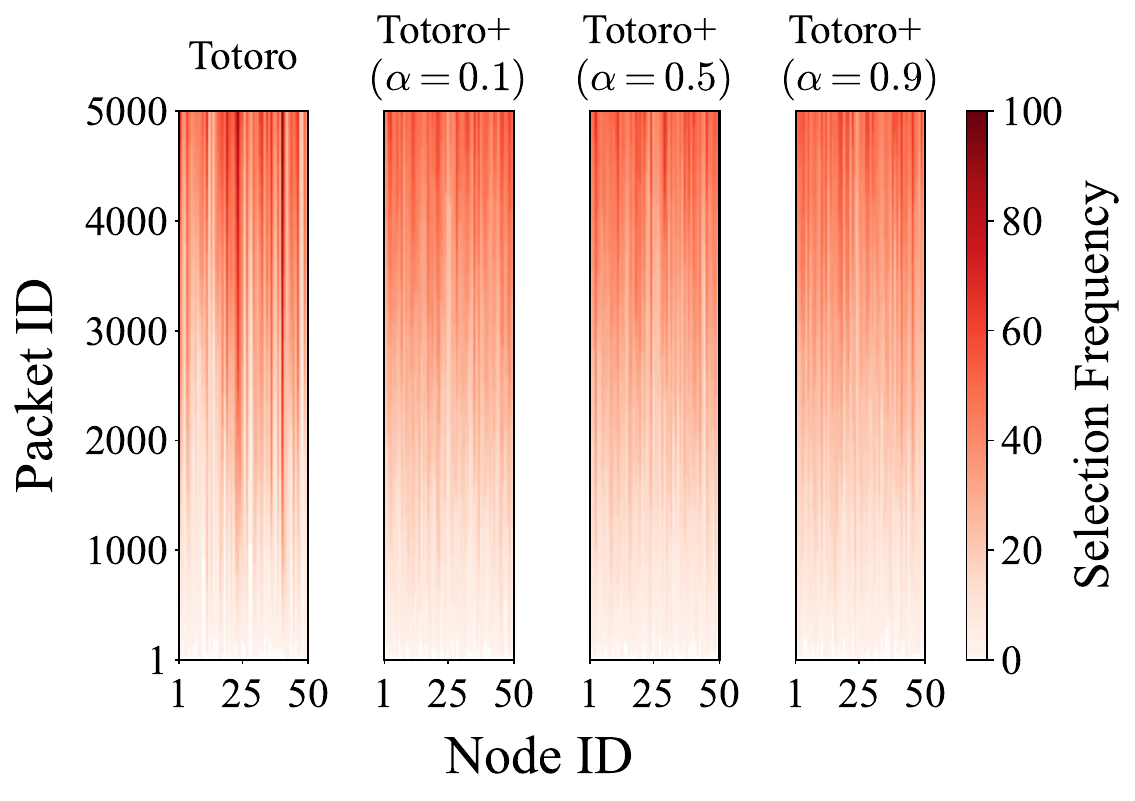} 
        \caption{Node selection frequencies generated by different algorithms.}
        \label{fig:selection_frequency}
    \end{minipage}
    \vspace{-0.05in}
\end{figure*}

\noindentparagraph{Totoro$^+$ scales with \#applications. } 
Figure~\ref{fig:GoogleSpeech} and Figure~\ref{fig:FEMNIST} show the time-to-accuracy performance comparison of Totoro$^+$, OpenFL, and FedScale.
The shaded regions represent the standard deviations, and 
the red dotted lines represent the target accuracies for Google Speech (53.0\%) and FEMNIST (75.5\%). The results show that Totoro$^+$ scales well with a large number of applications’ training tasks. Totoro$^+$ takes almost the same total training time to reach model convergence for training 1 model (15.41 hours), 5 models (15.43 hours), 10 models (15.44 hours), and 20 models (15.47 hours) when fanout is equal to 32. The rationale behind the results lies in that (1) Totoro$^+$ decomposes the conventional “\emph{1:n}” architecture into an “\emph{m:n}” architecture, which ensures that every peer can participate in the process of model training, gradients calculation, and aggregation; and (2) ideally, any peer in the system can act as a worker, a master, a forwarder, or any combination of the above. This helps avoid the scalability bottleneck caused by any single node, auto-scale as needed, balance the workload, and improve the scalability.

\subsection{Adaptivity analysis}
\label{subsec:adaptivity_analysis}


We evaluate Totoro$^+$'s game-theoretic distributed hop-by-hop routing algorithm using a real-world distribution of edge servers and end devices in the EUA dataset~\cite{edge_network_dataset}. We randomly extracted 900 end devices and 100 edge servers in Melbourne, Australia, from coordinate $-37.820$ to $-37.808$ in latitude and $144.955$ to $144.975$ in longitude, as shown in Figure~\ref{fig:node_distribution_melbourne}. 

These 1,000 nodes establish a dataflow tree, where the 900 end devices are worker nodes, an edge server serves as a master, and the remaining 99 edge servers are internal nodes. The worker and internal nodes run 
three algorithms to select the next hops for model updates.
\begin{itemize}
    \item Totoro$^+$: The game-theoretic distributed hop-by-hop routing algorithm in Algorithm~\ref{pseudo:our_algorithm}.
    \item Totoro~\cite{totoro}: The bandit-based distributed hop-by-hop algorithm that does not consider bandwidth capacities.
    \item OPT: The optimal algorithm that always selects the optimal next hop and considers bandwidth capacities.
\end{itemize}


We refer to~\cite{awstream, edge_network_bandwidth} to set the constrained bandwidth of each node in the range between 20Mbps and 100Mbps. Each worker and internal node has constrained bandwidth. The more packets it forwards simultaneously, the slower its data rate becomes. For example, if a node with 100Mbps bandwidth needs to forward model updates from four nodes, the data rate is $100/4 = 25$. Also, the higher the data rate, the higher the corresponding reward. 

\noindentparagraph{Totoro$^+$ adapts to edge networks with constrained bandwidth capacity.} Figure~\ref{fig:cum_packet_delays} shows the cumulative packet latency of three algorithms. We can see that Totoro$^+$ achieves a cumulative packet latency that grows slowly. In contrast, Bandit suffers a higher cumulative packet latency. This is because Totoro$^+$ considers that nodes' bandwidth decreases proportionally if they have to forward many packets simultaneously. Hence, Totoro$^+$ tends to distribute network traffic more evenly over all nodes.

Figure~\ref{fig:latency_cdf} shows cumulative probabilities of different algorithms over latency. 
Totoro$^+$ consistently outperforms Totoro in latency, regardless of the choice of $\alpha$, delivering nearly all responses under 2000 ms. The results explain why Totoro$^+$ can achieve lower cumulative packet latency in Figure~\ref{fig:cum_packet_delays}.


\begin{figure}[t!]
    \begin{minipage}{.23\textwidth}
        \centering
        \captionsetup{width=3.8cm} 
        \includegraphics[width=\linewidth]{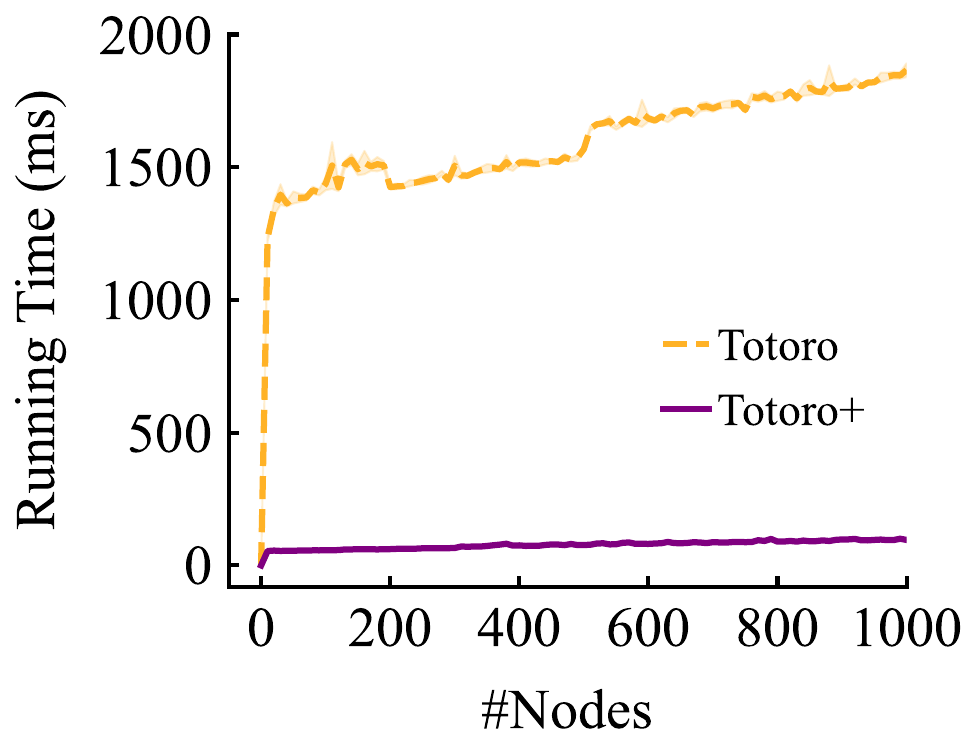}
        \caption{Running time comparison.}
        \label{fig:running_time}
    \end{minipage}\hfill
    \begin{minipage}{.225\textwidth}
        \centering
        \captionsetup{width=\textwidth} 
        \includegraphics[width=\linewidth]{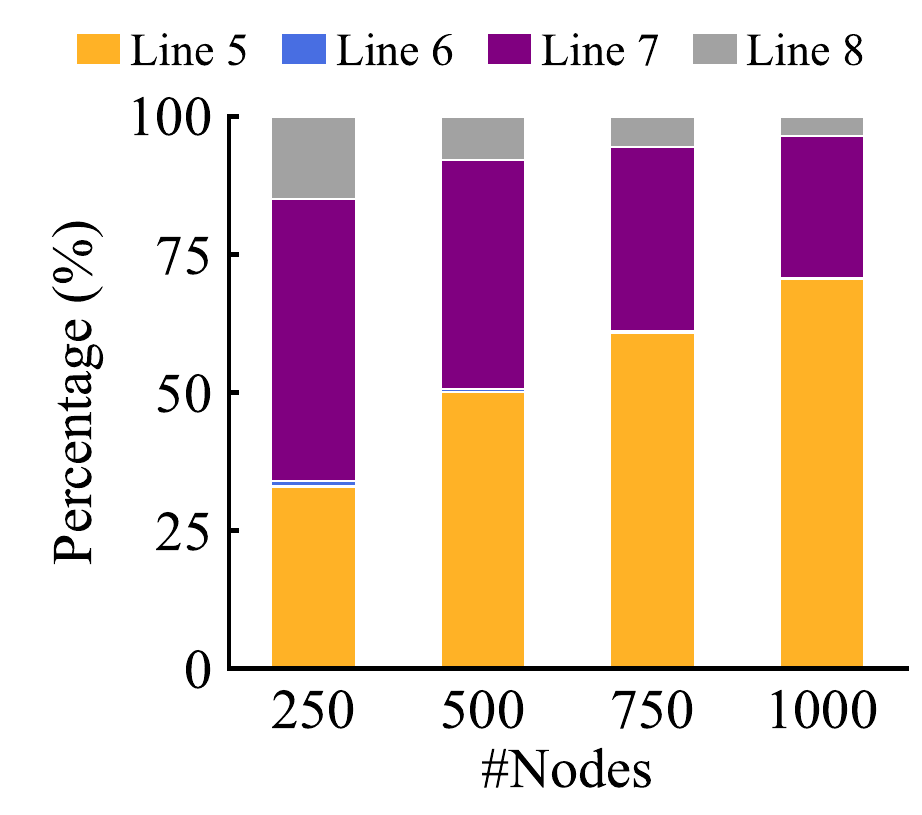} 
        \caption{Totoro$+$'s running time breakdown.}
        \label{fig:running_time_breakdown}
    \end{minipage}
   \vspace{-0.05in}
\end{figure}


\noindentparagraph{Totoro$^+$ achieves an $\epsilon$-approximate Nash equilibrium.}
Figure~\ref{fig:nash_regret} shows the Nash regret comparison of Totoro$^+$, OPT, and a reference line. We refer to the solution from~\cite{congestion_games} as the reference line, where the solution is proved to achieve sublinear Nash regret with the G-optimal design. The results show that Totoro$^+$ achieves a sublinear Nash regret closer to OPT. Therefore, according to Corollary~\ref{coro:epsilon_approximate}, Totoro$^+$ achieves an $\epsilon$-approximate Nash equilibrium. In contrast, Totoro does not consider bandwidth capacities but selects the optimal next hop so far. Therefore, when too many nodes select the same next hop, the overall communication time increases, which significantly increases Nash regret.

Figure~\ref{fig:selection_frequency} shows the node selection frequencies generated by different algorithms, in which each color grid shows the
frequencies of selecting the $x_{th}$ node for each packet. The $X$-axis represents the sequential order of next-hop nodes. The $Y$-axis represents the sequential order of packets. 
Although Totoro relies on a distributed hop-by-hop routing algorithm to send packets, it does not distribute packets evenly across all possible next hops. Instead, Totoro$^+$'s improved game-theoretic routing algorithm considers bandwidth capacities across nodes and distributes packets more evenly across all possible next hops, thereby leading to lower Nash regret.

\noindentparagraph{Totoro$^+$ forwards packets rapidly.}
Figure~\ref{fig:running_time} shows the running time comparison of Totoro and Totoro$^+$. We can see that Totoro$^+$ maintains consistently low running time ($\approx$50 ms) regardless of the number of nodes, whereas Totoro's running time sharply increases up to around 1500 ms when the number of nodes increases from 0 to ~100. This is because a node has to run the convex optimization to evaluate the empirical transmission cost for each possible hop. Instead, Totoro$^+$'s operations can be implemented by parallel matrix multiplications.

Figure~\ref{fig:running_time_breakdown} shows the Totoro$^+$'s running time breakdown. Lines 5, 6, 7, and 8 represent 4 lines in Algorithm~\ref{pseudo:our_algorithm}. We can see that the proportion of Line 5 increases significantly when the number of nodes increases. This is because Line 5 requires calculating determinant among all possible policies, which is consistent with the results in Theorem~\ref{theo:time_complexity}. 

\begin{figure}[t!]
    \centering    
    \begin{minipage}{0.235\textwidth}
        \vspace{-0.16in}
        \centering
        \captionsetup{width=\textwidth} 
        \includegraphics[width=\textwidth]{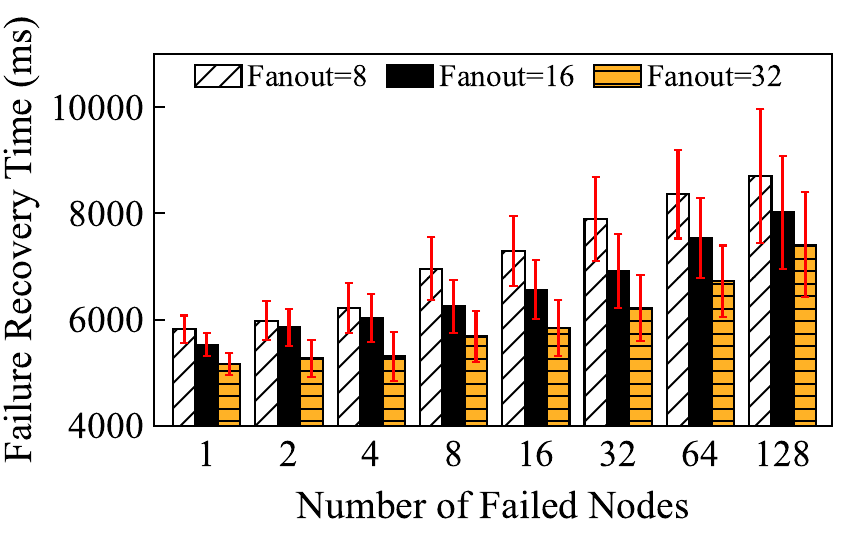} 
        \caption{Totoro$^+$ tolerates many simultaneous node failures in a tree.}
        \label{fig:failure_nodes}
    \end{minipage}\hfill
    \begin{minipage}{0.23\textwidth}
        \centering
        \includegraphics[width=\textwidth]{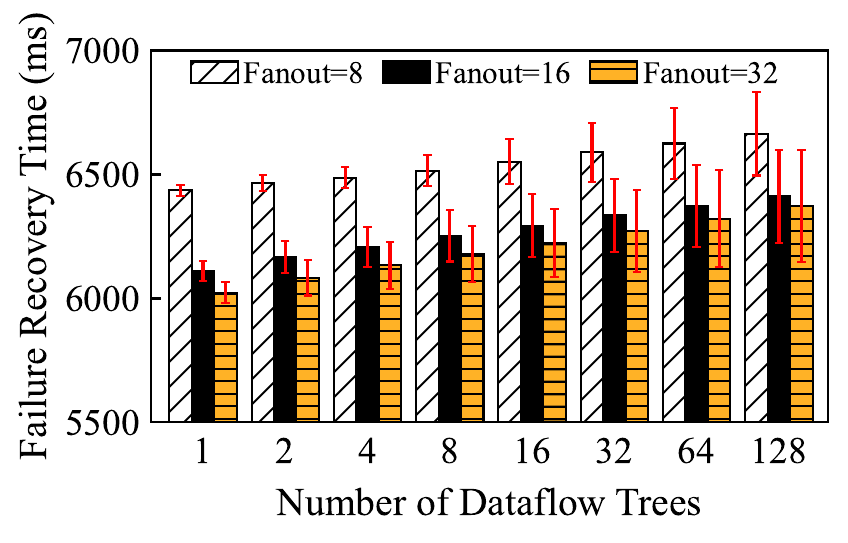}
        \caption{Totoro$^+$ achieves a stable recovery time for many simultaneous node failures in many trees.}
        \label{fig:failure_tree}
    \end{minipage}
    \vspace{-0.03in}
\end{figure}

\subsection{Failure recovery analysis}
\label{subsec:failure_recovery}
In this section, we evaluate how long it takes Totoro$^+$ to recover from node failures in a single and many dataflow trees.

\noindentparagraph{Totoro$^+$ tolerates simultaneous node failures in a dataflow tree.}
Figure~\ref{fig:failure_nodes} shows the failure recovery time for recovering a single dataflow tree with an exponentially increasing number of failed nodes. The tree has 1,000 nodes at the beginning, and then a varying number of random nodes fail or leave the tree simultaneously. The results show that the failure recovery time increases linearly when the number of failed nodes grows exponentially (from 1 to 128).
This is because only $\mathcal{O}(\log_{2^b} N)$ nodes are involved in recovering the tree. Even though the root fails, its immediate child nodes can locate a new root, and it can recover the state from the replicas.

\noindentparagraph{Totoro$^+$ achieves a stable recovery time for many simultaneous trees' failures.} Figure~\ref{fig:failure_tree} shows the failure recovery time for an exponentially increasing number of dataflow trees. Each tree has at most 1,000 nodes, and 5\% of nodes fail or leave simultaneously. The results show that Totoro$^+$ achieves a stable recovery time for many simultaneous trees' failures. This is because each failed node can be quickly detected and recovered by its neighbors through keep-alive messages without having to talk to a central coordinator, so many simultaneous failures can be repaired in parallel.

\begin{figure}[t!]
    \centering
  \includegraphics[width=0.6\linewidth]{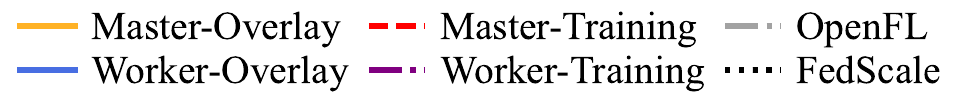}\\
  \centering
    \subfloat[The CPU overhead.]{\includegraphics[width=0.24\textwidth]{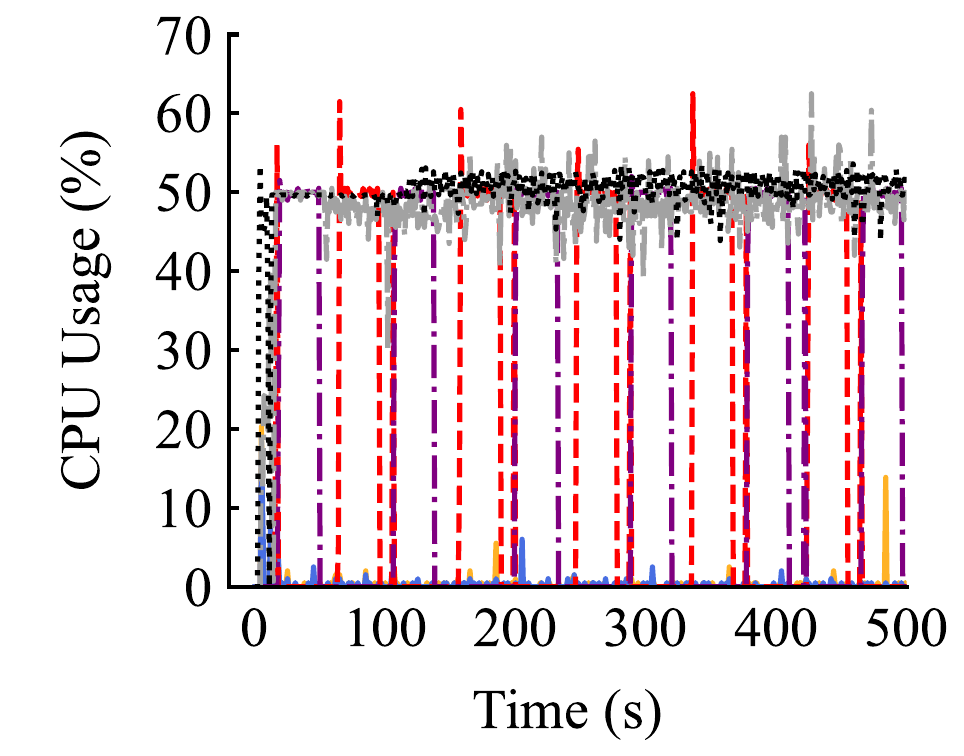}
    \label{fig:cpu_overhead}}
    \subfloat[The memory overhead.]{\includegraphics[width=0.24\textwidth]{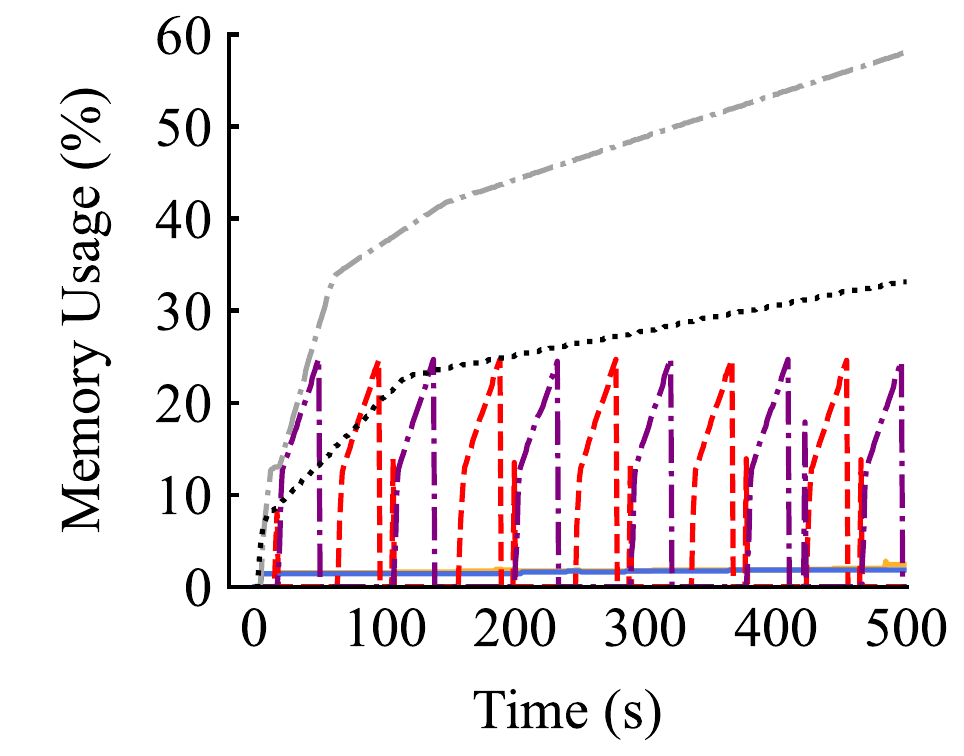}
    \label{fig:memory_overhead}}
    \vspace{0.1in}
 \caption{Overhead comparison of Totoro$^+$, OpenFL, and FedScale.}
  \label{fig:overhead}
  \vspace{-0.05in}
\end{figure}

\subsection{Overhead analysis}
\label{subsec:eval_overhead}
We train a feedforward model for text classification with a single Totoro$^+$’s dataflow tree of 10 nodes and compare the overhead with OpenFL and FedScale.

\begin{table*}[t]
\footnotesize
\resizebox{\textwidth}{!}{%
\begin{tabular}{llclllcc}
\textbf{Federated learning systems} &
  \textbf{Architecture} &
  \textbf{\begin{tabular}[c]{@{}l@{}}Central\\parameter\\server\end{tabular}} &
  \textbf{\begin{tabular}[c]{@{}l@{}}Computing\\ environment\end{tabular}} &
  \textbf{\begin{tabular}[c]{@{}l@{}}Master/worker\\ communication\\structure\end{tabular}} &
  \textbf{\begin{tabular}[c]{@{}l@{}}Type of\\learned\\ model\end{tabular}} &
  \textbf{\begin{tabular}[c]{@{}l@{}}Scale to\\massive diverse\\ applications\end{tabular}} &
  \textbf{\begin{tabular}[c]{@{}l@{}}Adapt to\\edge \\networks\end{tabular}} \\ \hline \\[-2.0ex]

FedAT~\cite{fedat}, TiFL~\cite{tifl}, Oort~\cite{oort}, PyramidFL~\cite{pyramidfl} &
  Centralized &
  \ding{51} &
  Datacenter &
  Hub-and-spoke &
  Global consensus &
 \ding{55} &
  \ding{55} \\
FLOAT~\cite{float}, REFL~\cite{refl} &
  Centralized &
  \ding{51} &
  Cloud &
  Hub-and-spoke &
  Global consensus &
 \ding{55} &
  \ding{51} \\
AdaFL~\cite{adafl}, FeS~\cite{fes}, DoFed-SPP~\cite{dofed} &
  Centralized &
  \ding{51} &
  Edge &
  Hub-and-spoke &
  Global consensus &
 \ding{55} &
  \ding{51} \\
Client-edge-cloud~\cite{client-edge-cloud}, Edge-DemLearn~\cite{edge-demlearn} &
  Hierarchical &
  \ding{51} &
  Edge &
  Hub-and-spoke &
  Global consensus &
  \ding{55} &
  \ding{55} \\
  \begin{tabular}[c]{@{}l@{}}HiFlash~\cite{hiflash}, Hwamei~\cite{hwamei}, ShapeFL~\cite{shapefl},\\FedUC~\cite{feduc}, 
AHFL~\cite{ahfl}\end{tabular}
  &
  Hierarchical &
  \ding{51} &
  Edge &
  Dynamically structured tree &
  Global consensus &
  \ding{55} &
  \ding{51} \\
$D^2$~\cite{d2}, 
BrainTorrent~\cite{braintorrent}, Lalitha et al.~\cite{fully_dfl} &
  Decentralized &
  \ding{55} &
  Datacenter &
  P2P (one-hop neighbor) &
  Global consensus &
  \ding{55} &
  \ding{55} \\
ByRDiE~\cite{byrdie}, BRIDGE~\cite{bridge}, Gupta et al.~\cite{byzantine_fault_tolerance} &
  Decentralized &
  \ding{55} &
  Datacenter &
  P2P (one-hop neighbor) &
  Global consensus &
  \ding{55} &
  \ding{55} \\
Bellet et al.~\cite{personalized_private_p2p_ml}, Vanhaesebrouck et al.~\cite{decentralized_collaborative_learning_personalized_models} &
  Decentralized &
  \ding{55} &
  Datacenter &
  P2P (one-hop neighbor) &
  Personalized &
  \ding{55} &
  \ding{55} \\
  \begin{tabular}[c]{@{}l@{}}AsyDFL~\cite{asydfl}, YOGA~\cite{yoga}, FRACTAL~\cite{fractal},\\DFLStar~\cite{dflstar}, FedHP~\cite{fedhp}\end{tabular}
&
  Decentralized &
  \ding{55} &
  Edge &
  P2P (one-hop neighbor) &
  Personalized &
  \ding{55} &
  \ding{51} \\ \\[-2.0ex] \hline & \\[-2.0ex]
\textbf{Totoro$^+$} &
  Decentralized &
  \ding{55} &
  Edge &
  Dynamically structured tree &
  Global consensus &
  \ding{51} &
  \ding{51}
\end{tabular}%
}
\vspace{0.05in}
\caption{Overview of state-of-the-art federated learning system designs compared to Totoro$^+$.}
\label{table:compare}
\end{table*}

\noindentparagraph{CPU overhead.} 
Figure~\ref{fig:cpu_overhead} shows the CPU overhead comparison of Totoro$^+$, OpenFL, and FedScale. For a fair comparison, we divide the CPU overhead into two parts: ``Overlay'' and ``Training'', where Overlay represents DHT-related tasks(including constructing P2P overlay, overlay maintenance, building dataflow trees, replanning paths, etc.), whereas Training includes any FL-related tasks (e.g., model training, model validation, model aggregation, etc.). Also, ``Master'' and ``Worker'' represent resource usage on the master node and worker node in a Totoro$^+$'s dataflow tree, respectively. The results show that the master and worker nodes collaboratively run the FL-related tasks, so the CPU usage on these two nodes takes turns reaching the peak. In contrast, OpenFL and FedScale maintain high CPU usage all the time. For DHT-related tasks, Totoro$^+$ only adds negligible CPU overhead, demonstrating Totoro$^+$’s portability to resource-constrained edge nodes.

\noindentparagraph{Memory overhead.} 
Figure~\ref{fig:memory_overhead} shows the memory overhead comparison of Totoro$^+$, OpenFL, and FedScale. The initial increase in the memory overhead of Overlay is because of the construction of the P2P overlay, routing tables, neighborhood sets, leaf sets, and dataflow trees. After the DHT starts functioning, no additional memory overhead is added in either the master (root node) or the leaf nodes. For FL-related tasks, the memory usage of Training on the master and worker nodes matches the trend of CPU usage. OpenFL and FedScale demand more and more memory as they allocate and cache models, training data, and other necessary metadata in memory all the time.

\section{Related Work}
This section discusses existing FL system designs and their shortcomings for deployment in edge platforms (Table~\ref{table:compare}).

\noindentparagraph{Centralized FL systems.} 
The centralized server-client architecture is widely used in state-of-the-art FL systems such as Oort~\cite{oort}, FedAT~\cite{fedat}, PyramidFL~\cite{pyramidfl}, TiFL~\cite{tifl}, and FEDRECON~\cite{FEDRECON}. In these systems, a powerful centralized parameter server and sufficient node-to-server bandwidth are provided to maintain the communications between the parameter server and clients. These systems focus on optimizing participant selection strategies~\cite{tifl, fedat, oort, pyramidfl}, improving communication efficiency~\cite{acgd, fedavg}, tackling data heterogeneity~\cite{fedprox, fedopt,dofed}, or ensuring privacy~\cite{differentially_private_fl, shuffle_model_dp_fl}.
For example, many algorithms have been proposed to reduce the communication overhead concentrated on the central server by reducing communication rounds with local updates allowance~\cite{dont_use_large, adacomm}, employing compression techniques to reduce the transmitted bits~\cite{qsparse, fetchsgd}, and sampling a subset of clients~\cite{tifl, fedat, oort}. 
Some systems harness workers' resource strengths with feedback integration~\cite{float, refl, adafl, fes}. FLOAT~\cite{float} adopts a multi-objective reinforcement learning agent that autonomously selects optimization techniques and configurations to meet individual client resource conditions. AdaFL~\cite{adafl} aims at NLP federated learning and 
introduces an online configurator that automatically adjusts the tuning depth and width based on training accuracy, per-round training time, and network time for each communication round.
While these systems work well in resource-rich datacenters, the centralized control plane with a static assignment of parameter servers may not adapt well to resource-constrained edge settings with millions of nodes, numerous FL applications, and unreliable network links.

\noindentparagraph{Hierarchical FL systems.} In hierarchical FL systems, an intermediate layer (e.g., edge servers) is inserted between the central server and client devices to perform partial aggregations on distributed local models before the global aggregation. 
The intermediate layer helps mitigate the non-independent and identically distributed (non-IID) effects of client data~\cite{edge-demlearn}, control the staleness of model updates~\cite{hiflash} from client devices with heterogeneous resources, synchronize edge and cloud aggregation frequencies~\cite{hwamei}, associate client devices with edge servers for training and communication efficiency~\cite{shapefl,feduc}, 
reduce the communication burden on the central server~\cite{hier_fl, client-edge-cloud}, and offload the computation burden from client devices to edge servers~\cite{ahfl}.
Abad et al.~\cite{hier_fl} employ small cell-base stations to orchestrate federated learning among mobile users and periodically exchange model updates with the macro-base station for global model learning. 
AHFL~\cite{ahfl} proposes to offload the training overhead from end devices to edge servers and inject local noise to secure data privacy while the cloud is responsible for model aggregation only. 
Although this structure enhances scalability with additional edge aggregators, they may become points of failure, causing disconnection between the central server and client devices and interrupting the training process. Some studies like  HiFlash~\cite{hiflash}, FedUC~\cite{feduc}, and AHFL~\cite{ahfl} leverage asynchronous aggregation between edge aggregators and cloud aggregators to avoid points of failure, but additional efforts to balance cloud aggregation and edge aggregation frequencies.

\noindentparagraph{Peer-to-peer FL (P2PFL) systems.}
P2PFL is a distributed learning protocol without a central parameter server. The key idea is to leverage P2P communication between individual clients to exchange model updates. Model aggregation and updates are performed locally on client nodes with the acquired gradients from their one-hop neighbors. Previous P2PFL research considered design from different angles: adversarial models (non-Byzantine vs. Byzantine settings) and learned models (global consensus vs. personalized). Non-Byzantine algorithms~\cite{d2, braintorrent, fully_dfl} focus on improving convergence speed. Conversely, Byzantine algorithms~\cite{byrdie, bridge,byzantine_fault_tolerance} aim to ensure that the trained model remains close to the model learned without adversaries. Global consensus learning~\cite{d2, braintorrent, fully_dfl, byrdie, bridge, byzantine_fault_tolerance} aims to generate a single global model at the end of training, while personalized learning~\cite{personalized_private_p2p_ml,decentralized_collaborative_learning_personalized_models} allows each peer to train a unique model. 
Some studies aim to improve the training efficiency of P2PFL~\cite{fedhp, yoga, asydfl, fractal}. AsyDFL~\cite{asydfl} proposes to push a subset of gradient updates to peer to trade-off communication cost and training accuracy and leverage asynchronous model aggregation to mitigate the straggler effect.
FedHP~\cite{fedhp} integrates adaptive control of both local updating frequency and network topology to speed up convergence and improve model accuracy.
These efforts mostly focus on novel algorithm development for P2P model aggregation and communication topology construction without considering cross-layer support of implementation. 
Totoro$^+$ backs up the P2PFL work with a back-end architecture for task management, application-specific customization, and failure recovery.

\section{Conclusion}
We present Totoro$^+$, a novel scalable federated learning system for edge networks. It presents three design innovations: a locality-aware P2P multi-ring structure, a publish/subscribe-based forest abstraction, 
and a game-theoretic path planning model. 
We evaluate Totoro$^+$ on 500 Amazon EC2 nodes using real-world CV and NLP datasets at different scales. 
We compare Totoro$^+$ with the state-of-the-art and demonstrate substantial advancements in scalability and load balancing while significantly speeding up the total training time for many concurrently running applications, reducing the communication overhead, and efficiently adapting to unreliable edge networks and churns. 
We provide a theoretical guarantee that the game-theoretic path planning model achieves an $\epsilon$-approximate Nash equilibrium.
We present the discussions and future work in Appendix~\ref{sec:discussions}.

\section*{Acknowledgment}


The authors would like to thank the Editor-in-Chief, Prof. Xian-He Sun; the Associate Editor, Prof. Amelie Chi Zhou; and the anonymous reviewers for their valuable comments and constructive suggestions, which significantly improved the quality of this article. The authors also thank Yue Lin and Shih-Jui Liang for their assistance with the experimental evaluation and validation of the proposed routing algorithm.

\bibliographystyle{IEEEtran}
\bibliography{References}

\vspace{-12mm}

\begin{IEEEbiography}[{\includegraphics[width=1in,height=1.in,clip,keepaspectratio]{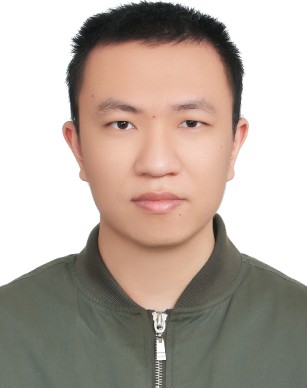}}]
{Cheng-Wei Ching}
received the MS degree in Computer Science and Information Engineering from National Chung Cheng University, Taiwan, in 2021. He is pursuing a PhD in Computer Science and Engineering at the University of California Santa Cruz, USA.
His research interests include distributed systems for machine learning, 
approximation algorithms and their applications, and mobile edge computing. 
\end{IEEEbiography}

\vspace{-12mm}

\begin{IEEEbiography}[{\includegraphics[width=1in,height=.9in,clip,keepaspectratio]{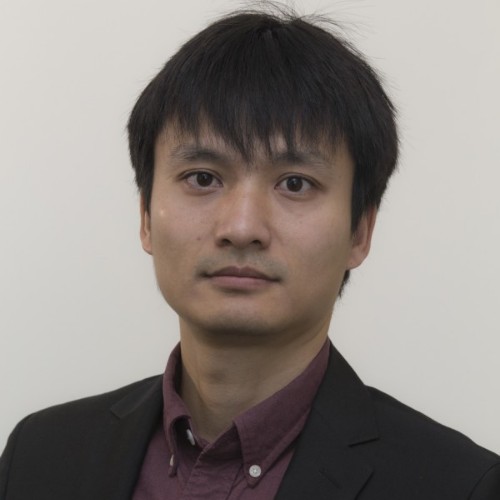}}]
{Xin Chen} received the BS degree in Computer Science from Shandong University, China, in 2009, the MS degree in Software Engineering from Tsinghua University, China, in 2012, and the PhD degree in Computer Science from the Georgia Institute of Technology, USA, in 2020. His research focuses on computer systems and machine learning systems, including solving path planning and collision detection problems.
\end{IEEEbiography}

\vspace{-12mm}

\begin{IEEEbiography}[{\includegraphics[width=1in,height=1.05in,clip,keepaspectratio]{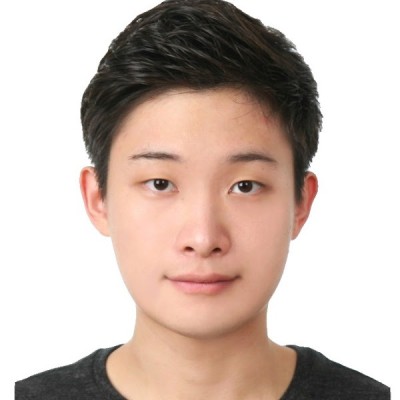}}]
{Taehwan Kim} received his BS from Bucknell University in 2020 and his MS from Virginia Tech in 2022. He is now a software engineer at Google, where he focuses on Gemini post-training infrastructure.
\end{IEEEbiography}

\vspace{-12mm}

\begin{IEEEbiography}[{\includegraphics[width=1.in,height=1.05in,clip,keepaspectratio]{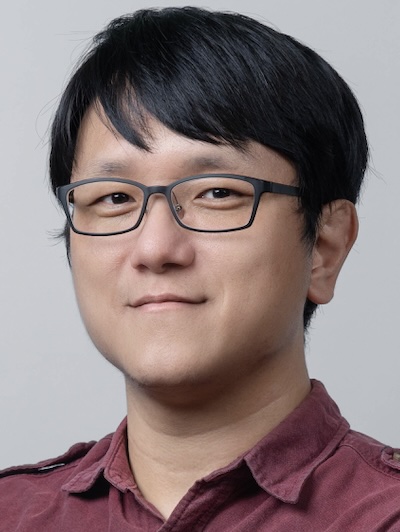}}]
{Jian-Jhih Kuo} (Member, IEEE) received the B.S. and Ph.D. degrees in computer science from National Chung Cheng University (CCU), Taiwan, in 2008, and National Tsing Hua University, Taiwan, in 2014, respectively. He was a Postdoctoral Fellow in the Institute of Information Science, Academia Sinica, Taiwan. He joined the Department of Computer Science and Information Engineering, CCU, Taiwan, in 2018, where he is currently an associate professor. His research interests include computer networking, quantum networking, traffic engineering, and cloud and edge computing. He was a recipient of the Ta-You Wu Memorial Award from National Science and Technology Council (NSTC), Taiwan, in 2025, NSTC Project for Excellent Junior Research Investigators in 2022 and 2025, and CCU Young Faculty Award in 2021.
\end{IEEEbiography}

\vspace{-110mm}

\begin{IEEEbiography}[{\includegraphics[width=1in,height=1.1in,clip,keepaspectratio]{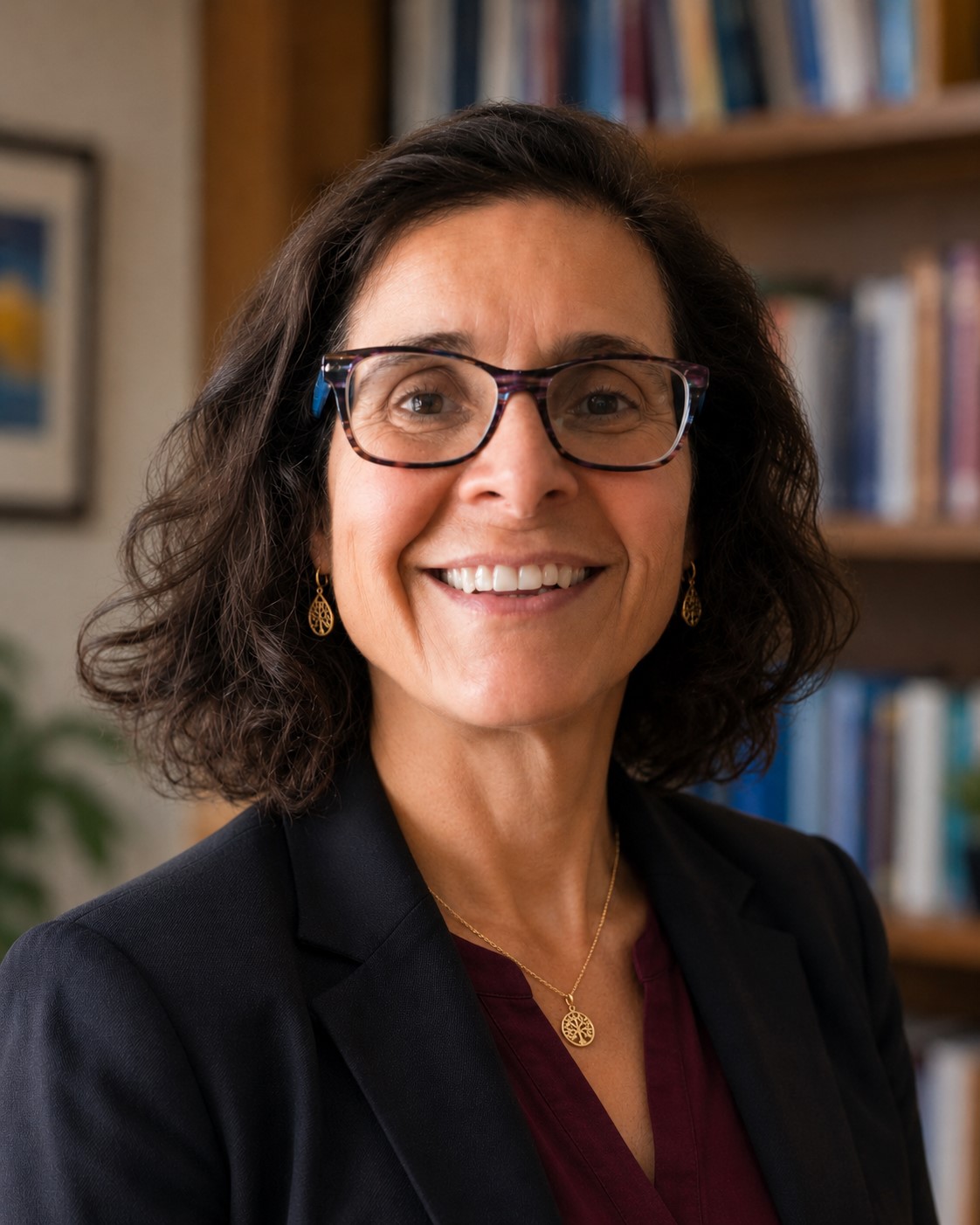}}]
{Dilma Da Silva} is a Regents Professor and holder of the Ford Design Professorship II in the Department of Computer Science and Engineering at Texas A\&M University. From 2022 to 2026, she served in several leadership roles at the U.S. National Science Foundation. Her primary research interests include distributed systems, high-performance computing, operating systems, computer science education, and quantum computing. Dilma received her Ph.D. in computer science from the Georgia Institute of Technology in 1997.
\end{IEEEbiography}

\vspace{-110mm}

\begin{IEEEbiography}[{\includegraphics[width=1in,height=1.1in,clip,keepaspectratio]{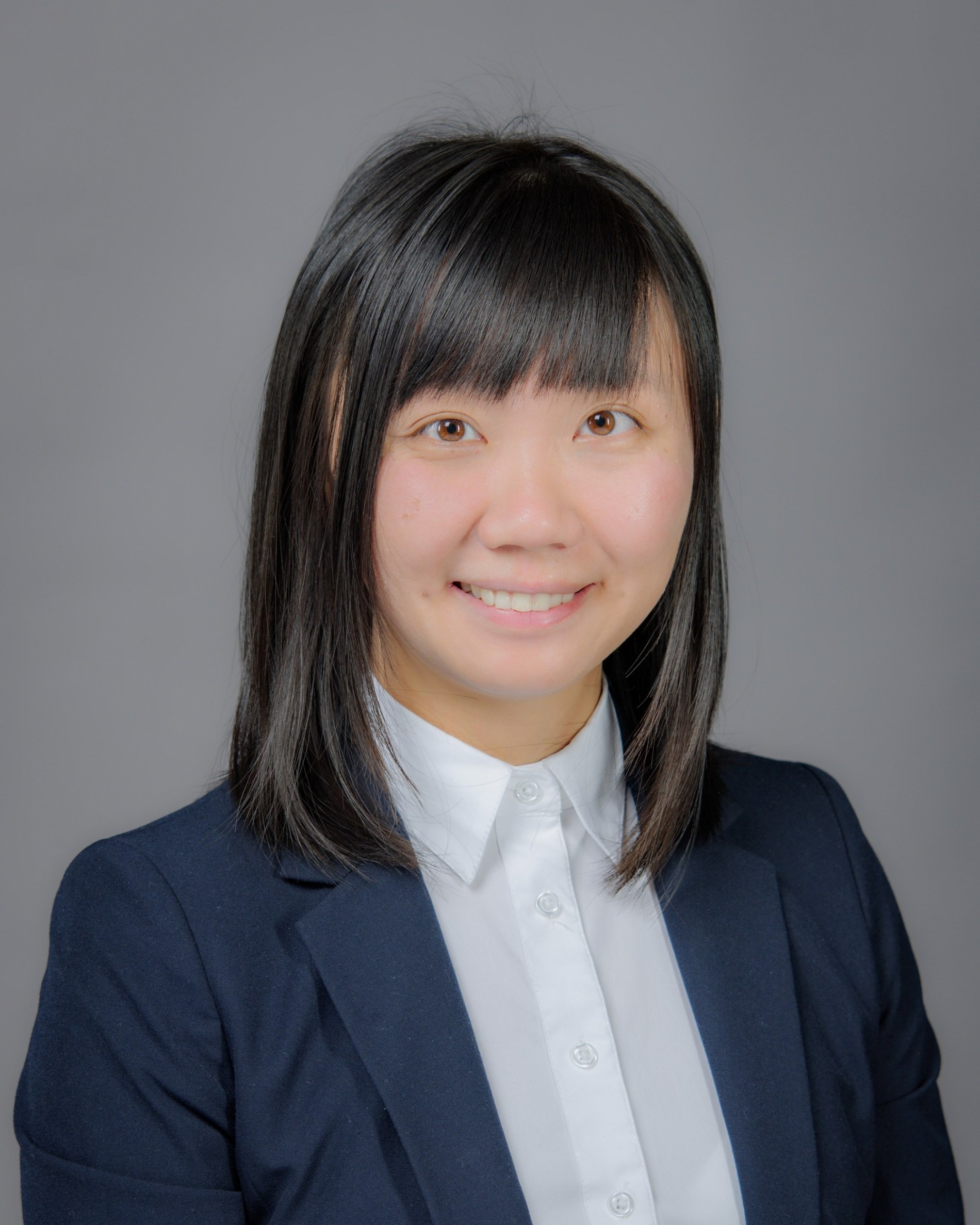}}]
{Liting Hu} received the BS degree in Computer Science from the Huazhong University of Science and Technology, China, in 2007, and the PhD degree in Computer Science from the Georgia Institute of Technology, USA, in 2016. She conducts AI-driven experimental computer systems research in stream processing systems, cloud and edge computing, distributed systems, and systems virtualization. Currently she is an associate professor at the University of California Santa Cruz.
\end{IEEEbiography}

\clearpage
\appendices
\setcounter{section}{0}
\setcounter{subsection}{0}

{\color{dv}
\section{Application Advertisement and Discovery in Totoro and Totoro$^+$}
\label{app:advertise_discover}

\noindentparagraph{Totoro's design and limitations.}
In Totoro, the system does not have an independent system component for federated learning (FL) application advertisement and discovery. Instead, Totoro implicitly relies on the AppId and its DHT-based P2P overlay as a rendezvous mechanism for applications. Once an FL application is created, the node whose node ID is numerically closest to the AppId is automatically designated as the master. The system assumes that participating nodes already know the corresponding AppId of the application, and they route \texttt{JOIN} messages through the overlay network toward that AppId. Due to the properties of the P2P overlay, \texttt{JOIN} messages are eventually aggregated at the master node or any other nodes in the dynamically constructed dataflow tree, allowing nodes to participate in model dissemination or gradient aggregation. This design essentially abstracts application access and participation as a join-by-identifier mechanism, rather than providing an explicit application directory or advertisement plane.

However, this design also introduces several system-level challenges. First, due to the absence of an explicit application advertisement or discovery plane, nodes cannot directly query or enumerate the FL applications currently running in the system. As a result, application existence and access rely on additional, out-of-band coordination mechanisms (e.g., a third party advertisement platform~\cite{aaap, adsbpc}). Second, application-related metadata (e.g., model type, resource requirements, data constraints, or participation conditions) is not natively published or disseminated within the system, limiting the ability of nodes to perform conditional matching or application selection based on their own capabilities or environmental states. Finally, as the number of concurrently running FL applications increases, AppId dissemination and participant management increasingly depend on external control planes. In heterogeneous and highly dynamic edge environments, this dependency may further increase system management complexity and coordination overhead.

\noindentparagraph{Advertise-discover tree in Totoro$^+$.}
To address the above limitations in application advertisement and discovery, Totoro$^+$ introduces an additional application management structure, namely the \textit{Advertise-Discover Tree (AD tree)}. In Totoro$^+$, when a new FL application is created, its corresponding master node immediately subscribes to a globally shared AD tree. This subscription process follows the same mechanism as a standard FL application dataflow tree. After joining the AD tree, the new master uploads the AppId and associated metadata of the FL application it manages to the AD tree. At the same time, the master receives broadcast messages from the AD tree containing the AppIds and corresponding information of all FL applications currently running in the overlay. 

Following this mechanism, the AD tree forms a distributed and scalable application directory that continuously maintains and updates application-level state within the overlay network. When a new node joins the Totoro$^+$ overlay network and wants to explore which FL applications are currently available, it only needs to subscribe to the AD tree to obtain the list of active applications along with their descriptive metadata. Based on the AppIds obtained from the AD tree, the node can then subscribe to the corresponding FL application dataflow tree and participate in training. By explicitly abstracting application advertisement and discovery as an independent dataflow tree, Totoro$^+$ provides native application visibility and dynamic exploration capabilities while preserving decentralization. Moreover, the additional overhead introduced by the AD tree is negligible, as only the masters of FL application dataflow trees need to subscribe to the AD tree, and newly joined nodes can unsubscribe after identifying the desired FL application.

\section{Bandit-Based Model v.s. Game-Theoretic Model}
\label{app:math_model_comparison}
We compare the mathematical models adopted in Totoro$^+$ (game-theoretic) and Totoro (bandit-based) formulations in this section. The comparison is conducted along five dimensions: decision maker, reward mechanism, source of uncertainty, optimization objective, and optimal solution.

\noindentparagraph{Decision maker.}
In the bandit-based model, each node is modeled as an independent learner. In each round, a node selects a single action, namely choosing one path $p_k\in \mathcal{P}$ from the set of all possible paths $\mathcal{P}$. The decision made by one node does not affect the decisions of other nodes, as nodes are unaware of each other’s existence. Moreover, the model does not consider interference or competition effects when multiple nodes select the same path. Using a casino analogy, each node corresponds to an independent player, and each action represents choosing a slot machine to play. Although multiple players may be selecting and playing slot machines simultaneously, they are assumed to be unaware of one another, and the payoff of a slot machine does not change despite being played by multiple players (e.g., jackpots or missed pulls).

In contrast, the game-theoretic model jointly considers the simultaneous path selection decisions of $N$ nodes. The action space expands from $\mathcal{P}$ to the Cartesian product $\prod_{n=1}^N \mathcal{P}_n$ , representing all possible combinations of paths chosen by the $N$ nodes. Accordingly, the action in each round is no longer a single path $p_k$, but a vector
$\textbf{p}=(p_1,\cdots,p_N)$, which captures the paths selected by all nodes. Under the casino analogy, this corresponds to a setting where $N$ players are simultaneously present in the casino. Each player chooses a slot machine to play, while being aware of the presence of other players. Each player therefore considers whether a slot machine remains worthwhile after it has been played by others, or whether switching to a different machine would yield a better payoff. From the perspective of decision making, the game-theoretic formulation more closely reflects practical networking scenarios. In particular, \textit{when multiple nodes select the same path, the effective performance of that path may degrade due to shared transmission resources (e.g., under bandwidth-sharing assumptions), a phenomenon that is naturally captured in the game-theoretic model but ignored in the bandit formulation}.

\noindentparagraph{Reward mechanism.}
In the bandit-based model, the reward is defined as the path delay,
\[
D_\theta=\sum_{i\in E}x_i \cdot \frac{1}{\theta_i},
\]
where $x_i$ indicates whether link $i$ is included in a selected path, and $\theta_i$ denotes the success probability of link $i$. Under this formulation, although $\theta_i$ is unknown, it can be learned from historical observations. Importantly, the reward associated with a path is independent of the actions of other nodes. The model does not account for performance degradation when multiple nodes select the same link, as the environment is assumed to be \textit{stationary} and \textit{exogenous}.

In contrast, the game-theoretic model defines the mean reward of node $n$ as
\[
r_n(\textbf{p})=r^{p_n} (n^{p_n} (\textbf{p}),\theta_{p_n})
\]
after selecting path $p_n$, 
where $\theta_{p_n}$ represents the quality of path $p_n$, and the reward explicitly depends on the number of nodes selecting the same path, denoted by $n^{p_n} (\textbf{p})$. As a result, the reward is endogenous and jointly determined by the collective actions of all nodes. From the perspective of reward mechanisms, the game-theoretic formulation better reflects practical networking environments. When a path is utilized by multiple nodes, the reward should decrease due to reduced transmission efficiency and increased delay caused by shared resources. \textit{This congestion effect is naturally captured in the game-theoretic model but ignored in the bandit-based model}.

\noindentparagraph{Uncertainty source.}
In the bandit-based model, uncertainty primarily arises from the unknown environmental parameters. Specifically, the success probability $\theta_i$ of each link $i\in E$ is modeled as an unknown but fixed random but i.i.d parameter that remains unchanged throughout the routing process. Although $\theta_i$ is unknown a priori, it can be gradually estimated and learned from historical transmission outcomes. Consequently, the uncertainty in the bandit-based model is statistical uncertainty, which stems from estimation errors of random environmental parameters rather than from the behavior of other nodes. Since the reward (or delay) depends solely on the selected path and its underlying link success rates, the environment is assumed to be stationary and exogenous, and is not affected by the decisions of other nodes. 

In contrast, in the game-theoretic model, uncertainty does not arise solely from the unknown path quality $\theta_p$, but more fundamentally from the strategic choices of other nodes. As the reward obtained by node $n$,
\[
r_n(\textbf{p})=r^{p_n} (n^{p_n} (\textbf{p}),\theta_{p_n}),
\] 
explicitly depends on the joint action $\textbf{p}=(p_1,\cdots,p_N)$, a single node cannot determine its realized reward without considering the actions of other nodes, even if $\theta_{p_n}$ is known. In other words, each node faces \textit{strategic uncertainty}, which arises from the decisions of other rational nodes operating on the same time scale, rather than from an independently learnable random environment. From the perspective of uncertainty sources, the bandit-based model focuses on \textit{how to learn in a fixed and exogenous stochastic environment}, whereas the game-theoretic model characterizes \textit{how to cope with uncertainty induced by strategic interactions among multiple nodes}. In addition to unknown path qualities, the latter must account for the actions of other nodes, and thus more accurately reflects the dynamics and uncertainty of practical edge routing scenarios, where multiple nodes simultaneously compete for limited network resources.

\noindentparagraph{Optimization objective.}
In the bandit-based model, the routing problem is formulated as an online learning problem, whose optimization objective is to minimize the cumulative regret with respect to a fixed optimal path. Specifically, the bandit-based model defines the regret at time $K$ as
\[
    R^\pi(K) =\mathbb{E}\Bigg[\sum_{k=1}^K D^\pi(k)\Bigg]-K D_\theta(p^*),
\]
where
$p^* \in \argmin_{p\in \mathcal{P}}D_\theta(p)$
denotes the single path with the minimum expected delay under the true but unknown environment parameter $\theta$. This objective reflects the central concern of the bandit-based model: in a \textit{stationary} and \textit{exogenous} environment, how to gradually learn and approach a fixed optimal path with minimal exploration cost. Accordingly, regret measures the additional delay incurred during the learning process due to incomplete information, and fundamentally serves as an indicator of \textit{learning efficiency}. 

In contrast, in the game-theoretic model, the optimization objective is not defined relative to a static optimal path, but is instead evaluated using \textit{Nash regret}. The model focuses on whether, under simultaneous path selection by multiple nodes, any node has an incentive to unilaterally deviate from the currently adopted joint policy. The objective function at time $T$ is defined as
\[
    \text{Nash-Regret}(T)=\sum_{t=1}^T \max_{n\in [N]}(V^{\pi^+_n,\pi^t_{-n}}_n - V^{\pi^t}_n),
\]
where $\pi^+_n$ denotes the best-response policy of node $n$ given that the strategies of all other nodes are fixed, $\pi^t_{-n}$ represents the marginal joint policy of nodes $1,\cdots, n-1, n+1,\cdots, N$ at time $t$, and $V_n^\pi$ represents the value of policy $\pi$ for node $n$. This objective captures \textit{strategy stability}, rather than the learning of a globally shortest path. When the Nash regret is sufficiently small, no individual node can significantly improve its expected reward through unilateral strategy deviation. Therefore, from the perspective of optimization objectives, the bandit-based model aims to rapidly identify a fixed optimal action, whereas the game-theoretic model emphasizes suppressing incentives for deviation and promoting stability in strategic interactions. The corresponding regret definitions, comparison baselines, and performance interpretations of the two models thus differ in a fundamental manner.

\noindentparagraph{Optimal solution.}
In the bandit-based model, the notion of an optimal solution does not correspond to a static strategy equilibrium, but rather to a learning policy $\pi$ that guarantees \textit{sublinear regret}. The optimality of such a policy is evaluated by whether, as $K$ increases,
\[
\frac{R^\pi(K)}{K}\rightarrow 0,
\]
which ensures that the system’s average latency converges to the latency of the optimal path $p^*$ in the long run. In other words, the optimal solution in the bandit-based model emphasizes \textit{learning dynamics} and \textit{convergence speed}, and assumes that once the environment parameters are sufficiently learned, the system behavior eventually concentrates on a single optimal path. 

In contrast, in the game-theoretic model, the concept of an optimal solution is closely tied to \textit{equilibrium}. The objective is not to identify a single path that is globally optimal for all nodes, but to learn a joint policy $\pi^*$ under which no node can achieve a significant gain by unilaterally changing its strategy. When the Nash regret converges to zero, the resulting solution can be interpreted as an approximate \textit{Nash equilibrium} or a \textit{coarse correlated equilibrium}. Such solutions reflect a stable operating state under multi-node competition and resource sharing, rather than the minimization of a single performance metric. Therefore, from the perspective of optimal solutions, the bandit-based model focuses on how to efficiently learn a globally optimal behavior, whereas the game-theoretic model focuses on whether the system can reach a stable strategy profile with no incentive for unilateral deviation under strategic coupling and competition. The former prioritizes learning efficiency and asymptotic performance, while the latter places greater emphasis on system-level stability and fairness. As a result, the game-theoretic formulation is better suited to path selection in edge environments with limited resources and multiple competing nodes.

\section{Preliminaries on Games}
\label{app:game_prelim}
This section presents the game-theoretic preliminaries that form the foundation
of the mathematical modeling and algorithmic design of Totoro$^+$.
The problem studied in this work involves multiple nodes making decentralized
routing decisions under shared resource constraints, where the reward obtained
by each node depends on both its own action and the collective behavior of other
nodes.
Such interactions are naturally captured by game-theoretic models, especially
in settings with incomplete information and limited feedback.

We first introduce general-sum matrix games~\cite{gen_sum_game_icml_2005, gen_sum_game_jmlr_2023} as a general framework for modeling
multi-agent decision-making with possibly misaligned objectives. We then introduce the notion of general policies for these games~\cite{congestion_games, monderer1996}.
Within this framework, we review equilibrium notions and performance metrics,
including Nash equilibrium and Nash regret~\cite{nash_regret_nips, nash_regret_icml}, which are central to the analysis of
learning dynamics in repeated games.
We then discuss potential games~\cite{monderer1996} and congestion games~\cite{congestion_games}, which impose additional
structure on general-sum games and allow us to explicitly model performance
degradation due to resource contention.
Finally, we describe different types of feedback available to agents in
congestion games, distinguishing between semi-bandit and bandit feedback. In this work, we focus on the bandit feedback setting, which reflects realistic
communication constraints and directly motivates the learning algorithms
developed for Totoro$^+$.

\noindentparagraph{General-sum matrix games}~\cite{gen_sum_game_icml_2005, gen_sum_game_jmlr_2023}.
We study the framework of general-sum matrix games, which is specified by the tuple
$\mathcal{G} = (\{\mathcal{A}_i\}_{i=1}^m, R)$,
where $m$ denotes the number of players, $\mathcal{A}_i$ is the action space of player $i$, and
$R(\cdot \mid \textbf{a})$ represents a reward distribution supported on $[0, r_{\max}]^m$ with mean
$\textbf{r}(\textbf{a})$. Let $\mathcal{A} = \mathcal{A}_1 \times \cdots \times \mathcal{A}_m$ denote the joint action
space, and let $\textbf{a} = (a_1, \ldots, a_m) \in \mathcal{A}$ be a joint action profile. Once all players select
actions $\textbf{a} \in \mathcal{A}$, a reward vector is drawn according to $\textbf{r} \sim R(\cdot \mid \textbf{a})$,
and player $i$ receives a reward $r_i \in [0, r_{\max}]$ with expectation $r_i(\textbf{a})$.
Each player aims to maximize her own expected reward.

\noindentparagraph{General policy}~\cite{congestion_games, monderer1996}. A general policy $\pi$ is defined as a probability distribution in $\Delta(\mathcal{A})$, the simplex over the
joint action space $\mathcal{A}$. A product policy
$\pi = (\pi_1, \ldots, \pi_m)$ corresponds to an element of $\Delta(\mathcal{A}_1) \times \cdots \times \Delta(\mathcal{A}_m)$, under which $\textbf{a} = (a_1, \ldots, a_m) \sim \pi$
indicates that each $a_i$ is independently sampled according to $a_i \overset{\text{i.i.d.}}{\sim} \pi_i$.
The value of a policy $\pi$ for player $i$ is defined as $V_i^\pi = \mathbb{E}_{\textbf{a} \sim \pi}[r_i(\textbf{a})]$.

\noindentparagraph{Nash equilibrium and Nash regret}~\cite{nash_regret_nips, nash_regret_icml}.
Given a general policy $\pi$, let $\pi_{-i}$ denote the marginal joint policy of all players except player $i$.
The best response of player $i$ to policy $\pi$ is given by
\[
\pi_i^\dagger = \argmax_{\mu \in \Delta(A_i)} V_i^{\mu, \pi_{-i}},
\]
with the corresponding value $V_i^{\dagger, \pi_{-i}} := V_i^{\pi_i^\dagger, \pi_{-i}}$.
Our objective is to characterize approximate Nash equilibria of the matrix game, defined as follows.

\begin{defi}
A product policy $\pi$ is an \textit{$\epsilon$-approximate Nash equilibrium} if $\max_i \left( V_i^{\dagger, \pi_{-i}} - V_i^\pi \right) \le \epsilon$.
\end{defi}

An $\epsilon$-approximate Nash equilibrium can be achieved by ensuring sublinear growth of the Nash regret,
which we define below. For a more comprehensive discussion, see Section~3 in
Ding et al.~\cite{ding2022}.

\begin{defi}
With $\pi^k$ denoting the policy at the $k$-th episode, the Nash regret after
$K$ episodes is defined as
\[
\text{Nash-Regret}(K)
=
\sum_{k=1}^K
\max_{i \in [m]}
\left(
V_i^{\dagger, \pi_{-i}^k}
-
V_i^{\pi^k}
\right).
\]
\end{defi}

\begin{remark}
If $\max_{i \in [m]}$ in the definition of Nash regret is replaced by $\sum_{i=1}^m$, the resulting single-step
Nash regret at episode $k$ reduces to the Nikaid\^o--Isoda function evaluated at $\pi^k$, which is a
widely used objective for equilibrium computation
\cite{nikaido1955, raghunathan2019}.
Such a replacement increases the regret bounds by a factor of $m$, but does not affect our main conclusions.
\end{remark}

\noindentparagraph{Potential games}~\cite{monderer1996}.
A potential game is a class of general-sum games for which there exists a
potential function $\Phi : \Delta(\mathcal{A}) \to [0, \Phi_{\max}]$
such that, for any player $i \in [m]$ and any policies $\pi_i, \pi_i', \pi_{-i}$,
the following condition holds:
\[
\Phi(\pi_i, \pi_{-i}) - \Phi(\pi_i', \pi_{-i})
=
V_i^{\pi_i, \pi_{-i}} - V_i^{\pi_i', \pi_{-i}}.
\]

It follows directly that any policy maximizing the potential function
constitutes a Nash equilibrium.

\noindentparagraph{Congestion games}~\cite{congestion_games}.  
A congestion game models competitive interactions over shared resources whose
performance degrades as the level of simultaneous usage increases.
Formally, the game is specified by $\mathcal{G} = (\mathcal{F}, \{\mathcal{A}_i\}_{i=1}^m, \{R^f\}_{f \in \mathcal{F}})$,
where $\mathcal{F} = [F]$ denotes the set of facilities, and
$R^f(\cdot \mid n) \in [0,1]$ is the reward distribution associated with facility
$f$ when it is used by $n \in [m]$ players, with mean reward $r^f(n)$.
Each player $i$ selects an action $a_i \in \mathcal{A}_i$, where an action
corresponds to a subset of facilities, i.e., $a_i \subseteq \mathcal{F}$.

Given a joint action profile $\textbf{a} \in \mathcal{A}$ selected by all players,
each facility $f$ generates a random reward according to $r^f \sim R^f(\cdot \mid n^f(\textbf{a}))$,
where $n^f(\textbf{a}) = \sum_{i=1}^m \mathbf{1}\{ f \in a_i \}$
denotes the number of players concurrently utilizing facility $f$.
The total reward obtained by player $i$ is the aggregate reward over all
facilities included in $a_i$,
$r_i = \sum_{f \in a_i} r^f$, with expected value
$r_i(\textbf{a}) = \sum_{f \in a_i} r^f(n^f(\textbf{a})) \in [0, F]$.

\noindentparagraph{Connection to potential games}
\cite{monderer1996}.  
Congestion games constitute a special class of potential games in which
individual incentives are aligned with a global objective.
Specifically, they admit the potential function
$\Phi(\textbf{a}) = \sum_{f \in \mathcal{F}} \sum_{i=1}^{n^f(\textbf{a})} r^f(i)$.
This potential captures the cumulative impact of congestion on each facility.
For any player $i$, a unilateral deviation from $a_i$ to $a_i'$ while holding
$a_{-i}$ fixed leads to a change in the potential that exactly matches the
corresponding change in the player's reward:
\[
\Phi(a_i, a_{-i}) - \Phi(a_i', a_{-i})
=
r_i(a_i, a_{-i}) - r_i(a_i', a_{-i}).
\]
Extending the definition to mixed strategies, we define
\[
\Phi(\pi) = \mathbb{E}_{\textbf{a} \sim \pi}[\Phi(\textbf{a})].
\]
As a consequence, for any unilateral policy deviation of player $i$, we have
\[
\Phi(\pi_i, \pi_{-i}) - \Phi(\pi_i', \pi_{-i})
=
V_i^{\pi_i, \pi_{-i}} - V_i^{\pi_i', \pi_{-i}}.
\]
This connection will be instrumental in our subsequent analysis of learning dynamics under bandit feedback.

\noindentparagraph{Types of feedback.}
In congestion games, the reward information available to each player can take
one of two standard forms: semi-bandit feedback or bandit feedback, depending on
the level of observability in the model.
Under semi-bandit feedback, after selecting an action, player $i$ observes the
reward $r^f$ associated with each facility $f \in a_i$.
In contrast, under bandit feedback, after taking the action, player $i$ only
observes the aggregated reward $r_i = \sum_{f \in a_i} r^f$ without access to the individual components $r^f$.
In this work, we focus on the bandit feedback setting as it captures more realistic operating conditions in
practical scenarios, where fine-grained information about the contribution or
delay of each individual component is typically unavailable.
In particular, in Totoro$^+$, a node can usually observe only the overall
end-to-end transmission time or aggregate performance of a selected path,
rather than the latency or success rate incurred at each intermediate resource.
Accordingly, we adopt the bandit feedback setting, which reflects such practical constraints.

\section{Notation Table}
\label{app:notation_table}
Mathematical notations that will be introduced and used multiple times in this paper are summarized in Table~\ref{tab:notation_table}.

\begin{table}[t!]
\centering
\begin{tabular}{c|p{4.5cm}}    
    \Xhline{1.5\arrayrulewidth}    
    \noalign{\vskip 2pt}
    \rule{0pt}{.8\normalbaselineskip}Notation & Explanation \\ 
    \noalign{\vskip 2pt}\hline 
    \noalign{\vskip 2pt}
    
    \rule{0pt}{.8\normalbaselineskip}
    $N$ & Number of nodes in the dataflow tree \\

    $P$ & Total number of available paths to the root \\

    $p \in [P]$ & A routing path in $P$ \\

    $P_n$ & Set of available paths for node $n$ \\

    $\theta_p$ & Mean success rate of path $p$ \\

    $\mathcal{R}^p(\cdot \mid k,\theta_p)$ & Reward distribution of path $p$ when $k$ nodes select it \\

    $r^p(k,\theta_p)$ & Mean reward of path $p$ given $k$ nodes select it \\

    $\mathbf{p}=(p_1,\ldots,p_N)$ & Joint path selection of all nodes \\

    $p_n$ & Path selected by node $n$ \\

    $n^p(\mathbf{p})$ & Number of nodes selecting path $p$ under joint action $\mathbf{p}$ \\

    $r_n(\mathbf{p})$ & Reward received by node $n$ under joint action $\mathbf{p}$ \\

    $\pi_n$ & Mixed policy of node $n$ over paths $P_n$ \\

    $\Delta(P_n)$ & Probability simplex over $P_n$ \\

    $\pi=(\pi_1,\ldots,\pi_N)$ & Joint policy of all nodes \\

    $\pi_{-n}$ & Joint policy of all nodes except node $n$ \\

    $V_n^\pi$ & Expected reward (value) of node $n$ under policy $\pi$ \\

    $\pi_n^+$ & Best-response policy of node $n$ given $\pi_{-n}$ \\

    $T$ & Total number of packets transmitted \\

    $\tau$ & Number of packets for exploration per episode \\

    $\text{Nash-Regret}(T)$ & Cumulative Nash regret over $T$ packets \\

    $\alpha$ & Mixing parameter between exploitation and exploration \\

    $\beta$ & Step size for the Frank--Wolfe update \\

    $\rho_n^k$ & Exploratory policy of node $n$ at episode $k$ \\

    $\psi(p)$ & Feature vector associated with path $p$ \\

    $M(\cdot)$ & Policy correlation matrix \\

    $\widehat{\nabla}_n^k \Phi$ & Estimated gradient of the potential function for node $n$ at episode $k$ \\

    \noalign{\vskip 2pt}
    \Xhline{1.5\arrayrulewidth}
\end{tabular}
\caption{Summary of notations.}
\vspace{-0.05in}
\label{tab:notation_table}
\end{table}

}
\section{Numerical Example of Algorithm~\ref{pseudo:our_algorithm}}
\label{sec:numerical_example}

This appendix gives a numerical example of Algorithm~\ref{pseudo:our_algorithm}.

Suppose $4$ leaf nodes ($l_1, l_2, l_3, l_4$) run Algorithm~\ref{pseudo:our_algorithm} to select the next hops for forwarding model updates, and $2$ intermediate nodes ($m_1, m_2$) are available for selection (Figure~\ref{fig:algo_exp}). Let $\alpha=0.5$, $\beta=0.5$, and $\tau=2$. We use node $l_1$ to walk through Algorithm~\ref{pseudo:our_algorithm}. First, node $l_1$ follows its initial policy $\pi^1_{l_1}$ to select the next hops and collect rewards $r^{1,t}_{n}$.
Since there are two candidates for node $l_1$ to forward its packets, initial policy $\pi^1_{l_1}$ should be a 2-dimensional vector with the summation of its entries being $1$. For example, if $\pi^1_{l_1}=[0.5,0.5]$, the chance of selecting either intermediate node is $50\%$.

Suppose node $l_1$ forwards two packets to two intermediate nodes $m_1$ and $m_2$ each once, where the first packet goes to $m_1$ (i.e., $p^{1,1}_{l_1} = m_1$) and the second goes to $m_2$ (i.e., $p^{1,2}_{l_1} = m_2$), and the corresponding rewards $r^{1,1}_{l_1} = 0.4$ and $r^{1,2}_{l_1} = 0.8$. Since $\tau=2$, node $l_1$ updates its policy after forwarding two packets.

Since node $l_1$ has two candidates for packet forwarding, $P_{l_1} = \{m_1, m_2\}$, and $\Delta(P_{l_1})$ is a set that includes some feasible policies derived from $P_{l_1}$. Let
\begin{align}
    \Delta(P_{l_1}) = \Bigg[\begin{bmatrix}
           0.6 \\
           0.4 
         \end{bmatrix},
         \begin{bmatrix}
           0.5 \\
           0.5
         \end{bmatrix},
         \begin{bmatrix}
           0.3 \\
           0.7
         \end{bmatrix},
         \begin{bmatrix}
           0.1 \\
           0.9
         \end{bmatrix}\Bigg],
         \nonumber
\end{align}
where $[0.6, 0.4]$ represents a policy that has 0.6 of chance to select $m_1$ and 0.4 of chance to select $m_2$. For a policy $\lambda = [0.6, 0.4]$, the correlation matrix $M(\lambda)$ is 
\begin{align}
    M(\lambda) &= 0.6 \cdot \begin{bmatrix}
           1 \\
           0
         \end{bmatrix}\cdot \begin{bmatrix}
           1 \\
           0
         \end{bmatrix}^\top
         + 0.4 \cdot 
         \begin{bmatrix}
           0 \\
           1
         \end{bmatrix}\cdot 
         \begin{bmatrix}
           0 \\
           1
         \end{bmatrix}^\top  \nonumber\\
         &= \begin{bmatrix}
           0.6 & 0 \\
           0 & 0.4
         \end{bmatrix}.
         \nonumber
\end{align}
And its determinant $\text{det}(M(\lambda))=0.6\times 0.4 = 0.24$. Similarly, the determinants of other policies are $0.25$, $0.21$, and $0.09$. By line~\ref{pseudocode:d_optimal_design}, we know that $\rho^1_{l_1} = [0.1, 0.9]$.

\begin{figure}[t]
    \centering
    \includegraphics[width=\linewidth]{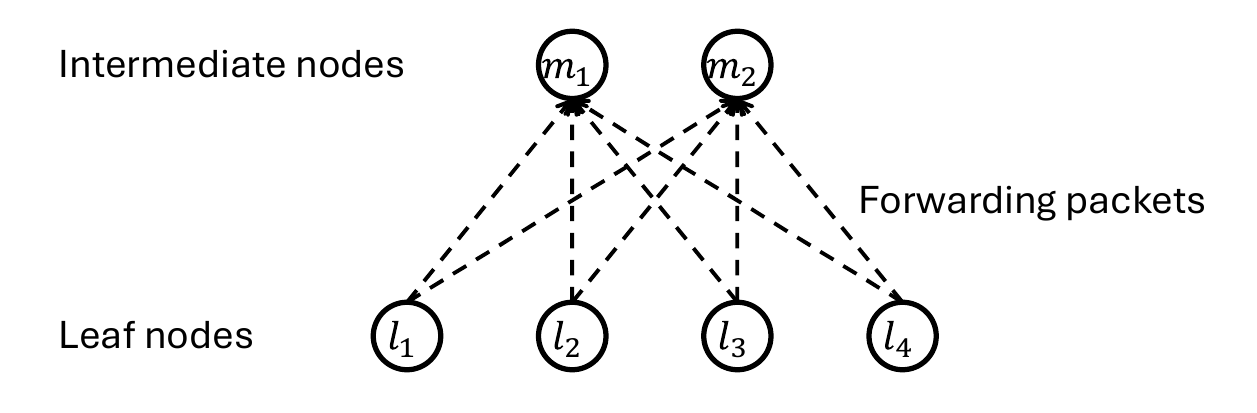}
    \caption{Example of a dataflow tree.}
    \label{fig:algo_exp}
\end{figure}

To estimate gradient $\widehat{\nabla}^1_{l_1}\Phi(p)$, we have
\begin{align}
    M(\pi^1_{l_n})^{-1} &=  \Bigg[0.5 \cdot\begin{bmatrix}
           1 \\
           0
         \end{bmatrix}\cdot \begin{bmatrix}
           1 \\
           0
         \end{bmatrix}^\top
         + 0.5 \cdot 
         \begin{bmatrix}
           0 \\
           1
         \end{bmatrix}\cdot 
         \begin{bmatrix}
           0 \\
           1
         \end{bmatrix}^\top\Bigg]^{-1} \nonumber\\
         &= \begin{bmatrix}
           0.5 & 0 \\
           0 & 0.5
         \end{bmatrix}^{-1} = \begin{bmatrix}
           2 & 0 \\
           0 & 2
         \end{bmatrix}.
         \nonumber
\end{align}

Next, by line~\ref{pseudocode:estimate_gradient}, we have
\begin{align}
    \widehat{\nabla}^1_{l_1}\Phi(m_1) &= \frac{1}{2}(
    \begin{bmatrix}
       1 \\
       0
    \end{bmatrix}^\top
    \begin{bmatrix}
           2 & 0 \\
           0 & 2
     \end{bmatrix}
     \begin{bmatrix}
       1 \\
       0
    \end{bmatrix}
    0.4
    ) \nonumber\\
    & + \frac{1}{2}(
    \begin{bmatrix}
       1 \\
       0
    \end{bmatrix}^\top
    \begin{bmatrix}
       2 & 0 \\
       0 & 2
     \end{bmatrix}
     \begin{bmatrix}
       0 \\
       1
    \end{bmatrix}
    0.8
    ) \nonumber \\
    &= 0.4. \nonumber
\end{align}
\begin{align}
    \widehat{\nabla}^1_{l_1}\Phi(m_2) &= \frac{1}{2}(
    \begin{bmatrix}
       0 \\
       1
    \end{bmatrix}^\top
    \begin{bmatrix}
       2 & 0 \\
       0 & 2
     \end{bmatrix}
     \begin{bmatrix}
       1 \\
       0
    \end{bmatrix}
    0.4
    ) \nonumber\\
    & + \frac{1}{2}(
    \begin{bmatrix}
       0 \\
       1
    \end{bmatrix}^\top
    \begin{bmatrix}
       2 & 0 \\
       0 & 2
     \end{bmatrix}
     \begin{bmatrix}
       0 \\
       1
    \end{bmatrix}
    0.8
    ) \nonumber \\
    &=0.8. \nonumber \\
    \widehat{\nabla}^1_{l_1}\Phi &= 
    \begin{bmatrix}
        \widehat{\nabla}^1_{l_1}\Phi(m_1) \\
        \widehat{\nabla}^1_{l_1}\Phi(m_2)
    \end{bmatrix} =
    \begin{bmatrix}
       0.4 \\
       0.8
    \end{bmatrix}. \nonumber
\end{align}

Next, by line~\ref{pseudocode:get_optimal_policy_by_gradient}, the inner products of each policy in $\Delta(P_{l_1})$ and $\widehat{\nabla}^1_{l_1}\Phi$ are  $0.56$, $0.60$, $0.68$, and $0.76$, and the optimal policy $\Tilde{\pi}^{2}_{l_1}$ should be $[0.1, 0.9]$ as its inner product with $\widehat{\nabla}^1_{l_1}\Phi$ is maximum. 

Lastly, by line~\ref{pseudocode:update_policy_with_exploration}, we have
\begin{align}
    \pi^{2}_{l_1} & = \alpha\big[\underbrace{\pi^1_{l_1}+\beta(\Tilde{\pi}^{2}_{l_1}-\pi^1_{l_1})}_{\text{Frank-Wolfe update}}\big] + \underbrace{(1-\alpha)\rho^1_{l_1}}_{\text{Exploration}} \nonumber\\
    &= 0.5\Bigg[
    \begin{bmatrix}
       0.5 \\
       0.5
    \end{bmatrix} + 
    0.5 \Bigg(
    \begin{bmatrix}
       0.1 \\
       0.9
    \end{bmatrix} - 
    \begin{bmatrix}
       0.5 \\
       0.5
    \end{bmatrix}
    \Bigg)
    \Bigg] +
    0.5 \begin{bmatrix}
       0.1 \\
       0.9
    \end{bmatrix} \nonumber\\
    & = \begin{bmatrix}
       0.2 \\
       0.8
    \end{bmatrix}. \nonumber
\end{align}

In the next episode, node $l_1$ follows policy $\pi^{2}_{l_1}$ to choose the next hops and collect rewards to update the policy. Other nodes follow the same process to update their policies in parallel. Also, if the packets intermediate nodes $m_1$ and $m_2$ receive have not reached their destinations, $m_1$ and $m_2$ follow the same process to forward the packets and update their policies in parallel.

{\color{dv}
\section{Observing Local Congestion in Totoro$^+$}
\label{app:local_congestion}
Observing local congestion at the network layer indeed requires support from lower-level instrumentation or third-party mechanisms. However, Totoro$^+$ is designed as an \textit{application-layer system}, and therefore congestion at a node is inferred based on \textit{application-level latency}, rather than network-layer indicators such as packet queues or link-level congestion metrics. Specifically, the overlay constructed by Totoro$^+$ operates at the application layer and forms a logical topology, while actual data transmission is carried out over the underlying physical network topology. In Totoro$^+$, congestion at an overlay node reflects a combination of physical network congestion and application-level task congestion. From a system design perspective, these effects can be jointly captured using end-to-end latency measurements, without requiring explicit knowledge of physical network congestion states.
As an illustrative example, a node in Totoro$^+$ may correspond to a drone, while another overlay node (e.g., a mobile phone) transmits application-level packets supported by Totoro$^+$. The congestion of interest in Totoro$^+$ is characterized by the response latency observed at the drone when handling these application-level requests, rather than by the underlying network transmission congestion itself. Therefore, although Totoro$^+$ abstracts away low-level congestion details, this abstraction remains practically meaningful for application-layer decision making.


\section{Implementation Details of Algorithm~\ref{pseudo:our_algorithm}}
\label{sec:impl_algo}


\begin{figure}[t!]
    \centering
    \includegraphics[width=0.9\linewidth]{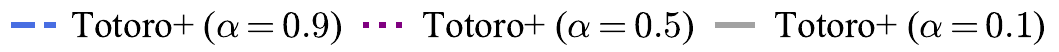}\vspace{-.05in}
    \includegraphics[width=0.9\linewidth]{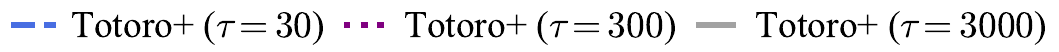}\\
    \begin{minipage}{0.24\textwidth}
        \centering
        \captionsetup{width=0.9\textwidth}
        \includegraphics[width=\textwidth]{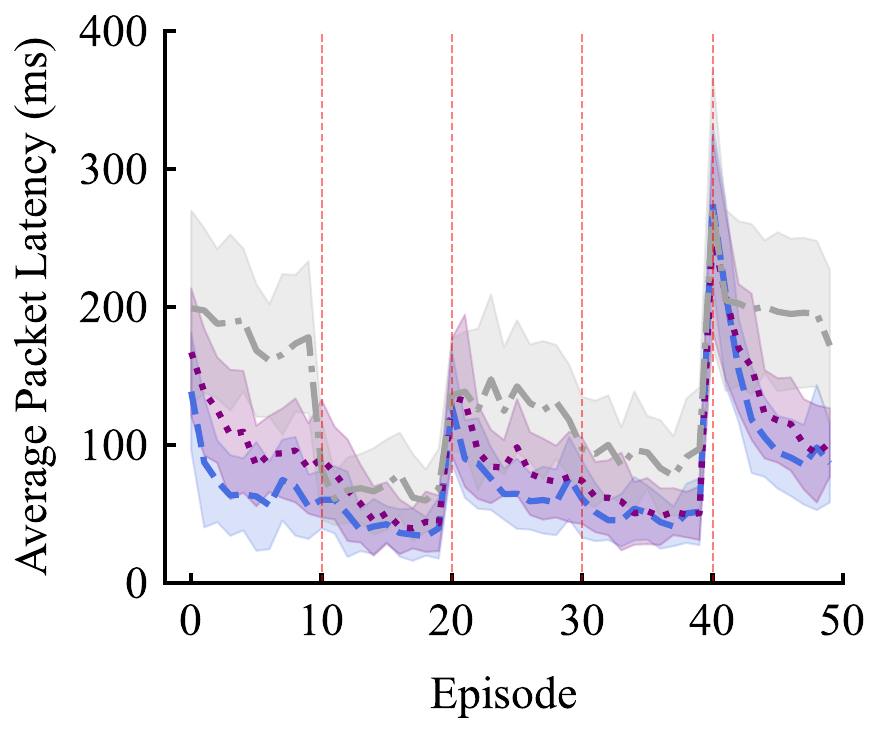} 
        \caption{Average packet latency with different $\alpha$.}
        \label{fig:failure_node_diff_alpha_c4}
    \end{minipage}
    \begin{minipage}{0.24\textwidth}
        \centering
        \captionsetup{width=0.9\textwidth}
        \includegraphics[width=\textwidth]{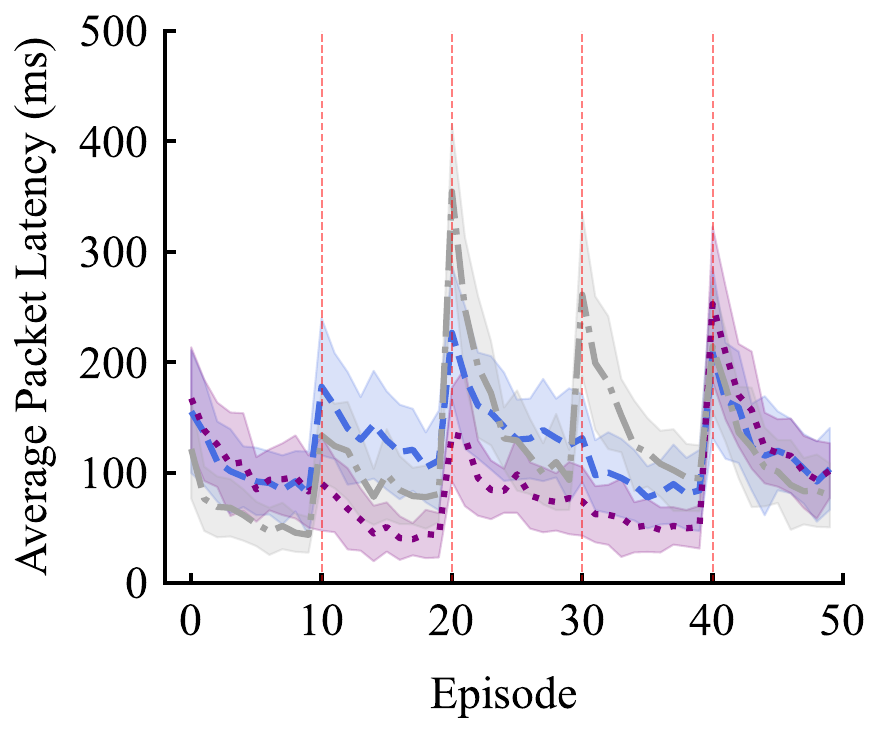}
        \caption{Average packet latency with different $\tau$.}
        \label{fig:failure_node_diff_tau_c4}
    \end{minipage}
    \vspace{0.1in}
\end{figure}

\begin{figure}[t!]
    \centering
    \includegraphics[width=0.9\linewidth]{figure/appendix/legend_tau_number_cropped.pdf}\\
    \begin{minipage}{0.24\textwidth}
        \centering
        \captionsetup{width=0.9\textwidth}
        \includegraphics[width=\textwidth]{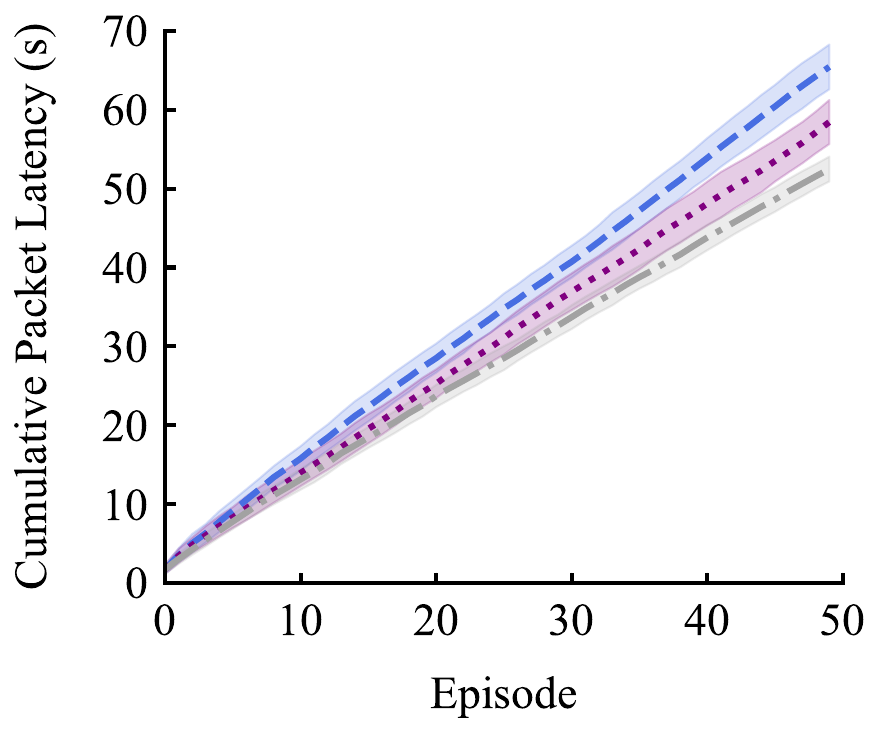} 
        \caption{Cumulative packet latency with different $\tau$.}
        \label{fig:packet_latency_diff_tau_c4}
    \end{minipage}
    \begin{minipage}{0.24\textwidth}
        \centering
        \captionsetup{width=0.9\textwidth}
        \includegraphics[width=\textwidth]{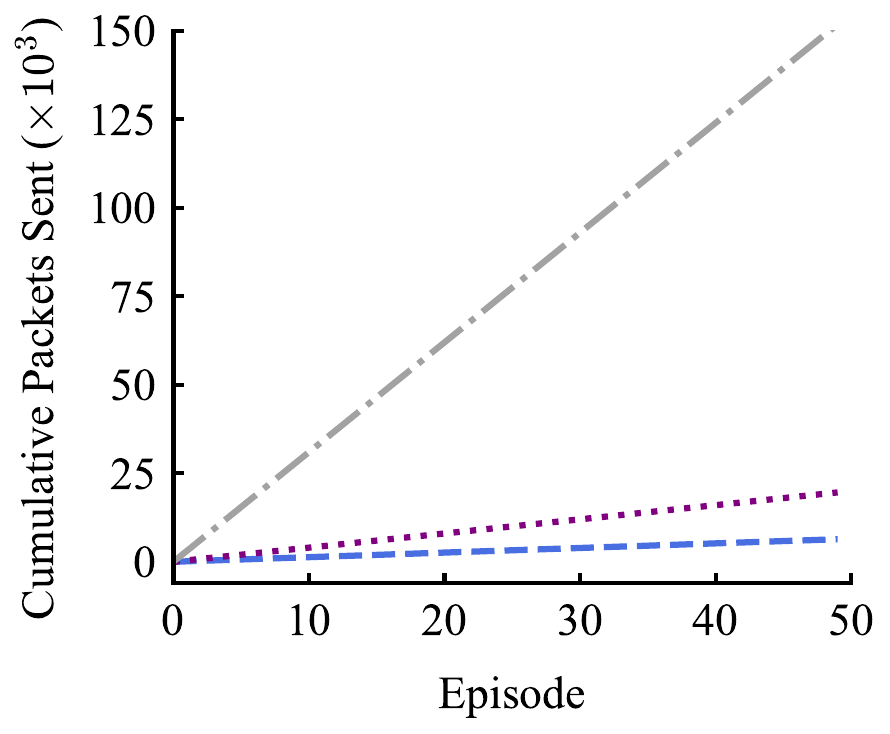}
        \caption{Cumulative packet sent with different $\tau$.}
        \label{fig:packet_sent_diff_tau_c4}
    \end{minipage}
    \vspace{-0.1in}
\end{figure}

\noindentparagraph{Latency uncertainty.}
In our implementation, the instantaneous reward is derived from the observed end-to-end latency of packet transmission. Since latency is a continuous-valued metric and may exhibit large variance across different network conditions, we normalize it into a bounded reward signal to stabilize the learning process. In addition, the end-to-end latency experienced by a packet is inherently affected by the joint routing decisions of multiple nodes: when multiple nodes simultaneously select the same next-hop node, increased contention and queuing effects can lead to higher latency. As a result, the observed latency naturally reflects both stochastic network conditions and the interaction among distributed routing decisions. 

Specifically, given an observed latency value $l$, we map it into the interval $[0,1]$ as 
\[
r = 1 - \frac{l}{l_{max}},
\]
where $l_{max}$ denotes the maximum latency observed or a predefined upper bound within the corresponding evaluation window. Under this mapping, lower latency yields higher reward, while larger latency is penalized proportionally. This normalization ensures a consistent reward scale across episodes and prevents extreme latency values from dominating the gradient estimation.

\noindentparagraph{Fluctuating bandwidth.}
Network bandwidth conditions may vary over time due to congestion, background traffic, or dynamic network environments. Totoro$^+$ naturally captures such dynamics through its episodic learning structure. At the beginning of each episode, rewards are re-sampled based on the current network conditions, implicitly reflecting any changes in available bandwidth. Consequently, when bandwidth fluctuates, the reward observations in subsequent episodes change accordingly, and nodes adapt their routing policies by updating local gradient estimates and re-optimizing hop-selection probabilities. This design allows Totoro$^+$ to operate effectively in non-stationary environments without requiring explicit bandwidth prediction or manual reconfiguration. 

To evaluate the robustness of Totoro$^+$ under bandwidth fluctuations, we conduct the following experiments. Specifically, the bandwidth of each node is randomly perturbed every ten episodes, with possible increases or decreases. We examine the impact of different values of the parameters $\alpha$ (Figure~\ref{fig:failure_node_diff_alpha_c4}) and $\tau$ (Figure~\ref{fig:failure_node_diff_tau_c4}) on the average packet latency in a routing scenario with 100 nodes. For a fair comparison, we focus exclusively on the latency observed after the sampling phase. Figure~\ref{fig:failure_node_diff_alpha_c4} shows that smaller values of $\alpha$ result in higher average packet latency. This is because a smaller $\alpha$ encourages exploration toward less familiar nodes, leading to slower convergence after environmental changes; nevertheless, the latency still gradually converges to a lower level over time. Figure~\ref{fig:failure_node_diff_tau_c4} illustrates the effect of varying the number of sampled rewards on latency. The results indicate that an appropriate number of reward samples enables faster latency reduction, as excessive exploration or insufficient exploration can lead to overconfidence or underconfidence in path selection, respectively.

\noindentparagraph{Transmission overhead.}
Totoro$^+$ is designed to limit the communication overhead incurred during learning. Each node performs reward sampling only at the beginning of an episode, and policy updates are based on at most $\tau$ packet transmissions within that episode. After reward sampling, each node updates its local routing policy solely based on its own observed rewards, without requiring any additional message passing or coordination with other nodes. As a result, the total number of packets generated by each node per episode is upper-bounded by $\tau$, independent of the network size or the number of available routing options. This bounded sampling strategy enables a controllable trade-off between learning accuracy and communication cost: larger $\tau$ provides more accurate gradient estimates at the cost of higher transmission overhead, while smaller $\tau$ reduces communication cost but increases estimation variance. 

To study how $\tau$ affects routing performance, we conduct routing experiments with 100 nodes under different values of $\tau$. For a fair comparison, we focus exclusively on the latency observed after the sampling phase. Under the cumulative packet latency metric (Figure~\ref{fig:packet_latency_diff_tau_c4}), larger values of $\tau$ clearly achieve lower overall latency, since sufficient exploration allows nodes to make more informed routing decisions. However, under the cumulative number of packets sent metric (Figure~\ref{fig:packet_sent_diff_tau_c4}), overly large values of $\tau$ lead to a higher total number of transmitted packets, prolonging the sampling phase. Therefore, achieving a balance between the overhead incurred during the sampling phase and the final routing performance remains an important and nontrivial design consideration.

\noindentparagraph{Time complexity.}
The time complexity of Totoro$^+$ is formally analyzed in Theorem~2, with the detailed proof provided in Appendix~\ref{sec:proof_time_complexity}. In Section~VII-E, we further compare the empirical running time of Totoro$^+$ against the baseline Totoro algorithm and provide a breakdown of the time spent in different stages of the algorithm to illustrate how the algorithmic components introduced in Algorithm~\ref{pseudo:our_algorithm} contribute to the overall runtime.
}


\section{Proof of Theorem~1}
\label{sec:proof_nash_regret}

\begin{lemma}
    Suppose Algorithm~\ref{pseudo:our_algorithm} runs $K$ episode.
    For any episode $k\in [K], n\in [N]$ and $p\in P_n$, we have
    \begin{align*}
        |\widehat{\nabla}^k_n \Phi(p)| \leq \frac{c_1}{(1-\alpha)},
    \end{align*}
    where $c_1$ is a constant.
    \label{lemma:bound_of_estimator}
\end{lemma}
\begin{proof}
    As $\pi^k_n = \alpha\big[\pi^k_n+\beta(\Tilde{\pi}^{k+1}_n-\pi^k_n)\big] + (1-\alpha) \rho_n^k$, we have
    \begin{align}    
        M(\pi^k_n) &= \sum_{p\in \pi^k_n} \pi^k_n(p)\psi(p)\psi(p)^\top \nonumber\\
        &\succeq (1-\alpha) \sum_{p\in \rho^k_n} \rho^k_n(p)\psi(p)\psi(p)^\top,
        \nonumber
    \end{align}
    for any hop $p\in P_n$, we get
    \begin{align}
        \|\psi(p)\|^2_{M(\pi^k_n)^{-1}} &\leq \frac{1}{(1-\alpha)}\|\psi(p)\|^2_{M(\rho^k_n)^{-1}} \nonumber\\
        &\leq \frac{c_1}{(1-\alpha)},
        \nonumber
    \end{align}
    where $c_1$ is a constant. Then for any $t\in [\tau]$, since $r^{k,t}_n \leq 1$, we have
    \begin{align}
        &|r^{k,t}_n \psi(p)^\top M(\pi^k_n)^{-1}\psi(p^{k,t}_i)| \nonumber\\
        &\leq r^{k,t}_i\|\psi(p)\|_{M(\pi^k_n)^{-1}} \|\psi(p^{k,t}_i)\|_{M(\pi^k_n)^{-1}} \nonumber\\
        &\leq \frac{c_1}{(1-\alpha)}.
        \nonumber
    \end{align}
    As a result, we have
    \begin{align}
        |\widehat{\nabla}^k_n \Phi(p)| 
        &= \bigg|\frac{1}{\tau}\sum_{t=1}^\tau \psi(p)^\top M(\pi^k_n)^{-1}\psi(p^{k,t}_n)r^{k,t}_n\bigg| \nonumber \\
        &\leq \frac{c_1}{(1-\alpha)}.
        \nonumber
    \end{align}
\end{proof}

\begin{lemma}
    Suppose Algorithm~1 runs $K$ episode.
    For any episode $k\in [K], n\in [N]$ and $p\in P_n$, we have
    \begin{align*}
        \mathbb{E}_k[(\psi(p)^\top M(\pi^k_n)^{-1}\psi(p^{k,t}_n)r^{k,t}_n)^2 ] \leq \frac{c_2}{(1-\alpha)},
        \nonumber
    \end{align*}
    where $c_2$ is a constant.\label{lemma:bound_of_estimator_second_order}
\end{lemma}
\begin{proof}
    We first show that for any $t\in [\tau]$, we have
    \begin{align}
        &\mathbb{E}_k[(\psi(p)^\top M(\pi^k_n)^{-1}\psi(p^{k,t}_n)r^{k,t}_n)^2 ] \nonumber\\
        & \leq \mathbb{E}_k[(\psi(p)^\top M(\pi^k_n)^{-1}\psi(p^{k,t}_n))^2 ] \nonumber\\
        & \leq \mathbb{E}_k[\psi(p)^\top M(\pi^k_n)^{-1}\psi(p^{k,t}_n)\psi(p^{k,t}_n)^\top M(\pi^k_n)^{-1}\psi(p)^\top]\nonumber\\
        & = \psi(p)^\top M(\pi^k_n)^{-1}\psi(p) \nonumber\\
        & \leq \frac{c_2}{(1-\alpha)}.
        \nonumber
    \end{align}
\end{proof}

\begin{lemma}
    Suppose Algorithm~\ref{pseudo:our_algorithm} runs $K$ episode. With probability $1-\delta$, for any episode $k\in [K], n\in [N]$ and $p\in P_n$, we have
    \begin{align}
        |\widehat{\nabla}^k_n \Phi(p) - \nabla^k_n \Phi(p)| 
        &\leq c\sqrt{\frac{\mathcal{O}(\log N)\log{(\frac{NK}{\delta})}}{(1-\alpha) \tau}} \nonumber\\
        &\quad + c\frac{\mathcal{O}(\log N)\log{(\frac{NK}{\delta})}}{(1-\alpha) \tau},
        \nonumber
    \end{align}
    where $c$ is a constant.
    \label{lemma:estimator_bound}
\end{lemma}
\begin{proof}
    Recall that 
    \begin{align}
        \widehat{\nabla}^k_i \Phi(p) = \frac{1}{\tau}\sum_{t=1}^\tau \psi(p)^\top M(\pi^k_n)^{-1}\psi(p^{k,t}_n)r^{k,t}_n,
        \nonumber
    \end{align}
     and $(p^{k,t}_n, r^{k,t}_n)$ are drawn independently at each $t\in [\tau]$. As we use linear regression to estimate gradients, $\widehat{\nabla}^k_n \Phi(a)$ is unbiased~\cite{congestion_games,bandit_algo_book}, i.e., $\mathbb{E}_k[\widehat{\nabla}^k_n \Phi(a)]=\nabla^k_n \Phi(a)$. In addition, Lemma~\ref{lemma:bound_of_estimator} shows that $\psi(p)^\top M(\pi^k_n)^{-1}\psi(p^{k,t}_n)r^{k,t}_n$ is bounded by $c_1/(1-\alpha)$ and Lemma~\ref{lemma:bound_of_estimator_second_order} shows that its second moment is bounded by $c_2/(1-\alpha)$. Then by Bernstein's inequality, for a fixed $k\in [K], n\in [N]$, and $p\in P_n$, with probability $1-\delta$, we have
    \begin{align}
        |\widehat{\nabla}^k_n \Phi(p) - \nabla^k_n \Phi(p)| \leq \sqrt{\frac{4c_2\log{(\frac{2}{\delta})}}{(1-\alpha) \tau}} + \frac{4c_1\log{(\frac{2}{\delta})}}{3(1-\alpha) \tau}.
        \label{eq:estimation_bound_with_fixed_episode}
    \end{align}
    Then we extend Ineq. (\ref{eq:estimation_bound_with_fixed_episode}) to all $k\in [K], n\in [N]$ and $p\in P_n$, yielding
    \begin{align}
        |\widehat{\nabla}^k_n \Phi(p) - \nabla^k_n \Phi(p)| 
        &\leq c\sqrt{\frac{\mathcal{O}(\log N)\log{(\frac{NK}{\delta})}}{(1-\alpha) \tau}} \nonumber\\
        &\quad + c\frac{\mathcal{O}(\log N)\log{(\frac{NK}{\delta})}}{(1-\alpha) \tau},
        \nonumber
    \end{align}
    where $c$ is a constant.
\end{proof}

\begin{lemma}
    $\Phi(\cdot)$ is N-Lipschitz and N-smooth with respect to the L1 norm $\|\cdot\|_1$.
    \label{lemma:lip_and_smooth}
\end{lemma}
\begin{proof}
    By~\cite{congestion_games,monderer1996potential}, we know that $\Phi(\pi)=\mathbb{E}_{\textbf{p}\sim \pi} \Phi(\textbf{p})$ and $\Phi(\textbf{p})\in [0,N]$, where $\textbf{p}$ is the joint hop chosen by all the nodes.
    \begin{align}
        \Phi(\pi) - \Phi(\pi')
        & = \mathbb{E}_{\textbf{p}\sim \pi} - \mathbb{E}_{\textbf{p}\sim \pi'} \nonumber\\
        & = \sum_{n\in [N]} \Big[ \mathbb{E}_{p_{1:n-1}\sim \pi'_{1:n-1},p_{n:N}\sim \pi_{n:N}}\Phi(\textbf{p}) \nonumber\\
        & \quad - \mathbb{E}_{p_{1:n}\sim \pi'_{1:n},p_{n+1:N}\sim \pi_{n+1:N}}\Phi(\textbf{a}) \Big] \nonumber\\
        & \leq \sum_{n\in [N]} \|\pi_n - \pi'_n\|_1\cdot \|\Phi\|_{\infty}\nonumber\\
        & \leq N\|\pi - \pi'\|_1.
        \nonumber
    \end{align}
    In addition, $\nabla_\pi \Phi(p_i) = \mathbb{E}_{p_{-i}\sim\pi_{-i}}\Phi(p_i,p_{-i})$, yielding
    \begin{align}
        \|\nabla_\pi \Phi - \nabla_{\pi'} \Phi\|_\infty \leq N\|\pi-\pi'\|_1.
        \nonumber
    \end{align}
\end{proof}

Finally, we can derive the following theorem with Lemma~\ref{lemma:bound_of_estimator}, Lemma~\ref{lemma:bound_of_estimator_second_order}, Lemma~\ref{lemma:estimator_bound}, and Lemma~\ref{lemma:lip_and_smooth}.

\begin{theorem}
    Let $T=k\tau$ and assume that each policy in $\Delta(P_n)$ has no zero element for all $n$. By running Algorithm~\ref{pseudo:our_algorithm} with gradient estimator $\widehat{\nabla}^k_n\Phi(p)$ defined in line~\ref{pseudocode:estimate_gradient} and exploratory policy $\rho^k_n$ defined in line~\ref{pseudocode:d_optimal_design}, if $k\geq \frac{2}{N}$ , then with probability $1-\delta$, we have    
    \begin{align}
        \text{Nash-Regret}(T) \leq \Tilde{\mathcal{O}}(N^2 T^{5/6}\log N).\nonumber
    \end{align}
    \label{theo:nash_regret_bound_appendix}
\end{theorem}
\begin{proof}
    Set $\nabla^k \Phi = \nabla \Phi(\prod^k)$. By Lemma~\ref{lemma:lip_and_smooth}, we have
    \begin{align}
        \Phi(\pi^{k+1}) &\geq \Phi(\pi^{k}) + \langle \nabla \Phi(\pi^k), \pi^{k+1} - \pi^{k} \rangle \nonumber\\
        & \quad - \frac{N}{2}\|\pi^{k+1}-\pi^{k}\|_1^2
        \nonumber \\
        & = \Phi(\pi^{k}) + \alpha\beta\langle \nabla \Phi(\pi^k), \Tilde{\pi}^{k+1}-\pi^k\rangle \nonumber\\
        & \quad + (1-\alpha) \langle \nabla^k\Phi, \rho^k-\pi^k\rangle \nonumber\\
        & \quad - \frac{N}{2}\|\alpha\beta(\Tilde{\pi}^k-\pi^k)+(1-\alpha)(\rho^k-\pi^k)\|^2_1 \nonumber\\
        & \geq \Phi(\pi^{k}) + \alpha\beta\langle \nabla \Phi(\pi^k), \Tilde{\pi}^{k+1}-\pi^k\rangle \nonumber\\
        & \quad - (1-\alpha )\|\nabla^k\Phi\|_\infty\|\rho^k-\pi^k\|_1 \nonumber\\
        & \quad - \frac{N}{2}\|\alpha\beta(\Tilde{\pi}^k-\pi^k)+(1-\alpha)(\rho^k-\pi^k)\|^2_1 \nonumber\\
        & \geq \Phi(\pi^{k}) + \alpha\beta\langle \nabla \Phi(\pi^k), \Tilde{\pi}^{k+1}-\pi^k\rangle \nonumber\\
        & \quad - 2(1-\alpha) N^2\mathcal{O}(\log N) \nonumber\\
        & \quad - \frac{N}{2}\|\alpha\beta(\Tilde{\pi}^k-\pi^k)+(1-\alpha)(\rho^k-\pi^k)\|^2_1 \nonumber\\
        & \geq \Phi(\pi^{k}) + \alpha\beta\langle \nabla \Phi(\pi^k), \Tilde{\pi}^{k+1}-\pi^k\rangle \nonumber\\
        & \quad - 2(1-\alpha) N^2\mathcal{O}(\log N) \nonumber\\
        & \quad - 2N^3(\alpha^2 + \beta^2)
        \nonumber
    \end{align}
    Define the true target policy at episode $k$ as follows
    \begin{align}
        \widehat{\pi}^{k+1}_n = \argmax_{\pi_n} \langle \pi_n, \nabla_n \Phi(\pi^k_n) \rangle,
        \nonumber
    \end{align}
    and the Frank Wolfe gap of join strategy $\pi$
    \begin{align}
        G(\pi) = \max_{\pi'} \langle \pi'-\pi, \nabla\Phi(\pi)\rangle
        \nonumber
    \end{align}
    according to~\cite{congestion_games}.
    Then we have 
    \begin{align}
        & \langle \nabla \Phi(\pi^k), \Tilde{\pi}^{k+1}-\pi^k\rangle \nonumber\\
        &= \langle \widehat{\nabla}\Phi(\pi^k),\Tilde{\pi}^{k+1}-\pi^k\rangle \nonumber\\
        & \quad + \langle \nabla \Phi(\pi^k) - \widehat{\nabla}\Phi(\pi^k), \Tilde{\pi}^{k+1}-\pi^k\rangle \nonumber\\
        & \geq \langle \widehat{\nabla}\Phi(\pi^k),\widehat{\pi}^{k+1}-\pi^k\rangle \nonumber\\
        & \quad + \langle \nabla \Phi(\pi^k) - \widehat{\nabla}\Phi(\pi^k), \Tilde{\pi}^{k+1}-\pi^k\rangle \nonumber\\
        & = \langle \nabla\Phi(\pi^k),\widehat{\pi}^{k+1}-\pi^k\rangle \nonumber\\
        & \quad + \langle \nabla \Phi(\pi^k) - \widehat{\nabla}\Phi(\pi^k), \Tilde{\pi}^{k+1}-\widehat{\pi}^{k+1}\rangle \nonumber\\
        & \geq G(\pi^k) \nonumber\\
        & \quad - \|\nabla \Phi(\pi^k) - \widehat{\nabla}\Phi(\pi^k)\|_\infty \|\Tilde{\pi}^{k+1}-\widehat{\pi}^{k+1}\|_1\nonumber\\
        & \geq G(\pi^k)  - 2N\|\nabla \Phi(\pi^k) - \widehat{\nabla}\Phi(\pi^k)\|_\infty  \nonumber\\
        & \geq G(\pi^k)  - \sqrt{\frac{4N^2\mathcal{O}(\log N)c\log{(\frac{NK}{\delta})}}{(1-\alpha) \tau}} \nonumber\\
        & \quad - \frac{2N\mathcal{O}(\log N)c\log{(\frac{NK}{\delta})}}{(1-\alpha) \tau}
        \nonumber
    \end{align}
    Apply it to the previous bound, and we have
    \begin{align*}
        \Phi(\pi^{k+1}) &\geq \Phi(\pi^{k}) + \alpha\beta G(\pi^k) \nonumber\\
        & \quad - \frac{\alpha\beta}{\sqrt{(1-\alpha)\tau}}\sqrt{4N^2\mathcal{O}(\log N)c\log{(NK/\delta)}} \nonumber\\
        & \quad - \frac{\alpha\beta}{(1-\alpha)\tau} 2N\mathcal{O}(\log N)c\log{(NK/\delta)} 
        \nonumber\\
        & \quad - 2(1-\alpha) N^2\mathcal{O}(\log N) -2N^3(\alpha^2+\beta^2).
    \end{align*}
    Summing over $k\in [K]$ and we have
    \begin{align*}
        \sum_{k=1}^K G(\pi^k) &\leq \frac{\Phi(\pi^{K+1})-\Phi(\pi^1)}{\alpha\beta} \nonumber\\
        & \quad + \frac{K}{\sqrt{(1-\alpha)\tau}}\sqrt{4N^2\mathcal{O}(\log N)c\log{(NK/\delta)}} \nonumber\\
        & \quad + \frac{K}{(1-\alpha)\tau} 2N\mathcal{O}(\log N)c\log{(NK/\delta)} 
        \nonumber\\
        & \quad + \frac{2K(1-\alpha) N^2\mathcal{O}(\log N)}{\alpha\beta} \\
        & \quad + \frac{2KN^3(\alpha^2+\beta^2)}{\alpha\beta}.
    \end{align*}
    Set $(1-\alpha)=\frac{1}{NK}$, $\beta=\frac{1}{N\sqrt{K}}$, and $\tau=K^2$ and notice that when $K \geq \frac{2}{N}$, we have $\alpha \geq \frac{1}{2}$ and 
    \begin{align*}
        \sum_{k=1}^K G(\pi^k) = \Tilde{\mathcal{O}}(N^2 K^{1/2}\log N).
    \end{align*}
    Let $T=K\tau$,
    we have
    \begin{align*}
        \text{Nash-Regret}(T)=\tau \sum_{k=1}^K G(\pi^k) =\Tilde{\mathcal{O}}(N^2 T^{5/6}\log N).
    \end{align*}
\end{proof}

{\color{dv}
\section{Proof of Corollary~2}
\label{app:proof_corollary2}
\begin{assumption}
\label{assump:async}
Each node $n$ maintains a local episode counter $k_n$ and performs updates at a sequence of (wall-clock) times.
When node $n$ performs its $(k_n+1)$-th update, it uses only its own local history
$\{(p_n^{k,t}, r_n^{k,t})\}_{t=1}^{\tau}$ collected since its previous update and its current policy $\pi_n^{k_n}$.
Moreover, (i) no inter-node communication is used, and (ii) the delay between two consecutive updates of any node is bounded:
there exists $D\ge 1$ such that between any two updates of node $n$, every other node updates at most $D$ times.
\end{assumption}

\begin{lemma}
\label{lemma:async_regret}
Suppose Assumption~\ref{assump:async} holds.
Assume further that each node's expected reward is Lipschitz in other nodes' mixed policies:
for every node $n$ and action $p\in P_n$,
\begin{align}
\big|\mathbb{E}[r_n \mid p, \pi_{-n}] - \mathbb{E}[r_n \mid p, \pi'_{-n}]\big|
\le L \sum_{m\neq n}\|\pi_m-\pi'_m\|_1 .
\label{eq:lipschitz_env}
\end{align}
Let $\eta:=\sup_{m}\sup_{\text{one update}}\|\pi_m^{\text{new}}-\pi_m^{\text{old}}\|_1$ be the maximum per-update policy change.
Then the Nash regret of the asynchronous execution satisfies, with probability $1-\delta$,
\begin{align}
\text{Nash-Regret}_{\mathrm{async}}(T) 
&\le \text{Nash-Regret}_{\mathrm{sync}}(T) \nonumber \\
& \quad + \widetilde{\mathcal{O}}(L N D \eta \, T).
\label{eq:async_regret_bound}
\end{align}
In particular, if $D\eta=o(1)$ (e.g., via diminishing stepsizes), then the regret remains sublinear in $T$.
\end{lemma}

\begin{proof}
Fix a node $n$. Consider a local block of $\tau$ packets collected by node $n$ between two consecutive local updates.
Let $\pi_{-n}(t)$ denote the mixed policies of other nodes at packet time $t$ within this block,
and let $\bar{\pi}_{-n}:=\pi_{-n}(1)$ be the other nodes' policy vector at the beginning of the block.

By Assumption~\ref{assump:async}, during this block, each other node $m\neq n$ performs at most $D$ updates.
Since each update changes its mixed policy by at most $\eta$ in $\ell_1$ norm, we have for all $t$,
\begin{align}
\|\pi_m(t)-\bar{\pi}_m\|_1 \le D\eta,\qquad \forall m\neq n,
\end{align}
and hence
\begin{align}
\sum_{m\neq n}\|\pi_m(t)-\bar{\pi}_m\|_1 \le (N-1)D\eta \le ND\eta.
\end{align}

By the Lipschitz condition \eqref{eq:lipschitz_env}, for any action $p_n(t)$ chosen by node $n$ at time $t$,
\begin{align}
&\left|\mathbb{E}[r_n(t)\mid p_n(t), \pi_{-n}(t)] - \mathbb{E}[r_n(t)\mid p_n(t), \bar{\pi}_{-n}]\right| \nonumber\\
&\quad \le L \sum_{m\neq n}\|\pi_m(t)-\bar{\pi}_m\|_1
\le LND\eta .
\end{align}

Therefore, relative to the synchronous (frozen) environment $\bar{\pi}_{-n}$ within the block,
the gradient estimator incurs an additional bias term of order $\mathcal{O}(LND\eta)$.
Plugging this bias bound into the regret in Appendix~\ref{sec:proof_nash_regret} yields \eqref{eq:async_regret_bound}.
\end{proof}
}
\section{Proof of Theorem~2}
\label{sec:proof_time_complexity}

\begin{proof}
    We analyze the time complexity of Algorithm~\ref{pseudo:our_algorithm} line by line. Whenever a node receives packets, and it is not the destination of the packets, it first runs line~\ref{pseudocode:collect_info} to collect rewards. The node selects the next hop out of $\mathcal{O}(\log N)$ possible candidate nodes for $\tau$ packets with the time complexity being $\mathcal{O}(\tau\log N)$. 
    
    Next, the node runs line~\ref{pseudocode:d_optimal_design} to identify an exploratory matrix. Since each node has at most $\mathcal{O}(\log N)$ possible candidate nodes, the time complexity of calculating the policy correlation matrix $M(\lambda)$ by Eq.~({\color{blue}3}) is $\mathcal{O}(\log^3N)$, and $M(\lambda)$ is a $\mathcal{O}(\log N)$ by $\mathcal{O}(\log N)$ matrix so the time complexity of calculating the determinant of $M(\lambda)$ is also $\mathcal{O}(\log^3N)$. The time complexity of line~\ref{pseudocode:d_optimal_design} is $\mathcal{O}(|\Delta (P_n)|\log^3N)$ as the node identifies at most $|\Delta (P_n)|$ policies in $\Delta (P_n)$ to identify the minimum.

    In line~\ref{pseudocode:estimate_gradient}, the time complexity of calculating the inverse of the policy correlation matrix $M(\pi^k_n)^{-1}$ only once is $\mathcal{O}(\log^3N)$, and that of the matrix multiplication $\psi(p)^\top M(\pi^k_n)^{-1}\psi(p^{k,t}_n)r^{k,t}_n$ is $\mathcal{O}(\log^2N + \log N)$. Since the node estimates gradients over $\tau$ rewards for at most $\mathcal{O}(\log N)$ candidate nodes, the time complexity of calculating gradient estimator $\widehat{\nabla}^k_n\Phi$ is $\mathcal{O}(
    \tau \log^3N)$.

    Identifying the optimal policy $\Tilde{\pi}^{k+1}_n$ in line~\ref{pseudocode:get_optimal_policy_by_gradient} calculates the dot product of two $\mathcal{O}(\log N)$-dimensional vectors at most $|\Delta(P_n)|$ times, leading to $\mathcal{O}(|\Delta(P_n)|\log N)$ of time complexity. For line~\ref{pseudocode:update_policy_with_exploration}, the time complexity of the linear combination of $\mathcal{O}(\log N)$-dimensional vectors is $\mathcal{O}(\log N)$.

    Finally, since each node runs Algorithm~\ref{pseudo:our_algorithm} in parallel, the time complexity of updating a policy in an episode is $\mathcal{O}(\tau \log^3N + |\Delta (P_n)| \log^3N)$.
    
\end{proof}

{\color{dv}
{\color{dv}
\section{Adaptivity to Edge Heterogeneity}
\label{app:adap_edge_hetero}
\noindentparagraph{Heterogeneous computational capacity.}
We add a section discussing how Totoro$^+$ handles heterogeneous computational capacities in Appendix~\ref{sec:impl_algo}, along with additional experiments demonstrating that Totoro$^+$ can effectively adapt to heterogeneous computational capacities in terms of both computation and communication overhead.

\noindentparagraph{Heterogeneous bandwidth.}
To address heterogeneous bandwidth conditions, Totoro$^+$ adopts a game-theoretic routing model that enables nodes to adaptively identify suitable next-hop nodes for FL task transmission. The algorithmic details and theoretical analysis are presented in Section~{\color{blue}V-C}.

\noindentparagraph{Asynchronous updates.}
Totoro$^+$ supports two forms of asynchronous updates
\begin{enumerate}
    \item \textit{Asynchronous model updates in FL.} The model update scheme is configurable by the application owner or developer. After an application is published, the FL aggregation mode can be configured as asynchronous~\cite{async_fed_opt} or semi-synchronous~\cite{semi_sync_fl} according to application requirements. In addition, when nodes subscribe to an application, their computational capacity and bandwidth characteristics can be examined before admitting them into the candidate participant pool. 
    This configurability represents an additional customization capability of Totoro$^+$, beyond its scalability and adaptivity advantages.

    \item \textit{Asynchronous path selection.} In Algorithm~\ref{pseudo:our_algorithm}, the update procedure is asynchronous. Each node first explores candidate paths for $\tau$ rounds and then selects subsequent paths based on the feedback obtained during these $\tau$ rounds. When multiple nodes select the same path, the reward observed in that round is reduced. Moreover, when updating the policy $\pi$, Algorithm~\ref{pseudo:our_algorithm} employs a linear regression model to predict the gradient direction and adjust the policy accordingly, without requiring coordination with other nodes. As a result, the proposed game-theoretic algorithm in Totoro$^+$ inherently supports asynchronous updates and is well suited to the uncertainty and dynamics of edge networks.
\end{enumerate}
}

\section{Adapt to edge nodes with heterogeneous resources.}
\label{app:adap_edge_hetero_resource}
We provide the following mechanism to adapt to edge nodes with heterogeneous resources.

\begin{figure*}[t!]
        \centering
        \includegraphics[width=0.15\linewidth]{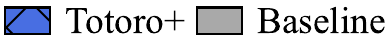}
        \vspace{-0.1in}
\end{figure*}
\begin{figure*}[t!]
\centering
    \begin{minipage}{0.31\textwidth}
        \captionsetup{width=\textwidth} 
        \includegraphics[width=\textwidth]{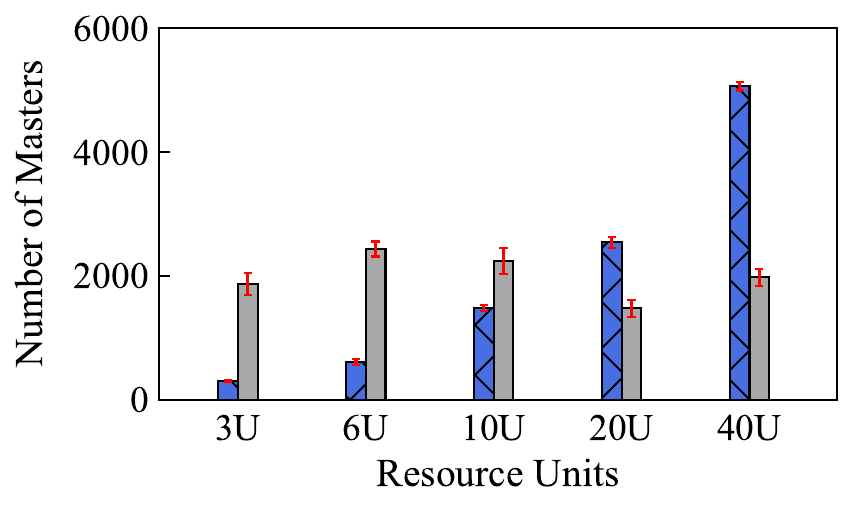}
        \caption{Totoro$^+$ allocates more masters to the node with more resource units.}
        \label{fig:number_masters_r2_r4}
    \end{minipage}\hspace{.1in}
    \begin{minipage}{0.31\textwidth}
        \captionsetup{width=\textwidth} 
        \includegraphics[width=\textwidth]{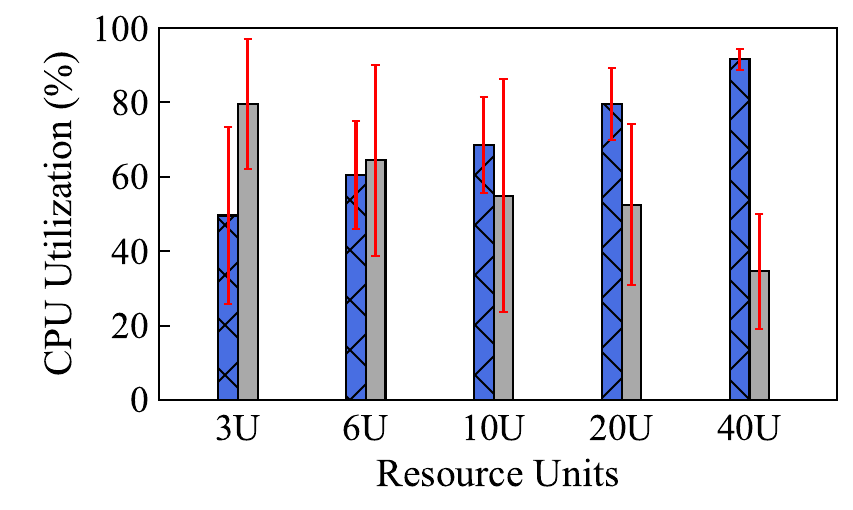}
        \caption{Totoro$^+$ keeps a high CPU utilization in the node with more resource units.}
        \label{fig:cpu_utilizations_r2_r4}
    \end{minipage}
    \vspace{.1in}
\end{figure*}

\begin{figure*}[h]
        \centering
        \includegraphics[width=0.25\linewidth]{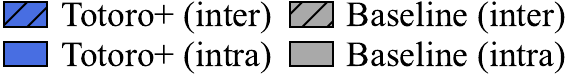}
        \vspace{-0.1in}
\end{figure*}
\begin{figure*}[h]
\centering
    \begin{minipage}{0.31\textwidth}
        \captionsetup{width=\textwidth} 
        \includegraphics[width=\textwidth]{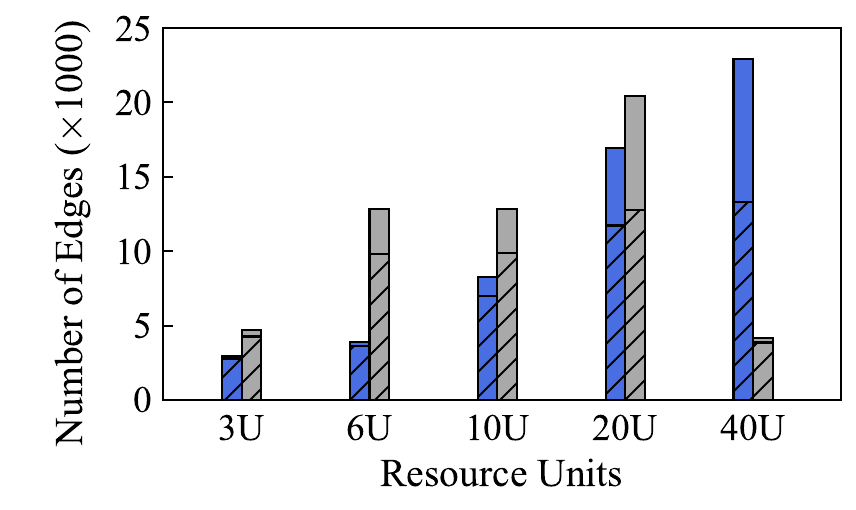}
        \caption{Totoro$^+$ enables more intra- and inter-node communication edges.}
        \label{fig:number_edges_r2_r4}
    \end{minipage}\hspace{.1in}
    \begin{minipage}{0.31\textwidth}
        \captionsetup{width=\textwidth} 
        \includegraphics[width=\textwidth]{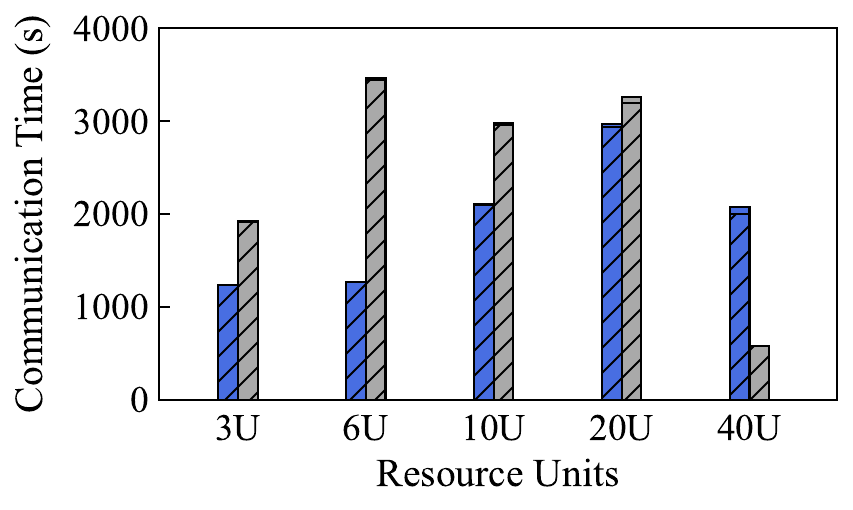}
        \caption{Totoro$^+$ reduces total intra- and inter-node communication times.}
        \label{fig:communication_times_r2_r4}
    \end{minipage}
\end{figure*}

A P2P node in Totoro$^+$ can be seen as a \textit{logical} node. We can map a physical edge node with rich resources to more P2P nodes in Totoro$^+$, whereas mapping resource-constrained edge nodes to fewer P2P nodes. For example, suppose three physical edge nodes with 2, 4, and 8 CPU cores, respectively. Then, the nodes with 4 and 8 CPU cores can serve as 2 and 3 logical P2P nodes in the DHT-based P2P overlay network, respectively. One can follow a similar way for different resources such as memory, network bandwidths, or GPUs.

\begin{figure*}[t!]
  \centering
  \vspace{0.15in}
  \captionsetup[subfigure]{width=4.1cm}
  \subfloat[Google Speech~\cite{google_speech_dataset} with ResNet-18~\cite{residualnet}]{\includegraphics[width=0.24\textwidth]{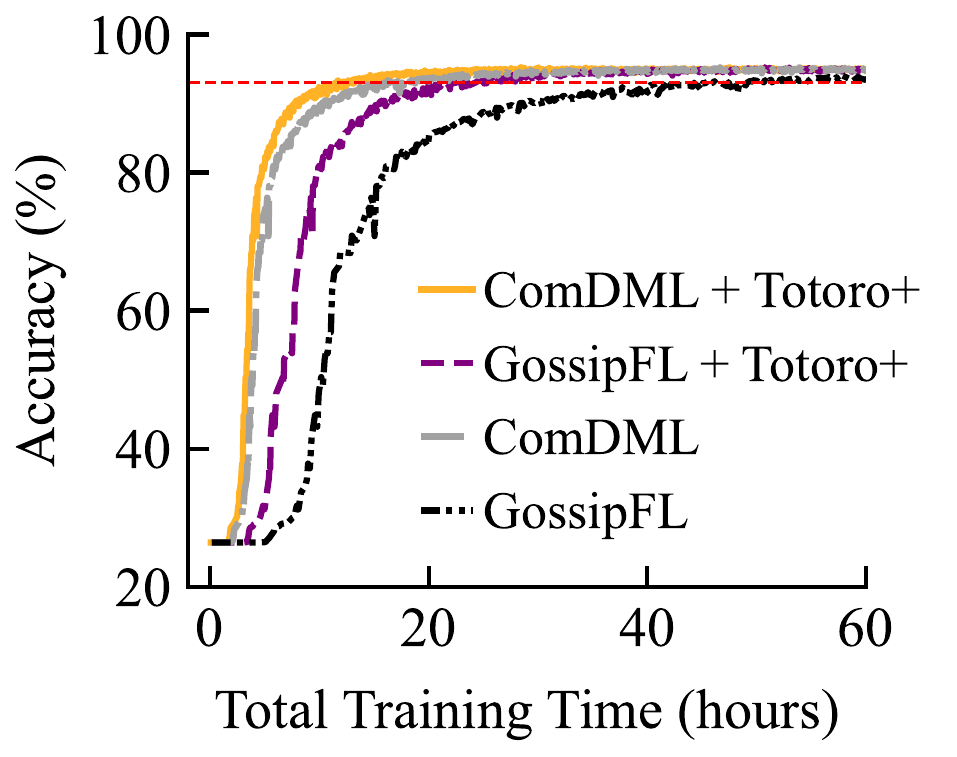}
  \label{fig:defl_speech_r1_c6}}
\hfill
  \subfloat[FEMNIST~\cite{femnist_dataset} with ShuffleNet~V2~\cite{shufflenet_v2}]{\includegraphics[width=0.24\textwidth]{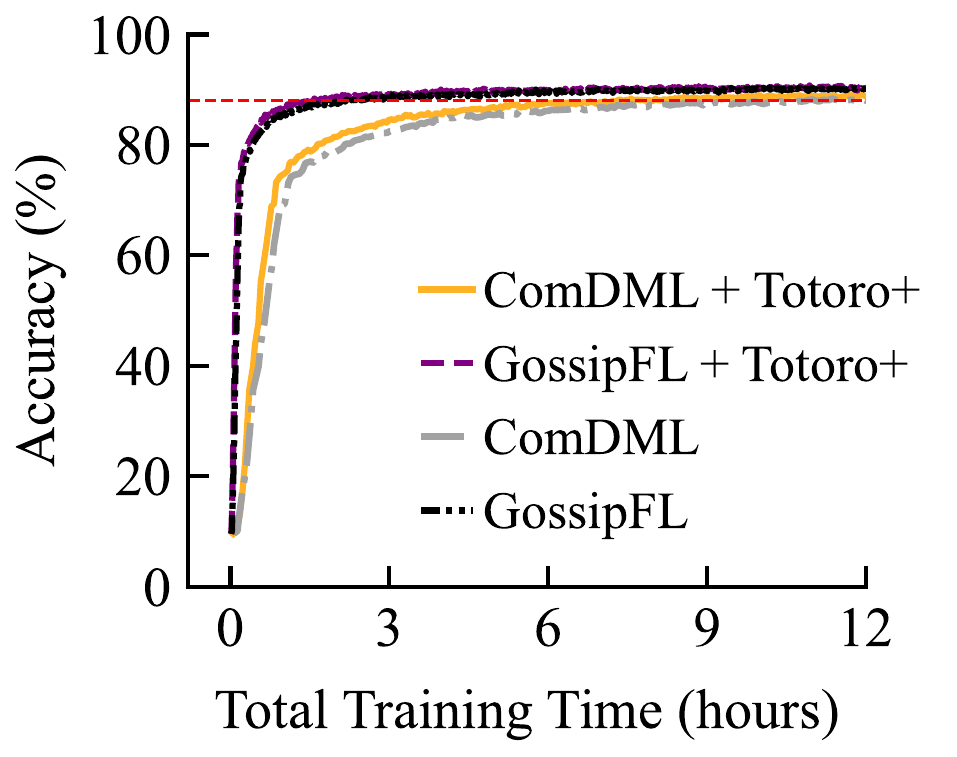}
  \label{fig:defl_femnist_r1_c6}}
 \hfill
   \subfloat[CIFAR-10~\cite{cifar10} with MobileViT~\cite{mobilevit}]{\includegraphics[width=0.24\textwidth]{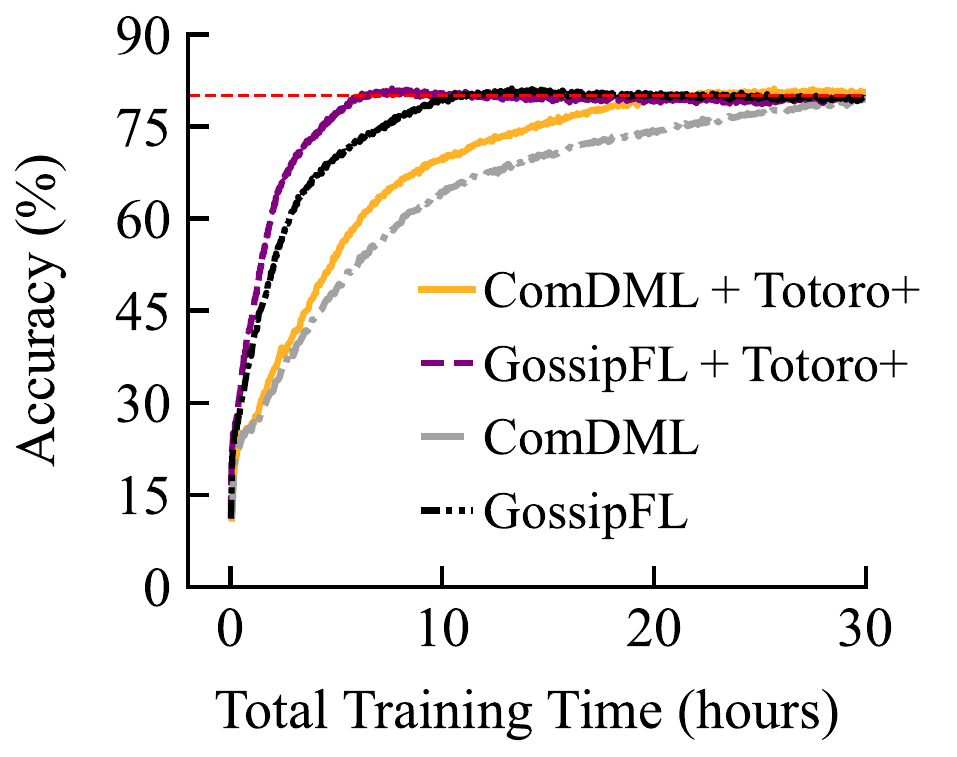}
   \label{fig:defl_cifar10_r1_c6}}
 \hfill
     \subfloat[MASSIVE~\cite{massive_dataset} with ALBERT~\cite{albert}]{\includegraphics[width=0.24\textwidth]{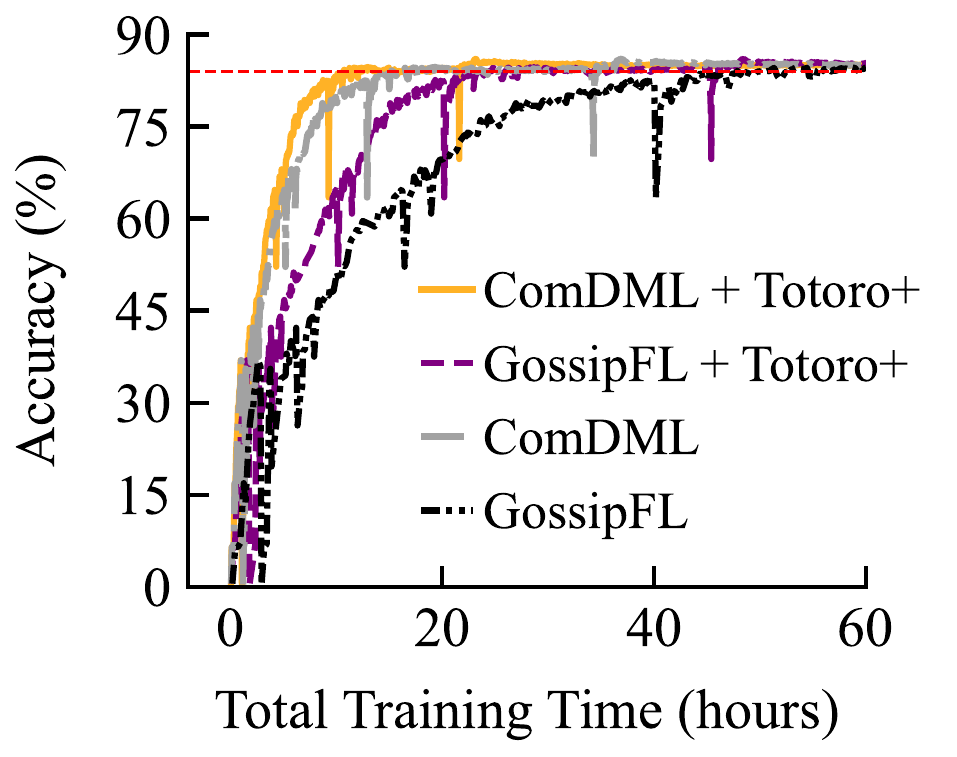}
   \label{fig:defl_massive_r1_c6}}
     \vspace{0.1in}
  \caption{Time-to-Accuracy performance. Decentralized federated learning frameworks running atop Totoro$^+$ (ComDML + Totoro$^+$ and GossipFL + Totoro$^+$) achieve faster time-to-accuracy convergence.}
  \label{fig:t2a_defl}
\end{figure*}

In the DHT-based P2P overlay, the node mapped to two P2P nodes is assigned two NodeIds. Although these two NodeIDs may not be numerically close, the proximity metrics (e.g., hop counts, RTT, cross-site link congestion levels) in the publish/subscribe-based forest
abstraction indicate that these two P2P nodes are physically close. Moreover, since the masters are distributed evenly across the overlay due to the hash function, the nodes mapped to two P2P nodes have more chances of being masters than those mapped to a P2P node. For failure recovery, since each P2P node operates individually and autonomously, the nodes mapped to multiple P2P nodes follow the same way to recover from failures. Lastly, since P2P nodes in the overlay represent \textit{resources}, a physical node mapped to several P2P nodes can manage its resources by shutting down or spawning P2P nodes at an application level. Similarly, an FL application can follow the publish/subscribe-based forest
abstraction to select P2P nodes that meet the resource requirements to run the tasks. 

To evaluate the feasibility of logical P2P nodes, we performed 10,000 runs of a simple CNN training task~\cite{fedavg} on a heterogeneous Amazon EC2 cluster. The cluster consists of 25 nodes, including 5 \texttt{t2.small} instances with 1 vCPU and 2 GB RAM, 5 \texttt{t2.medium} instances with 2 vCPUs and 4 GB RAM, 5 \texttt{t2.large} instances with 2 vCPUs and 8 GB RAM, 5 \texttt{t2.xlarge} instances with 4 vCPUs and 16 GB RAM, and 5 \texttt{t2.2xlarge} instances with 8 vCPUs and 32 GB RAM. For clarity of comparison, we denote Totoro$^+$ as the approach that employs logical P2P nodes, while Baseline represents the approach that “does not” use logical P2P nodes. In Totoro$^+$, different numbers of logical P2P nodes are mapped onto a single physical node according to its resource capacity: one logical node for \texttt{t2.small}, two \texttt{for t2.medium}, four for \texttt{t2.large}, eight for \texttt{t2.xlarge}, and sixteen for \texttt{t2.2xlarge}. In contrast, the Baseline treats each physical node as a single logical node. To study the effectiveness of logical P2P nodes on computation and communication overhead, we do not use the proposed game-theory algorithm to determine routing paths.

\noindentparagraph{Computation overhead.} We first compare the computation overhead between Totoro$^+$ and the Baseline. As shown in Figure~\ref{fig:number_masters_r2_r4}, Totoro$^+$ assigns a larger number of masters to physical nodes with higher capacity (i.e., \texttt{t2.2xlarge} instances, which have 8 vCPUs and 32 GB RAM), while nodes with fewer resources are assigned only a small number of masters. In contrast, the Baseline distributes masters evenly across all physical nodes. We further compare the average CPU utilization, as illustrated in Figure~\ref{fig:cpu_utilizations_r2_r4}. It can be clearly observed that Totoro$^+$ achieves significantly higher CPU utilization by placing more masters on high-capacity (40U) physical nodes. By contrast, the Baseline fails to fully exploit resource-rich nodes due to its uniform allocation strategy. These two figures demonstrate that the logical P2P node mechanism in Totoro$^+$ effectively adapts to node heterogeneity when handling computation overhead.

\noindentparagraph{Communication overhead.}
We next compare the communication overhead of Totoro$^+$ and the Baseline. Since communication between logical P2P nodes can occur either over intra-node edges or inter-node edges, depending on whether the communicating nodes reside on the same physical node, we first examine the distribution of intra-node and inter-node edges for both approaches (Figure~\ref{fig:number_edges_r2_r4}). For inter-node edges, we count the edges at the sender side. We observe that Totoro$^+$ assigns a larger number of both inter-node and intra-node edges to resource-rich nodes, whereas the Baseline distributes communication more evenly, with a stronger concentration on nodes with 6U, 10U, and 20U resources. Figure~\ref{fig:communication_times_r2_r4} further shows the communication time. As expected, inter-node edge communication incurs longer latency than intra-node edge communication due to cross–physical-node data transfers. However, Totoro$^+$ completes communication in less time overall, as it places a larger fraction of communication workloads on resource-rich nodes. In contrast, the Baseline distributes communication uniformly without considering node heterogeneity, resulting in less efficient utilization of available resources. Together, these results demonstrate that the logical P2P node mechanism in Totoro$^+$ effectively distributes communication overhead according to node heterogeneity.


\section{Totoro$^+$ v.s. Decentralized FL Systems}
\label{app:decsys_comparison}

In this section, we discuss the key differences between Totoro$^+$ and existing decentralized federated learning (DFL) systems, focusing on both system-level design and path orchestration strategies.

\noindentparagraph{System-level differences.}
Although decentralized federated learning (DFL) systems remove the need for centralized model aggregation and parameter servers, most existing DFL systems, still rely on a logically centralized control plane for resource management and application coordination~\cite{comdml, gossipfl, blockdfl, fsreal}.
Specifically, a central management interface is typically responsible for assigning devices to different FL applications. For example, consider a deployment scenario with three concurrent FL applications and $1,200$ available clients. The control plane partitions the client pool into three disjoint overlays, each dedicated to a single application. Clients within an overlay collaborate exclusively with peers assigned to the same task and execute the DFL protocol, i.e., performing local training and decentralized aggregation within that isolated topology. No cross-overlay interaction is allowed~\cite{comdml, gossipfl}.
In contrast, Totoro$^+$ eliminates the centralized control interface entirely and adopts a fully peer-to-peer (P2P) architecture for distributed management and coordination. Device participation, topology formation, and task execution are handled in a decentralized manner without relying on a global orchestrator. By removing the control-plane bottleneck, Totoro$^+$ enables improved scalability and more elastic adaptation as the number of clients and concurrent applications increases.

\noindentparagraph{Path planning models.}
The path planning model of Totoro$^+$ is grounded in a game-theoretic formulation that explicitly accounts for competition among nodes over shared network resources. In particular, Totoro$^+$ does not assume prior knowledge of bandwidth availability or link quality. Instead, routing decisions are made under incomplete information and are progressively refined based on observed transmission outcomes (bandit feedback).
In contrast, most existing decentralized federated learning (DFL) frameworks assume the existence of a logically centralized entity that can obtain real-time transmission statistics between nodes, such as available bandwidth and network congestion levels~\cite{comdml,gossipfl,fsreal}. Under this assumption, path or pairing decisions are typically derived from global network state information. To evaluate the impact of Totoro$^+$'s path planning model, we integrate it into both ComDML~\cite{comdml} and GossipFL~\cite{gossipfl} and compare time-to-accuracy performance with and without Totoro$^+$. We isolate the routing component while keeping the original training logic unchanged. This allows us to quantify whether replacing the baseline path planning model with Totoro$^+$, which operates without prior knowledge of bandwidth availability or link quality, improves time-to-accuracy performance.

We conduct experiments on four representative datasets: Google Speech~\cite{google_speech_dataset} for speech recognition, FEMNIST~\cite{femnist_dataset} for grayscale image classification, CIFAR-10~\cite{cifar10} for color image classification, and MASSIVE~\cite{massive_dataset} for intent classification. Each dataset is paired with a representative model architecture: ResNet-18~\cite{residualnet} for Google Speech, ShuffleNet~V2~\cite{shufflenet_v2} for FEMNIST, MobileViT~\cite{mobilevit} for CIFAR-10, and ALBERT~\cite{albert} for MASSIVE. The corresponding results are presented in Figures~\ref{fig:defl_speech_r1_c6}, \ref{fig:defl_femnist_r1_c6}, \ref{fig:defl_cifar10_r1_c6}, \ref{fig:defl_massive_r1_c6}.
The results show that, across both CNN-based applications (ResNet-18 in Figure~\ref{fig:defl_speech_r1_c6} and ShuffleNet~V2 in Figure~\ref{fig:defl_femnist_r1_c6}) and transformer-based models (MobileViT in Figure~\ref{fig:defl_cifar10_r1_c6} and ALBERT in Figure~\ref{fig:defl_massive_r1_c6}), DFL systems running atop Totoro$^+$ consistently reach the target accuracy (indicated by the red dashed line) within a shorter training time. This improvement stems from Totoro$^+$'s path planning model, which explicitly accounts for competition among nodes for shared transmission resources. Moreover, Totoro$^+$ operates without prior knowledge of bandwidth availability or link quality and instead adopts a bandit-based approach to progressively identify effective relay nodes for communication.
In contrast, existing frameworks such as ComDML and GossipFL rely heavily on accurate and up-to-date transmission information when making routing decisions. If the collected information is inaccurate or becomes stale, these systems may make suboptimal routing or pairing decisions, which can lead to slower time-to-accuracy performance.
}
\section{Discussions and Future Work}
\label{sec:discussions}

\subsection{Security of the advertise-discover tree.}

The advertise-discover (AD) tree enables edge nodes in the multi-ring for application advertisement and discovery. However, some adversarial edge nodes likely advertise malicious application information to other nodes. Also, adversarial edge nodes may discover and perform malicious attacks on specific FL applications. An efficient method is to rely on some trusted authorities, e.g., ICANN~\cite{icann} or a certification authority like Verisign~\cite{verisign}. The certification authority assigns a random NodeId to the new node and signs a NodeId certificate that binds the NodeID.

Suppose a node wants to locate an FL application. In that case, it first uses the \emph{AppId=hash(``advertise-discover application'')} as the key to route a \texttt{JOIN} message that contains its NodeId certificate. The internal node in the advertise-discover tree checks with the certification authority to verify the NodeId certificate after receiving the \texttt{JOIN} message. If the certificate is authentic, the new node joins the advertise-discover tree and gets the AppIds of FL applications running over the DHT-based P2P overlay. Next, joining the dataflow tree of the FL application follows the same way.

Likewise, if a node wants to advertise an FL application, it first contacts the certification authority to obtain an AppId certificate. The nodes that want to run the FL application first route \texttt{JOIN} messages using the AppId as the key, then receive the AppId certificate, check with the certification authority to verify the AppId certificate, join the dataflow tree, and run the FL application.

\subsection{Sending packets to multiple nodes simultaneously.}

\begin{algorithm}[t]
    \caption{Multicast-enabled distributed hop-by-hop routing algorithm}
    \label{pseudo:our_algorithm_multicast}
    \begin{algorithmic}[1]
        \INPUT{mixture weights $\alpha,\beta \in [0,1]$; initial policy $\pi^1_n$ for all $n\in N$.}
        \For{episode $k=1, 2, \cdots$}
            \For{packet $t=1,\cdots, \tau$}
                \State \parbox[t]{200pt}{Each node $n$ in parallel selects multiple hops $\textbf{p}^{k,t}_n\sim \pi^k_n$ to send a packet simultaneously, observes reward $r^{k,t}_n$. \strut}
                \label{pseudo:multicast_collect_info}
            \EndFor
            \For{node $n=1,\cdots, N$} \textbf{in parallel}
                \State \parbox[t]{200pt}{$\rho_n^k \leftarrow \argmin_{\lambda\in \Delta(\textbf{P}_n)} \text{det}(M(\lambda))$.\strut}
                \label{pseudo:multicast_d_optimal}

                \State \parbox[t]{200pt}{$\widehat{\nabla}^k_n\Phi(\textbf{p}) \leftarrow \frac{1}{\tau}\sum_{t=1}^\tau \psi(\textbf{p})^\top M(\pi^k_n)^{-1}\psi(\textbf{p}^{k,t}_n)r^{k,t}_n$, for all $\textbf{p}\in \textbf{P}_n$.\strut}
                \label{pseudo:multicast_gradient_estimator}
                
                \State \parbox[t]{200pt}{$\Tilde{\pi}^{k+1}_n\leftarrow \argmax_{\lambda\in \Delta(P_n)}\langle \lambda, \widehat{\nabla}^k_n\Phi \rangle$.\strut}
                \label{pseudo:multicast_find_optimal_policy}
                            
                \State \parbox[t]{200pt}{$\pi^{k+1}_n \leftarrow \underbrace{\alpha\big[\pi^k_n+\beta(\Tilde{\pi}^{k+1}_n-\pi^k_n)\big]}_{\text{Frank-Wolfe update}}$ 
                $+ \underbrace{(1-\alpha)\rho^k_n}_{\text{Exploration}}$.\strut}
                \label{pseudo:multicast_update_policy}
            \EndFor
        \EndFor
    \end{algorithmic}
\end{algorithm}

Currently, each node follows Algorithm~\ref{pseudo:our_algorithm} to forward packets to a ``single node'', and we have shown that Algorithm~\ref{pseudo:our_algorithm} achieves an $\epsilon$-approximate Nash equilibrium (Theorem {\color{blue}1}). However, multicast has been a communication paradigm that improves communication efficiency. Fortunately, we can extend Algorithm~\ref{pseudo:our_algorithm} to support multicast. 

We present the multicast-enabled distributed hop-by-hop routing algorithm in Algorithm~\ref{pseudo:our_algorithm_multicast}. We follow the setup in Appendix~\ref{sec:numerical_example} to illustrate the differences.
In line~\ref{pseudo:multicast_collect_info}, each node follows the initial policy to select multiple hops, e.g., $m_1$ and $m_2$, to send packets simultaneously and receive rewards. The range of rewards may become $[0, F]$, where $F$ denotes the maximum number of hops selected.
Then, since each node can select multiple hops, $\textbf{P}_{l_1}$ may contain $\{m_1, m_2, \{m_1, m_2\}\}$, and similarly each policy in $\Delta(\textbf{P}_{l_1})$ becomes a 3-dimensional vector, e.g., $[0.6, 0.1, 0.3]$, where it represents that $60\%$ of chance selecting $m_1$, $10\%$ of chance selecting $m_2$, and $30\%$ of chance selecting both $m_1$ and $m_2$. Subsequently, each node follows line~\ref{pseudo:multicast_gradient_estimator} to estimate the gradients of each possible action. Overall, Algorithm~\ref{pseudo:our_algorithm_multicast} can fully support multicast in a distributed manner, but we leave the upper bound of Nash regret of Algorithm~\ref{pseudo:our_algorithm_multicast} to future work.
{\color{dv}
\subsection{Real-world deployment of Totoro$^+$}
While many existing FL systems are operated by a single organization, there also exist more open and community-driven platforms (e.g., crowdsourcing-based FL platforms~\cite{kang2023incentive, tifedcrowd, p3m}), where application owners publish training requests according to their data or model requirements, and participating devices are recruited through incentive mechanisms designed to encourage contribution. 
Totoro$^+$ is particularly well-suited for these scenarios. First, its high scalability enables multiple FL applications and a large number of heterogeneous devices to execute training workflows in parallel in a distributed manner, without relying on a centralized server for synchronization or scheduling. Second, its adaptivity limits the impact of highly dynamic device participation on the overall FL process, which is common in open platforms or edge computing. Third, the customizability of Totoro$^+$ allows different FL applications to flexibly incorporate diverse mechanisms (e.g., incentive strategies and client selection policies) without embedding all such logic into a centralized server.
We also note that many existing P2P systems, including prior P2P-based FL designs (e.g.,~\cite{blockdfl, comdml, gossipfl, fsreal}), make similar assumptions regarding cooperative or incentive-compatible environments. Therefore, we believe it is well positioned to demonstrate the scalability, adaptivity, and customization advantages of Totoro$^+$ in environments where P2P functionality is permitted and open development is feasible.

\subsection{Adaptivity to non-IID settings}
Given the high heterogeneity of edge environments (e.g., diverse devices, data ownership, and application requirements), the customization capability of Totoro$^+$ enables application owners to deploy FL workflows with different mechanisms for handling non-IID data distributions. For example, before publishing an application, developers can specify customized algorithms for data filtering~\cite{safe, feddiv, fedfixer} or client selection~\cite{tifl,oort,client_selection_fl}, which are executed by the master assigned by Totoro$^+$. During application subscription, the master can selectively admit clients based on statistical properties of their local data (e.g., Dirichlet- or Gaussian-based distributions), select clients that satisfy specific data requirements, or employ personalization schemes that allow subscribed devices to update local model parameters differently under non-IID conditions. Through these extensions, Totoro$^+$ provides system-level support for a wide range of non-IID-aware FL algorithms.
}

\end{document}